\documentclass[aos,preprint]{imsart}

\usepackage{hyperref, amsmath, amsthm,bm, amsfonts}
\usepackage[numbers]{natbib}
\usepackage{graphicx,subcaption}
\usepackage{cleveref}
\crefname{equation}{}{}
\usepackage{multirow,tabularx}
\usepackage{color}
\usepackage{enumerate}
\usepackage[colorlinks,citecolor=blue,urlcolor=blue]{hyperref}
\RequirePackage[OT1]{fontenc}

\usepackage{xr}
\usepackage{array}

\newcolumntype{$}{>{\global\let\currentrowstyle\relax}}
\newcolumntype{^}{>{\currentrowstyle}}
\newcommand{\rowstyle}[1]{\gdef\currentrowstyle{#1}%
  #1\ignorespaces
}

\newtheorem{theorem}{Theorem}[section]
\newtheorem{proposition}{Proposition}[section]
\newtheorem{lemma}{Lemma}[section]

\theoremstyle{definition}
\newtheorem{assumption}{Assumption}
\crefname{assumption}{assumption}{assumptions}

\theoremstyle{remark}
\newtheorem{remark}{Remark}[section]

\numberwithin{equation}{section}
\renewcommand{\theequation}{\thesection.\arabic{equation}}

\newcommand{\Alpha}{\mathrm{A}}
\newcommand{\Beta}{\mathrm{B}}

\newcommand{\NC}{\mathrm{NC}}
\newcommand{\calC}{\mathcal{C}}
\newcommand{\RR}{\mathrm{RR}}
\newcommand{\Prob}{\mathrm{P}}

\newcommand{\real}{\mathbb{R}}

\newcommand\independent{\protect\mathpalette{\protect\independenT}{\perp}}
\def\independenT#1#2{\mathrel{\rlap{$#1#2$}\mkern2mu{#1#2}}}

\begin{document}

\begin{frontmatter}

  \title{Confounder Adjustment in
    Multiple Hypothesis Testing
  }
  \runtitle{Confounder Adjustment}
  \thankstext{T1}{The
    authors thank Bhaswar Bhattacharya, Murat Erdogdu, Jian Li,
    Weijie Su and Yunting Sun for helpful discussion.}

  \begin{aug}
    \author{\fnms{Jingshu}  \snm{Wang}\thanksref{t2}\ead[label=e1]{jingshuw@stanford.edu}},
    \author{\fnms{Qingyuan} \snm{Zhao}\corref{}\thanksref{t2}\ead[label=e2]{qyzhao@stanford.edu}},
    \author{\fnms{Trevor}  \snm{Hastie}\thanksref{t3}\ead[label=e3]{hastie@stanford.edu}},
    \author{\fnms{Art B.}  \snm{Owen}\thanksref{t4} \ead[label=e4]{owen@stanford.edu}}

    \thankstext{t2}{The first two authors contributed equally to this
      paper.}
    \thankstext{t3}{Supported in part by NSF Grant DMS-1407548 and NIH Grant 5R01-EB-001988-21.}
    \thankstext{t4}{Supported in part by NSF Grant DMS-1521145.}

    \runauthor{J. Wang et al.}

    \affiliation{Stanford University}

    \address{Department of Statistics\\
      Stanford University\\
      390 Serra Mall\\
      Stanford, California 94305\\
      USA\\
      \printead{e1}\\
      \printead{e2}\\
      \printead{e3}\\
      \printead{e4}}

  \end{aug}

  \begin{abstract}
    We consider large-scale studies in which thousands of significance
tests are performed simultaneously. In some of these studies, the
multiple testing procedure can be severely biased by latent confounding factors such as batch
effects and unmeasured covariates that correlate with both primary
variable(s) of interest (e.g.\ treatment variable, phenotype) and the
outcome. Over the past decade, many statistical methods have been proposed
to adjust for the confounders in hypothesis testing.
We unify these methods in
the same framework, generalize them to include multiple primary variables and
multiple nuisance variables, and analyze their statistical properties. In
particular, we provide theoretical guarantees for RUV-4
\citep{gagnon2013} and LEAPP \citep{sun2012}, which correspond to two
different identification conditions in the framework:
the first requires a set of ``negative controls'' that are known
\textit{a priori} to follow the null distribution; the second requires
the true non-nulls to be sparse.
Two different estimators which are 
based on RUV-4 and
LEAPP are then applied to these two scenarios.
We show that if the confounding
factors are strong, the resulting estimators can be asymptotically as
powerful as the oracle estimator which observes the latent confounding
factors.
For hypothesis testing, we show the asymptotic $z$-tests based on the
estimators can control the type I error. Numerical experiments show
that the false discovery rate is also controlled by the
Benjamini-Hochberg procedure when the sample size is reasonably large.

  \end{abstract}



\end{frontmatter}

\section{Introduction}
\label{sec:introduction}


Multiple hypothesis testing has become
an important statistical problem for many scientific fields, where
tens of thousands of tests are typically performed simultaneously.
Traditionally the tests are assumed to be
independent of each other, so the
false discovery rate (FDR) can be easily controlled by e.g.,
the Benjamini-Hochberg procedure \citep{benjamini1995}. Recent years have witnessed an extensive investigation of
multiple hypothesis testing under dependence, ranging from
permutation tests \citep{korn2004,tusher2001},
positive dependence
\citep{benjamini2001}, weak dependence \citep{clarke2009,storey2004}, accuracy calculation under
dependence \citep{efron2007,owen2005}
 to mixture models \citep{efron2010,sun2009} and latent factor models \citep{fan2012,fan2013,lan2014}. Many of these works provide theoretical guarantees for
FDR control under the assumption that
the individual test statistics are valid and may even be
correlated.

In this paper, we investigate a more challenging setting.
The test statistics may be correlated with each other due to
latent factors and those latent factors may also be correlated
with the variable of interest.
As a result, the test statistics are not only correlated but are also confounded.
We use the phrase ``confounding'' to emphasize
that these latent factors can significantly bias the individual
p-values, therefore this problem is fundamentally different from the
literature in the previous paragraph and poses an immediate threat to
the reproducibility of the discoveries. Many confounder adjustment
methods have already been proposed for multiple testing over the last
decade \citep{gagnon2013,leek2008,price2006,sun2012}. Our goal is to
unify these methods in the same framework and study their statistical properties.

\vspace{3mm}

\textit{The confounding problem.~~} We start with three real data examples to illustrate the confounding problem. The first microarray data (\Cref{fig:COPD-naive}) is used by \citet{singh2011} to identify candidate genes
associated with a chronic lung disease called emphysema. The second (\Cref{fig:gender-naive,fig:gender-batch}) and
third (\Cref{fig:alzheimers-batch}) data are used by \citet{gagnon2013} to
study the performance of various confounder adjustment methods. For each dataset, we plot the histogram of
t-statistics of a simple linear model that regresses the gene
expression on the variable of interest (disease status for the first
and gender for the second and third datasets). These statistics are
commonly used in genome-wide association studies (GWAS) to find
potentially interesting genes. See
\Cref{sec:datasets} for more detail of these datasets.

The histograms of t-statistics in \Cref{fig:confounding-problem}
clearly depart from the approximate theoretical null distribution
$\mathrm{N}(0,1)$. The bulk of the test statistics can be skewed
(\Cref{fig:COPD-naive,fig:gender-naive}), overdispersed
(\Cref{fig:COPD-naive}), underdispersed
(\Cref{fig:gender-naive,fig:gender-batch}), or noncentered
(\Cref{fig:alzheimers-batch}). In these cases, neither the theoretical null
$\mathrm{N}(0,1)$, nor even the empirical null as shown in the histograms, look
appropriate for measuring significance.
\Citet{schwartzman2010comment} proved that
a largely overdispersed histogram like \Cref{fig:COPD-naive} cannot be explained by correlation alone,
and is possibly due to the presence of confounding factors.
For a sneak preview of the
confounder adjustment, the reader can find the histograms after our confounder
adjustment in \Cref{fig:output} at the end of this paper.
The p-values of our test of confounding
(\Cref{sec:test-confound}) in \Cref{tab:output} indicate that all the three datasets
suffer from confounding latent factors.

Other common sources of confounding in gene expression profiling include
systematic ancestry differences \citep{price2006}, environmental changes
\citep{fare2003,gasch2000} and surgical manipulation
\citep{lin2006}. See \citet{lazar2013} for a survey.
In many studies, especially for observational clinical research and human expression data,
the latent factors, either genetic or technical, are confounded with primary variables
of interest due to the observational nature of the studies and heterogeneity of samples
\citep{ransohoff2005,rhodes2005}.
Similar
confounding problems also occur in other high-dimensional datasets
such as brain imaging \citep{schwartzman2008} and metabonomics \citep{craig2006}.

\begin{figure}[hp]
  \begin{subfigure}[b]{0.45\textwidth}
  \includegraphics[width = \textwidth]{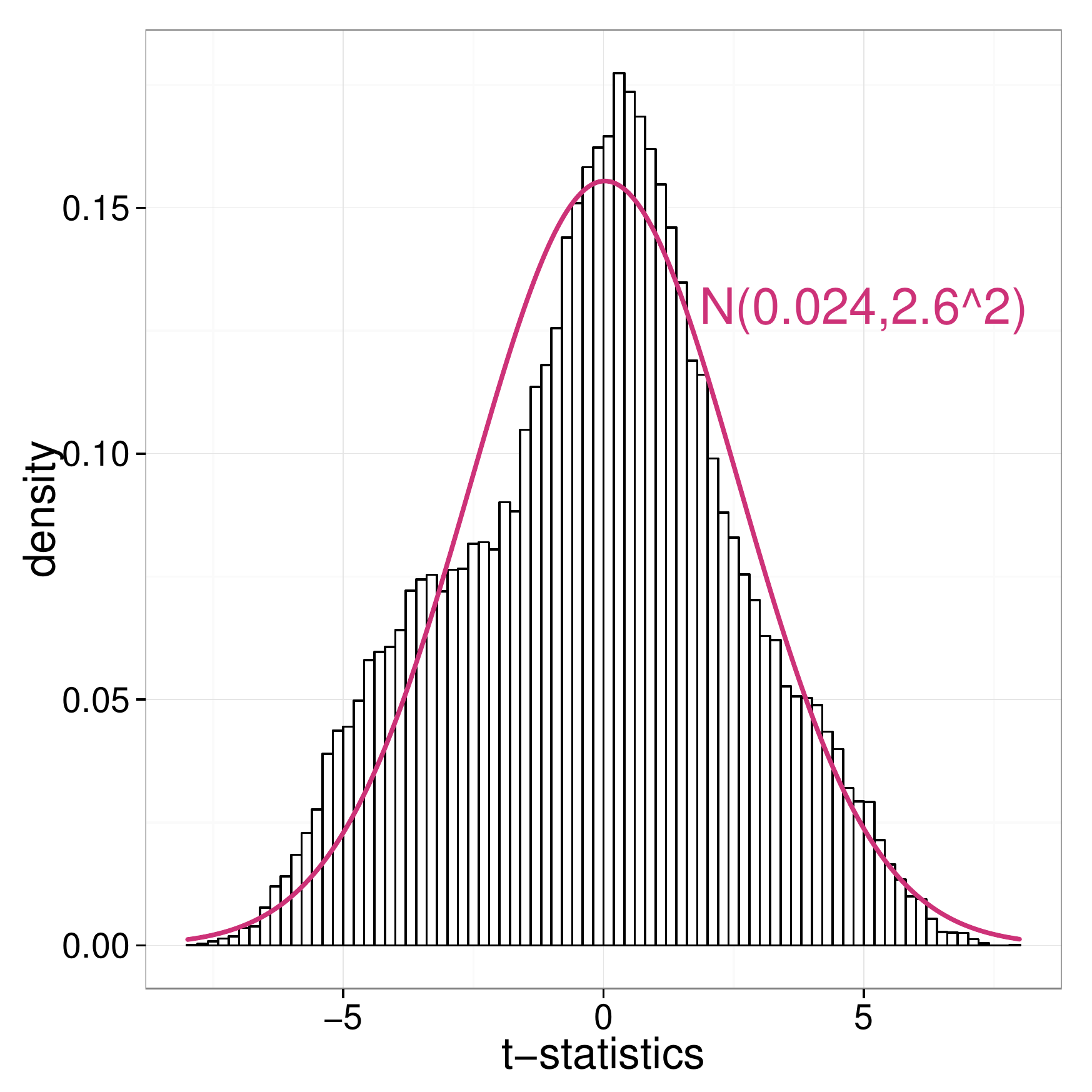}
  \caption{\small Dataset 1.}
  \label{fig:COPD-naive}
  \end{subfigure}
  ~
  \begin{subfigure}[b]{0.45\textwidth}
 \includegraphics[width = \textwidth]{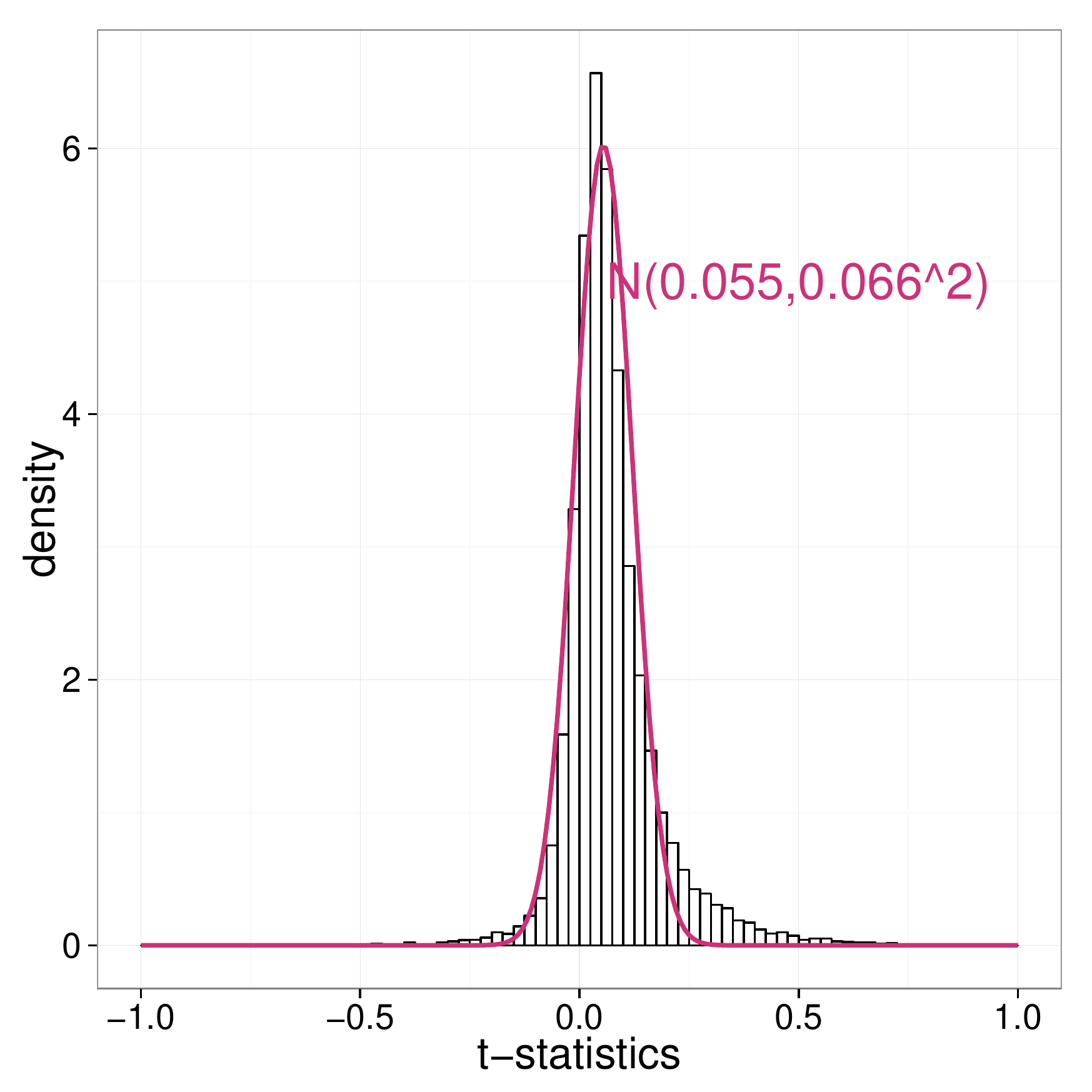}
  \caption{\small Dataset 2.}
  \label{fig:gender-naive}
  \end{subfigure}
  \\
  \centering
  \begin{subfigure}[b]{0.45\textwidth}
  \includegraphics[width = \textwidth]{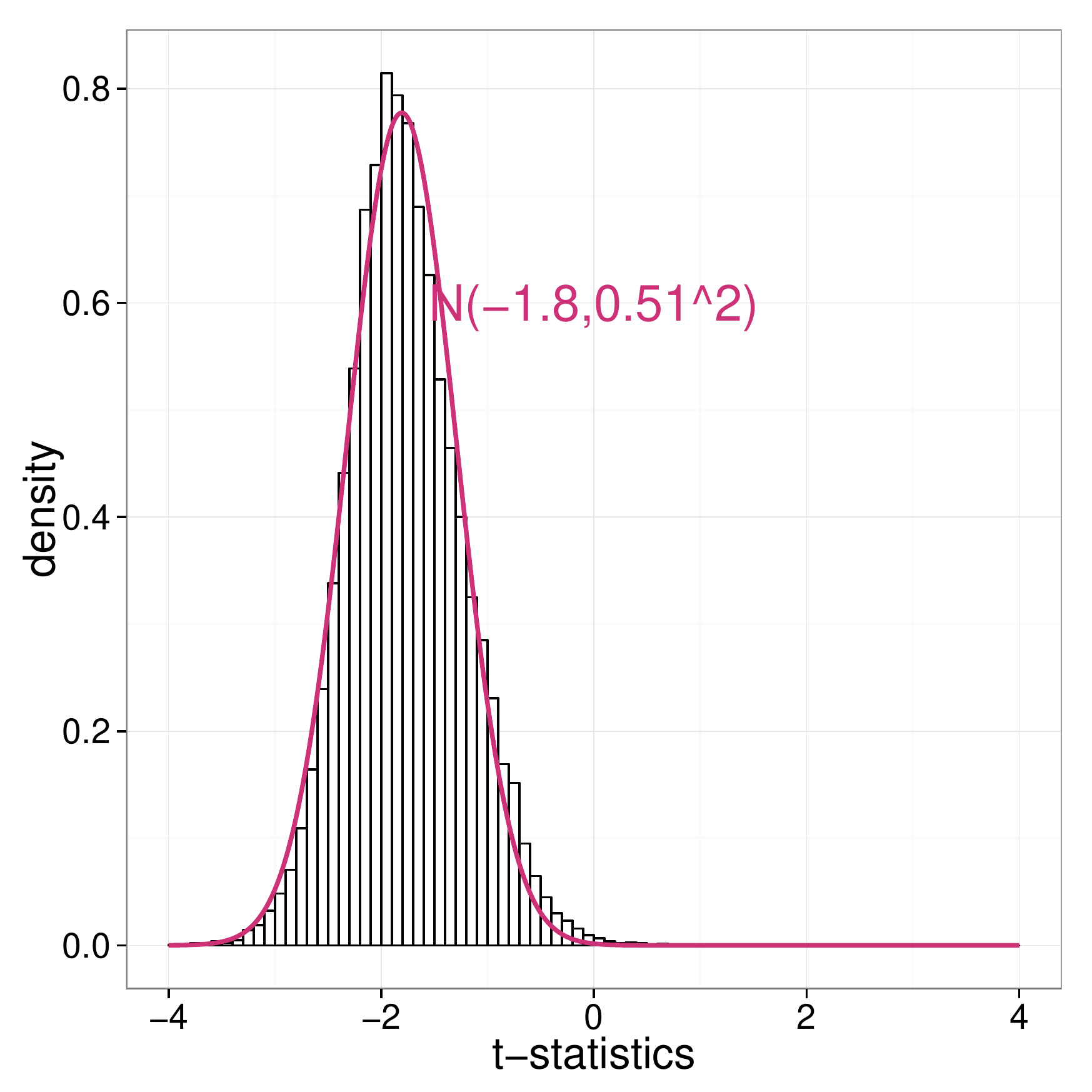}
  \caption{\small Dataset 3.}
  \label{fig:alzheimers-batch}
  \end{subfigure}
  ~
  \begin{subfigure}[b]{0.45\textwidth}
  \includegraphics[width = \textwidth]{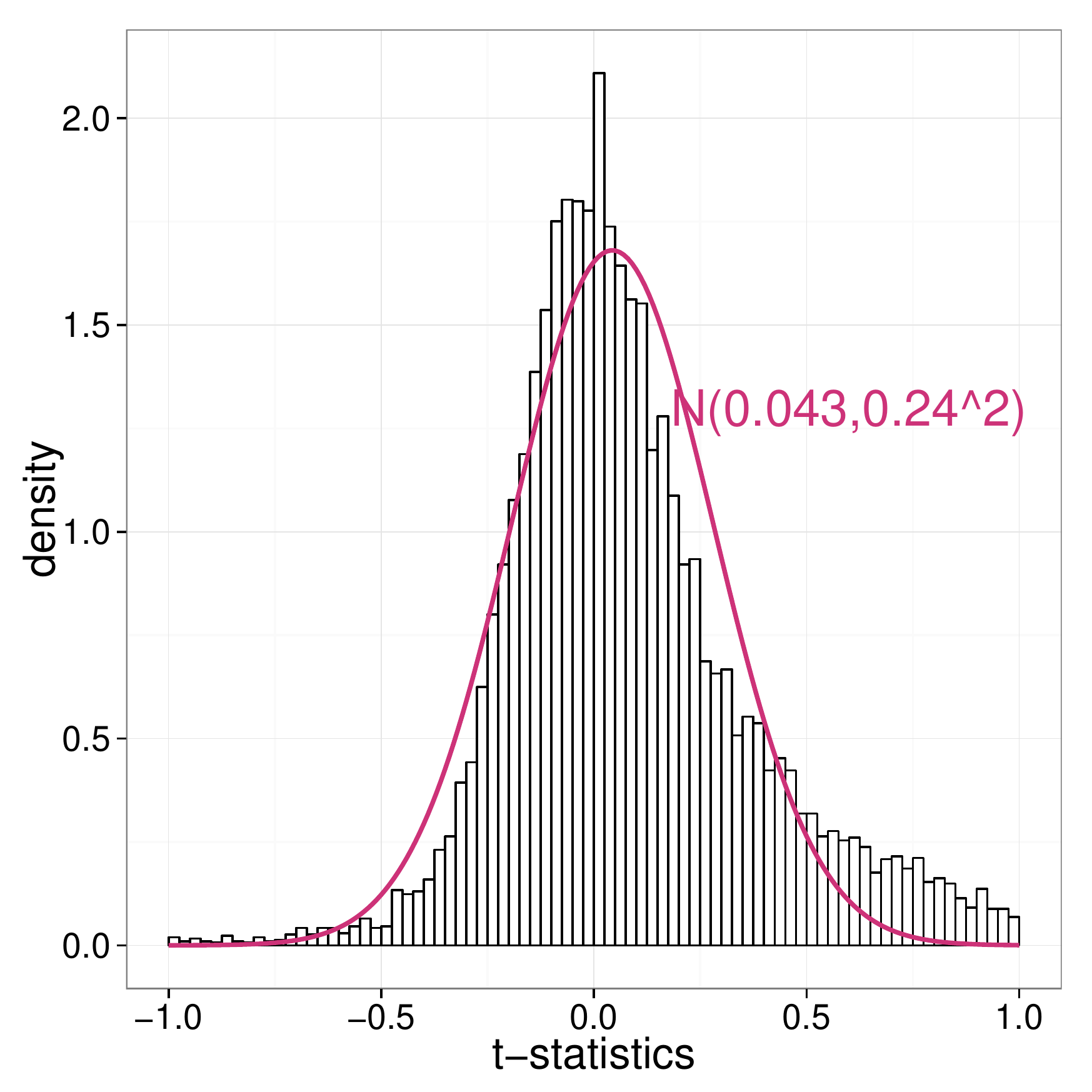}
  \caption{\small Dataset 2 (batch correction).}
  \label{fig:gender-batch}
  \end{subfigure}
  \caption{ \small
  Dataset 1 is the COPD dataset \citep{singh2011}. Dataset 2 and 3 are from
  \citet{gagnon2013}.
  Histograms of regression t-statistics in three microarray studies show clear departure from the
    theoretical null distribution $\mathrm{N}(0,1)$. The mean and
    standard deviation of the normal approximation are obtained from
    the median and median absolute deviation of the statistics. See
    \Cref{sec:gender-study} for the empirical distributions after confounder
    adjustment.}
  \label{fig:confounding-problem}
\end{figure}


\vspace{3mm}

\textit{Previous methods.~~} As early as \citet{alter2000},
principal component analysis has been suggested to estimate the confounding
factors. This approach can work reasonably well if the confounders
clearly stand out. For example, in population
genetics, \citet{price2006} proposed a procedure called EIGENSTRAT that
removes the largest few principal
components from their SNP genotype data, claiming they closely resemble the ancestry
difference. In gene expression data, however, it is often unrealistic to
assume they always represent the confounding factors. The largest
principal component may also correlate with the
primary effects of interest. Therefore, directly removing them can result in loss of statistical power.

More recently, an emerging literature considers the confounding problem
in similar statistical settings and a variety of methods have been
proposed for confounder adjustment
\citep{friguet2009,gagnon2013,gagnon2012,leek2007,leek2008,sun2012}. 
These statistical methods are shown to
work better than the EIGENSTRAT procedure for gene expression data. However,
little is known about their theoretical properties. Indeed, the
authors did not focus on model identifiability and rely on impressive heuristic calculations
to derive their estimators.
In this paper, we address the identifiability problem, rederive the
estimators in \citep{gagnon2013,sun2012} in a more principled way and provide
theoretical guarantees for them.

Before describing the modeling framework, we want to clarify our
terminology. The confounding factors or confounders considered in the
present paper are referred to by different names in the literature, such as
``surrogate variables'' \citep{leek2007}, ``latent factors'' \citep{friguet2009}, ``batch effects''
\citep{leek2010}, ``unwanted variation'' \citep{gagnon2012} and
``latent effects'' \citep{sun2012}. We believe they are all describing
the same phenomenon: that there exist some unobserved variables that
correlate with both the primary variable(s) of interest and the outcome
variables (e.g.\ gene expression). This problem is generally known
as confounding \citep{fisher1935,kish1959}. A famous example is Simpson's paradox.
The term  ``confounding'' has multiple
meanings in the literature.
We use the meaning from \cite{greenland1999}:
``a mixing of effects of extraneous factors (called confounders) with the effect of
interest''.

\vspace{3mm}

\textit{Statistical model of confounding.~~} Most of the confounder
adjustment methods mentioned above are built around the following model
\begin{equation} \label{eq:previous-model}
\bm{Y} = \bm{X} \bm{\beta}^T + \bm{Z} \bm{\Gamma}^T + \bm{E}
\end{equation}
Here $\bm{Y}$ is a $n \times p$ observed matrix (e.g.\ gene expression); $\bm{X}$ is an $n \times 1$ observed primary variable of
interest (e.g.\ treatment-control, phenotype, health trait); $\bm{Z}$
is an $n \times r$ latent confounding factor matrix; $\bm{E}$ is
often assumed to be a Gaussian noise matrix. The $p \times 1$ vector
$\bm{\beta}$ contains the primary effects we want to estimate.

Model \eqref{eq:previous-model} is very general for multiple
testing dependence. \citet[Proposition 1]{leek2008} suggest that
multiple hypothesis tests based on linear regression can always be
represented by \eqref{eq:previous-model} using sufficiently many
factors. However, equation \eqref{eq:previous-model} itself is not
enough to model confounded tests.  To
elucidate the concept of confounding, we need to characterize the
relationship between the latent variables $\bm{Z}$ and the primary
variable $\bm{X}$. To be more specific, we assume the
regression of $\bm{Z}$ on $\bm{X}$ also follows a linear relationship
\begin{equation} \label{eq:linear-model-z-sec1}
\bm{Z} = \bm{X} \bm{\alpha}^T + \bm{W},
\end{equation}
where $\bm{W}$ is a $n \times r$ random noise matrix
independent of $\bm{X}$ and $\bm{E}$ and the $r \times 1$ vector
$\bm{\alpha}$ characterizes the extent of confounding in this
data. By plugging \eqref{eq:linear-model-z-sec1} in
\eqref{eq:previous-model}, the linear regression of $\bm{Y}$ on
$\bm{X}$ gives an unbiased estimate of the marginal effects
\begin{equation} \label{eq:marginal-effect}
\bm{\tau} = \bm{\beta} + \bm{\Gamma} \bm{\alpha}
\end{equation}
When $\bm{\alpha} \ne \bm{0}$, $\bm{\tau}$ is not the same as
$\bm{\beta}$ by \eqref{eq:marginal-effect}. In this case, the data
$(\bm{X},\bm{Y})$ are confounded by $\bm{Z}$. Since the
confounding factors $\bm{Z}$ are data artifacts in this model, the
statistical inference of $\bm{\beta}$ is much more interesting
than that of $\bm{\tau}$. See \Cref{sec:marginal-effects-vs} for more
discussion on the marginal and the direct effects.


Following LEAPP \citep{sun2012}, we use
a QR decomposition to decouple the estimation of  $\bm{\Gamma}$ from $\bm{\beta}$.
The inference procedure splits into the following two steps:

\begin{description}
\item[Step 1] By regressing out $\bm{X}$ in \cref{eq:previous-model},
  $\bm{\Gamma}$ is the loading matrix in a factor analysis model and can be efficiently estimated by maximum likelihood.
\item[Step 2] Equation \eqref{eq:marginal-effect} can be viewed as a
  linear regression of the marginal effects $\bm{\tau}$ on the factor
  loadings $\bm{\Gamma}$. To estimate $\bm{\alpha}$ and $\bm{\beta}$,
  we replace $\bm{\tau}$ by its observed value and $\bm{\Gamma}$ by
  its estimate in Step 1.
\end{description}

As mentioned before,
other existing confounder adjustment methods including SVA
\citep{leek2008} and RUV-4 \citep{gagnon2013} can be unified in this two-step statistical procedure.
See \Cref{sec:comp-with-exist} for a detailed discussion of these methods.

\vspace{3mm}

\textit{Contributions.~~} Our first contribution in \Cref{sec:model}
is to establish identifiability for the confounded multiple testing
model. In the first step of estimating
factor loadings $\bm{\Gamma}$, identifiability is well studied in
classical multivariate statistics. However, the second step of
estimating the effects $\bm{\beta}$ is not identifiable
  without additional constraints. We consider two different sufficient
  conditions for global identifiability. The
  first condition assumes the researcher has a ``negative control''
  variable set for which there should be no direct effect. This
  negative control set often serves as a quality control precaution in microarray studies \citep{gagnon2012}, but they can also be used to adjust for the confounding factors. The second identification condition assumes at
  least half of the true effects are zero, i.e., the true alternative
  hypotheses are sparse. These two identification conditions
  correspond to the approaches of RUV-4 \citep{gagnon2013} and LEAPP \citep{sun2012}, respectively.

Our second contribution in \Cref{sec:stat-infer} is to derive valid and efficient statistical methods under these identification
conditions in the second step. In order to estimate the effects, it is essential to estimate the coefficients $\bm{\alpha}$
relating  the primary variable to the confounders. Under the two
different identification conditions, we study two
different regression methods which are analytically tractable and
equally well performing alternatives to RUV-4 and LEAPP. For the
negative control (NC) scenario, $\hat{\bm{\alpha}}^{\mathrm{NC}}$ and
$\hat{\bm{\beta}}^{\mathrm{NC}}$ are obtained
by generalized least squares using the negative
controls. For the sparsity scenario, $\hat{\bm{\alpha}}^{\mathrm{RR}}$ and
$\hat{\bm{\beta}}^{\mathrm{RR}}$ are obtained by using a simpler
and more analytically tractable robust regression (RR) than the one used in LEAPP.

When the factors are strong (as large as the noise magnitude),
for both scenarios we find that the
resulting estimators of $\bm{\beta}$ are asymptotically as efficient
as the oracle estimator which is allowed to observe the confounding
factors.  It is surprising that no essential loss of efficiency is incurred
by searching for the confounding variables.
Our asymptotic  analysis relies on some recent theoretical results for factor analysis due to \citet{bai2012}. The
asymptotic regime we consider has both $n$, the number of
observations, and $p$, the number of outcome variables (e.g.\ genes),
going to infinity. The most important condition that we require for asymptotic
efficiency in
the negative control scenario is that the number of negative controls increases to infinity;
in the sparsity scenario, we need the $L_1$ norm of the effects to satisfy $\|\bm{\beta}\|_1 \sqrt{n} / p \to
0$. The fact that $p \gg n$ in many multiple hypothesis
testing problems plays an important role in these asymptotics.

Next in \Cref{sec:stat-infer}, we show that the asymptotic $z$-statistics based
on the efficient estimators of $\bm{\beta}$ can control the type I
error. This is not a trivial corollary from the asymptotic
distribution of the test statistics because the size of $\bm{\beta}$ is growing and the
$z$-statistics are weakly correlated. Proving FDR control is more
technically demanding and is beyond the scope of this paper. Instead,
we use numerical simulations to study the empirical performance
(including FDR) of our tests. We also give a significance test of
confounding (null hypothesis $\bm{\alpha} = \bm{0}$) in
\Cref{sec:stat-infer}. This test can help the experimenter to
determine if there is any hidden confounder in the design or the
experiment process.

In \Cref{sec:extens-mult-regr}, we generalize the
confounder adjustment model to include multiple primary variables and multiple nuisance covariates.
We show the statistical methods and theory for the single
primary variable regression problem \cref{eq:previous-model} can be
smoothly extended to the multiple regression problem.


\vspace{3mm}

\textit{Outline.~~} \Cref{sec:model} introduces the model and
describes the two identification
conditions. \Cref{sec:stat-infer} studies the statistical inference.
\Cref{sec:extens-mult-regr} extends our framework to a
linear model with multiple primary variables and multiple known controlling
covariates. 
\Cref{sec:discussions}
discusses our theoretical analysis in the context of previous literature, including
the existing procedures for debiasing the confounders and existing
theoretical results of multiple hypothesis testing under dependence
(but no confounding). \Cref{sec:numerical-experiments} studies the
empirical behavior of our estimators in simulations and real data
examples. Technical proofs of the results are provided in Supplement
\citep{wang2015supplement}.

To help the reader follow this paper and compare our methods and theory
with existing approaches, \Cref{tab:literature} summarizes some
related publications with more detailed discussion in
\Cref{sec:discussions}.

\setlength{\tabcolsep}{3pt}
\begin{table}[t]
  \footnotesize
  \begin{center}
  \begin{tabular}{c|c|c|}
    \cline{2-3}
    & \multicolumn{2}{c|}{Noise conditional on latent
      factors} \\ \cline{2-3}
    & Independent 
    & Correlated \\ \cline{1-3}
    \multicolumn{1}{|c|}{\begin{tabular}[x]{@{}c@{}} Positive or weak \\ dependence\end{tabular}} &
    \multicolumn{2}{c|}{
    \begin{tabular}[x]{@{}c@{}}\citet{benjamini2001}\\\citet{storey2004}
    \\ \citet{clarke2009}\end{tabular} }
      \\ \cline{1-3}
    \multicolumn{1}{|c|}{\begin{tabular}[x]{@{}c@{}} Unconfounding \\
        factors \end{tabular}} &
    \multicolumn{1}{c|}{\begin{tabular}[x]{@{}c@{}}         \citet*{friguet2009} \\
    \citet{desai2012cross} \\ \end{tabular}} &
    \multicolumn{1}{c|}{ \begin{tabular}[x]{@{}c@{}}\citet*{fan2012}\\\citet{lan2014}\\\textit{Discussed in \Cref{sec:latent-but-unconf,sec:marginal-effects-vs}}\end{tabular}
    } \\
    \cline{1-3}
    \multicolumn{1}{|c|}{\begin{tabular}[x]{@{}c@{}}Confounding \\
    factors \\ \end{tabular}}
    & \begin{tabular}[x]{@{}c@{}} \citet{leek2007,leek2008}\\
    \citet{gagnon2012} \\ 
    \citet*{sun2012} \\ \textit{Studied in
      \Cref{sec:model,sec:stat-infer,sec:extens-mult-regr}} \\ \textit{Discussed in \Cref{sec:comp-with-exist}}\end{tabular}
    & \begin{tabular}[x]{@{}c@{}} \textit{Discussed in
      \Cref{sec:inference-when-sigma}} \\ \textit{(future research)} \\ \end{tabular} \\ \cline{1-3}
  \end{tabular}
  \end{center}
  \vspace{0.5em}
  \caption{\small Selected literature in multiple hypothesis testing under
    dependence.\\ The categorization is partially subjective as
some authors do not use exactly the same terminology.}
  \label{tab:literature}
\end{table}

\vspace{3mm}
\textit{Notation.~~} Throughout the article, we use bold upper-case
letters for matrices and lower-case letters for vectors. We use
Latin letters for random variables and Greek letters for model
parameters. Subscripts of matrices are used to indicate row(s)
whenever possible. For example, if $\mathcal{C}$ is a set of indices,
then $\bm{\Gamma}_{\mathcal{C}}$ is the corresponding rows of $\bm{\Gamma}$. The $L_0$ norm of a vector is defined as the number
of nonzero entries: $\|\bm{\beta}\|_0
= |\{ 1 \le j \le p: \beta_j \ne 0 \}|$. A random matrix $\bm{E} \in \mathbb{R}^{n \times
  p}$ is said to follow a \textit{matrix normal} distribution with mean $\bm{M}\in \mathbb{R}^{n \times p}$, row covariance
$\bm{U} \in \mathbb{R}^{n \times n}$ and
column covariance $\bm{V} \in \mathbb{R}^{p \times p}$, abbreviated as
$\bm{E} \sim \mathrm{MN}(\bm{M}, \bm{U}, \bm{V})$, if the
vectorization of $\bm{E}$ by column follows the multivariate normal distribution $\mathrm{vec}(\bm{E}) \sim
\mathrm{N}(\mathrm{vec}(\bm{M}), \bm{V} \otimes \bm{U})$. When $\bm{U}
= \bm{I}_n$, this means the rows of $\bm{E}$ are i.i.d.\
$\mathrm{N}(0, \bm{V})$. We use the
usual notation in asymptotic statistics that a random variable is
${O}_p(1)$ if it is bounded in probability, and
${o}_p(1)$ if it converges to $0$ in
probability. Bold symbols ${\bm{O}}_p(1)$ or
${\bm{o}}_p(1)$ mean each entry of the vector is ${O}_p(1)$ or ${o}_p(1)$.

\section{The Model}
\label{sec:model}

\subsection{Linear model with confounders}
\label{sec:linear-model-with}

We consider a single primary variable of interest in this
section. It is common to add intercepts and known confounder effects
(such as lab and batch effects) in the regression model. This extension to
multiple linear regression does not change the main theoretical
results in this paper and is discussed in
\Cref{sec:extens-mult-regr}.

For simplicity, all the variables in this section are assumed to have mean $0$ marginally. Our
model is built on equation \eqref{eq:previous-model} that is already widely
used in the existing literature and we rewrite it here:
\refstepcounter{equation}\label{eq:linear-model}
\begin{equation}
  \tag{\theequation a}\label{eq:linear-model-y}
   {\bm{Y}}_{n \times p} = {\bm{X}}_{n \times 1} \,
   {\bm{\beta}}_{p \times 1}^T + \bm{Z}_{n \times r} \,
   \bm{\Gamma}_{p \times r}^T +
   \bm{E}_{n \times p}.
\end{equation}
As mentioned earlier, it is also crucial to model the dependence
of the confounders $\bm{Z}$ and the primary variable $\bm{X}$. We
assume a linear relationship as in \cref{eq:linear-model-z-sec1}
\begin{equation} \tag{\theequation b} \label{eq:linear-model-z}
\bm{Z} = \bm{X} \bm{\alpha}^T + \bm{W},
\end{equation}
and in addition some distributional assumptions on $\bm{X}$, $\bm{W}$ and the noise matrix $\bm{E}$
\begin{align}
  \tag{\theequation c} \label{eq:distribution-x}
  &X_i \overset{\mathrm{i.i.d.}}{\sim} \mathrm{mean}~0,~
  \mathrm{variance}~1,~i=1,\dotsc,n, \\
 \tag{\theequation d} \label{eq:distribution-w}
&\bm{W} \sim
\mathrm{MN}(\bm{0}, \bm{I}_n, \bm{I}_r),~\bm{W} \independent \bm{X}, \\
  \tag{\theequation e} \label{eq:distribution-e}
  &\bm{E} \sim \mathrm{MN}(\bm{0},
  \bm{I}_n, \bm{\Sigma}),~\bm{E} \independent (\bm{X},\bm{Z}).
\end{align}

The parameters in the model \cref{eq:linear-model} are $\bm{\beta} \in
\mathbb{R}^{p \times 1}$ the primary effects we are most interested in,
$\bm{\Gamma} \in \mathbb{R}^{p \times r}$ the influence of confounding
factors on the outcomes, $\bm{\alpha} \in
\mathbb{R}^{r \times 1}$ the association of the primary variable with
the confounding factors, and $\bm{\Sigma} \in
\mathbb{R}^{p \times p}$ the noise covariance
matrix. We assume $\bm{\Sigma}$ is diagonal $\bm{\Sigma}
= \mathrm{diag}(\sigma_1^2,\dotsc,\sigma_p^2)$, so the noise
for different outcome variables is independent. We discuss possible
ways to relax this independence assumption in \Cref{sec:inference-when-sigma}.

In \eqref{eq:distribution-x}, $X_i$ is not required to be
Gaussian or even continuous. For example, a binary or categorical
variable after normalization also meets this assumption. As mentioned
in \Cref{sec:introduction}, the parameter vector $\bm{\alpha}$
measures how severely the data are confounded. For a more
intuitive interpretation, consider an oracle procedure of estimating
$\bm{\beta}$ when the confounders $\bm{Z}$ in
\cref{eq:linear-model-y} are observed. The best linear unbiased estimator in this
case is the ordinary least squares $(\hat{{\beta}}^{\mathrm{OLS}}_j,
\hat{\bm{\Gamma}}^{\mathrm{OLS}}_j)$, whose variance is $ \sigma_j^2
\mathrm{Var}(X_i, \bm{Z}_i)^{-1} / n$. Using \cref{eq:linear-model-z,eq:distribution-w},
it is easy to show that $\mathrm{Var}(\hat{\beta}_j^{\mathrm{OLS}}) =
(1 + \|\bm{\alpha}\|_2^2) \sigma_j^2/n$ and
$\mathrm{Cov}(\hat{\beta}_j^{\mathrm{OLS}}, \hat{\beta}_k^{\mathrm{OLS}}) = 0$ for $j \ne k$. In summary,
\begin{equation} \label{eq:var-beta-oracle}
\mathrm{Var}(\hat{\bm{\beta}}^{\mathrm{OLS}}) = \frac{1}{n} (1 + \| \bm{\alpha}
\|_2^2) \bm{\Sigma}.
\end{equation}
Notice that in the unconfounded linear model in which $\bm{Z} =
\bm{0}$, the variance of the OLS estimator of $\bm{\beta}$ is
$\bm{\Sigma}/n$. Therefore, $1 + \|\bm \alpha\|_2^2$ represents the
relative loss of efficiency when we add observed variables $\bm{Z}$ to the
regression which are correlated with $\bm{X}$. In
\Cref{sec:inference-beta-alpha}, we show that the oracle efficiency
\eqref{eq:var-beta-oracle} can be asymptotically achieved even when
$\bm{Z}$ is unobserved.



Let $\bm{\theta} =
(\bm{\alpha}, \bm{\beta}, \bm{\Gamma}, \bm{\Sigma})$ be all the
parameters and $\bm{\Theta}$ be the parameter space. Without any
constraint, the model \cref{eq:linear-model} is unidentifiable. In
\Cref{sec:identification-gamma,sec:identification-beta} we show how to restrict $\bm{\Theta}$ to ensure identifiability.


\subsection{Rotation}
\label{sec:rotation}

Following \citet{sun2012}, we introduce
a transformation of the data to make the
identification issues clearer. Consider the Householder rotation matrix $\bm{Q}^{T} \in
\mathbb{R}^{n \times n}$ such that
$\bm{Q}^T \bm{X} = \|\bm{X}\|_2\bm{e}_1 = (\|\bm{X}\|_2,0,0,\dotsc,0)^T$. Left-multiplying $\bm{Y}$ by
$\bm{Q}^T$, we get
$\tilde{\bm{Y}} = \bm{Q}^T \bm{Y} = \|\bm{X}\|_2 \bm{e}_1 \bm{\beta}^T
    + \tilde{\bm{Z}} \bm{\Gamma}^T +  \tilde{\bm{E}}$,
where
\begin{equation} \label{eq:Z-tilde}
\tilde{\bm{Z}} = \bm{Q}^T\bm{Z} = \bm{Q}^T (\bm{X \alpha}^T +
\bm{W}) = \|\bm{X}\|_2 \bm{e}_1 \bm{\alpha}^T + \tilde{\bm{W}},
\end{equation}
and $\tilde{\bm{W}} = \bm{Q}^T \bm{W} \overset{d}{=} \bm{W}$,
$\tilde{\bm{E}} = \bm{Q}^T \bm{E} \overset{d}{=} \bm{E}$. As a
consequence, the first and the rest of the rows of $\tilde{\bm{Y}}$ are
\begin{align}
  &\tilde{\bm{Y}}_1 = \|\bm{X}\|_2 \bm{\beta}^T + \tilde{\bm{Z}}_1 \bm{\Gamma}^T +
  \tilde{\bm{E}}_1  \sim \mathrm{N}(\|\bm{X}\|_2 (\bm{\beta} +
  \bm{\Gamma}\bm{\alpha})^T, \bm{\Gamma} \bm{\Gamma}^T + \bm{\Sigma}) , \label{eq:Y-tilde-1} \\
  &\tilde{\bm{Y}}_{-1} = \tilde{\bm{Z}}_{-1} \bm{\Gamma}^T +
  \tilde{\bm{E}}_{-1} \sim \mathrm{MN}(\bm{0},
  \bm{I}_{n-1}, \bm{\Gamma} \bm{\Gamma}^T + \bm{\Sigma}). \label{eq:Y-tilde-m1}
\end{align}
Here $\tilde{\bm{Y}}_1$ is a $1 \times p$ vector,
$\tilde{\bm{Y}}_{-1}$ is a $(n-1) \times p$ matrix, and the
distributions are conditional on $\bm{X}$.

The parameters $\bm{\alpha}$ and
$\bm{\beta}$ only appear in \cref{eq:Y-tilde-1}, so their inference
(step 1 in our procedure) can be completely separated from
the inference of $\bm{\Gamma}$ and $\bm{\Sigma}$ (step 2 in our
procedure). In fact, $\tilde{\bm{Y}}_1 \independent
\tilde{\bm{Y}}_{-1}|\bm{X}$ because $\tilde{\bm{E}}_1 \independent
\tilde{\bm{E}}_{-1}$, so the two steps use mutually independent
information. This in turn greatly simplifies the theoretical analysis.

We intentionally use the symbol $\bm{Q}$ to resemble the QR decomposition of $\bm{X}$. In
\Cref{sec:extens-mult-regr} we show how to use the QR decomposition to
separate the primary effects from confounder and nuisance effects when
$\bm{X}$ has multiple columns. Using the same notation, we discuss
how SVA and RUV decouple the problem in a slightly different manner in \Cref{sec:sva}.

\subsection{Identifiability of \texorpdfstring{$\bm \Gamma$}{TEXT}}
\label{sec:identification-gamma}

Equation \eqref{eq:Y-tilde-m1} is just the exploratory factor analysis
model, thus $\bm{\Gamma}$ can
be easily identified up to some rotation under some mild conditions.
 Here we assume a classical sufficient condition for
the identification of $\bm{\Gamma}$ \citep[Theorem 5.1]{anderson1956}.


\begin{lemma}
\label{lemma:identify-gamma}
Let $\bm{\Theta} = \bm{\Theta}_0$ be the parameter space such that
\begin{enumerate}
\item If any row of $\bm{\Gamma}$ is deleted, there remain two
  disjoint submatrices of $\bm{\Gamma}$ of rank $r$, and
\item $\bm{\Gamma}^T \bm{\Sigma}^{-1} \bm{\Gamma}/p$ is
  diagonal and the diagonal elements are distinct, positive, and
  arranged in decreasing order.
\end{enumerate}
Then $\bm{\Gamma}$ and $\bm{\Sigma}$ are identifiable in the model
\cref{eq:linear-model}.
\end{lemma}

  In \Cref{lemma:identify-gamma}, condition (1) requires that $p\geq 2r + 1$.
  Condition (1) identifies $\bm{\Gamma}$ up to a
  rotation which is sufficient to identify $\bm{\beta}$. To
  see this, we can reparameterize $\bm{\Gamma}$ and $\bm{\alpha}$ to
  $\bm{\Gamma}\bm{U}$ and $\bm{U}^T\bm{\alpha}$ using an $r \times r$
  orthogonal matrix $\bm{U}$. This reparameterization does not change
  the distribution of $\tilde{\bm{Y}}_1$ in \cref{eq:Y-tilde-1} if
  $\bm{\beta}$ remains the same. Condition (2) identifies the rotation
  uniquely but is not necessary for our theoretical analysis in later
  sections.

\subsection{Identifiability of \texorpdfstring{$\bm \beta$}{TEXT}}
\label{sec:identification-beta}

The parameters $\bm{\beta}$ and $\bm{\alpha}$ cannot be identified from
\eqref{eq:Y-tilde-1} because they have in total $p + r$
parameters while $\tilde{\bm{Y}}_1$ is a length $p$ vector. If we write
$\mathcal{P}_{\bm{\Gamma}}$ and $\mathcal{P}_{\bm{\Gamma}^{\perp}}$ as
the projection onto the column space and orthogonal space of $\bm{\Gamma}$
so that $\bm{\beta} = \mathcal{P}_{\bm{\Gamma}} \bm{\beta} +
\mathcal{P}_{\bm{\Gamma}^{\perp}} \bm{\beta}$, it is impossible to
identify $\mathcal{P}_{\bm{\Gamma}} \bm{\beta}$ from \cref{eq:Y-tilde-1}.

This suggests that we should further restrict the
parameter space $\bm{\Theta}$.
We will reduce the degrees of freedom by restricting at least $r$ entries of
$\bm{\beta}$ to equal $0$. We consider two different sufficient conditions to identify $\bm{\beta}$:

\begin{description}
\item[Negative control] $\bm{\Theta}_1 = \{ (\bm{\alpha},
    \bm{\beta}, \bm{\Gamma}, \bm{\Sigma}): \bm{\beta}_{\mathcal{C}} =
\bm{0},~\mathrm{rank}(\bm{\Gamma}_{\mathcal{C}}) = r\}$ for a known
negative control set $|\mathcal{C}| \ge
r$.
\item[Sparsity] $\bm{\Theta}_2(s) = \{ (\bm{\alpha},
    \bm{\beta}, \bm{\Gamma}, \bm{\Sigma}): \|\bm{\beta}\|_0 \le \lfloor (p-s)/2
  \rfloor,~\mathrm{rank}(\bm{\Gamma}_{\mathcal{C}}) = r,~\forall \mathcal{C} \subset
  \{1,\dotsc,p\},~|\mathcal{C}| = s \}$ for some $r \le s \le p$.
\end{description}

\begin{proposition} \label{prop:identifiable}
  If $\bm{\Theta} = \bm{\Theta}_0 \cap \bm{\Theta}_1$ or $\bm{\Theta}
  = \bm{\Theta}_0 \cap \bm{\Theta}_2(s)$ for some $r \le s \le p$, the parameters $\bm{\theta} =
  (\bm{\alpha}, \bm{\beta}, \bm{\Gamma}, \bm{\Sigma})$ in the model
  \cref{eq:linear-model}
  are identifiable.
\end{proposition}
\begin{proof}
  Since $\bm{\Theta} \subset \bm{\Theta}_0$, we know from \Cref{lemma:identify-gamma} that $\bm{\Gamma}$ and $\bm{\Sigma}$ are
  identifiable. Now consider two combinations of parameters
  $\bm{\theta}^{(1)} = (\bm{\alpha}^{(1)}, \bm{\beta}^{(1)}, \bm{\Gamma}, \bm{\Sigma})$ and
  $\bm{\theta}^{(2)} = (\bm{\alpha}^{(2)}, \bm{\beta}^{(2)}, \bm{\Gamma},
  \bm{\Sigma})$ both in the space $\bm{\Theta}$ and inducing the same
  distribution in the model \cref{eq:linear-model}, i.e.\
  $\bm{\beta}^{(1)} + \bm{\Gamma} \bm{\alpha}^{(1)} =  \bm{\beta}^{(2)} + \bm{\Gamma} \bm{\alpha}^{(2)}$.

  Let $\mathcal{C}$ be the set of indices such that $\bm{\beta}^{(1)}_{\mathcal{C}} =
  \bm{\beta}^{(2)}_{\mathcal{C}} = \bm{0}$. If $\bm{\Theta} = \bm{\Theta}_0 \cap
  \bm{\Theta}_1$, we already know $|\mathcal{C}| \ge r$. If $\bm{\Theta} =
  \bm{\Theta}_0 \cap \bm{\Theta}_2(s)$, it is easy to show that $|\mathcal{C}| \ge
  s$ is also true because both $\bm{\beta}^{(1)}$ and
  $\bm{\beta}^{(2)}$ have at most $\lfloor (p-s)/2
  \rfloor$ nonzero entries. Along with the rank constraint on
  $\bm{\Gamma}_\mathcal{C}$, this implies that $\bm{\Gamma}_{\mathcal{C}}
  \bm{\alpha}^{(1)} = \bm{\Gamma}_{\mathcal{C}} \bm{\alpha}^{(2)}$. However, the
  conditions in $\bm{\Theta}_1$ and $\bm{\Theta}_2$ ensure that
  $\bm{\Gamma}_{\mathcal{C}}$ has full rank, so $\bm{\alpha}^{(1)} =
  \bm{\alpha}^{(2)}$ and hence $\bm{\beta}^{(1)} =
  \bm{\beta}^{(2)}$.
\end{proof}


\begin{remark}
   The condition (2) in \Cref{lemma:identify-gamma} that uniquely
   identifies $\bm{\Gamma}$ is not necessary for the identification of
   $\bm{\beta}$. This is because for any set $|\mathrm{C}| \ge r$ and any
   orthogonal matrix $\bm{U} \in \mathbb{R}^{r \times r}$, we always
   have $\mathrm{rank}(\bm{\Gamma}_{\mathcal{C}}) =
   \mathrm{rank}(\bm{\Gamma}_{\mathcal{C}}) \bm{U}$. Therefore
   $\bm{\Gamma}$ only needs to be identified up to a rotation.
\end{remark}

\begin{remark}
Almost all dense matrices of $\bm{\Gamma} \in \mathbb{R}^{p \times r}$
satisfy the conditions. However, for $\bm{\Theta}_2(s)$ the sparsity of $\bm{\Gamma}$ allowed
depends on the sparsity of $\bm{\beta}$. The condition $\bm{\Theta}_2(s)$ rules
out some too sparse
$\bm{\Gamma}$. In this case, one may consider
using confirmatory factor analysis instead of exploratory factor
analysis to model the relationship between confounders and
outcomes. For some recent identification results in confirmatory
factor analysis, see \citet{grzebyk2004,kuroki2014}.
\end{remark}

\begin{remark}
The maximum allowed $\|\bm{\beta}\|_0$ in $\bm{\Theta}_2$, $\lfloor
(p-r)/2 \rfloor$, is exactly the maximum breakdown point of a robust
regression with $p$ observations and $r$ predictors
\citep{maronna2006}. Indeed we use a standard robust
regression method to estimate $\bm{\beta}$ in this case in \Cref{sec:unkn-zero-indic}.
\end{remark}

\begin{remark}
To the best of our
knowledge, the only existing literature that explicitly addresses the
identifiability issue is \citet[Chapter 4.2]{sun2011}, where
the author gives sufficient conditions
for \emph{local} identifiability of
$\bm{\beta}$ by viewing \cref{eq:linear-model-y} as a ``sparse plus low
rank'' matrix decomposition problem. See \citet[Section
3.3]{chandrasekaran2012} for a more general discussion of the local
and global identifiability for this problem. Local
identifiability refers to identifiability of the parameters
in a neighborhood of the true values.
In contrast, the conditions in \Cref{prop:identifiable} ensure that
$\bm{\beta}$ is \emph{globally} identifiable in the restricted
parameter space.
\end{remark}

\section{Statistical Inference}
\label{sec:stat-infer}

As mentioned earlier in \Cref{sec:introduction}, the statistical inference consists of two steps:
the factor analysis (\Cref{sec:inference-gamma-sigma}) and the linear
regression (\Cref{sec:inference-beta-alpha}).

\subsection{Inference for \texorpdfstring{$\bm\Gamma$}{TEXT} and \texorpdfstring{$\bm \Sigma$}{TEXT}}
\label{sec:inference-gamma-sigma}

The most popular approaches for factor analysis are principal
component analysis (PCA) and maximum likelihood (ML). \citet{bai2002} derived a class of estimators of
$r$ by principal component analysis using various information
criteria. The estimators are consistent under
\Cref{assumption:large-factor} in this section and some additional
technical assumptions in \citet{bai2002}. Due to this reason, we assume the number of confounding factors $r$ is known
in 
this section. See \citet[Section 3]{owen2015} for a comprehensive
literature review of choosing $r$ in practice.

We are most interested in the asymptotic behavior of factor analysis when both
$n,p \to \infty$. In this case, PCA cannot consistently estimate the
noise variance $\bm{\Sigma}$ \citep{bai2012}. For theoretical analysis, we
use the quasi maximum likelihood estimate in \citet{bai2012} to get $
\hat{\bm{\Gamma}}$ and $\hat{\bm{\Sigma}}$. This estimator is called
``quasi''-MLE because it treats the factors $\tilde{
  \bm{Z}}_{-1}$ as fixed quantities. Since the confounders $\bm{Z}$ in
our model \cref{eq:linear-model} are random variables, we introduce a
rotation matrix $\bm{R} \in \mathbb{R}^{r \times r}$ and let
$\tilde{\bm{Z}}_{-1}^{(0)} = \tilde{\bm{Z}}_{-1} (\bm{R}^{-1})^T$,
$\bm{\Gamma}^{(0)} = \bm{\Gamma} \bm{R}$ be the target factors and
factor loadings that are studied in \citet{bai2012}.

To make $\tilde{\bm{Z}}_{-1}^{(0)}$ and $\bm{\Gamma}^{(0)}$ identifiable, \citet{bai2012} consider five
different identification conditions. However, the parameter of
interest in model \cref{eq:linear-model} is $\bm{\beta}$ instead of
$\bm{\Gamma}$ or $\bm{\Gamma}^{(0)}$. As we have discussed in
\Cref{sec:identification-beta}, we only need the
column space of $\bm{\Gamma}$ to estimate $\bm{\beta}$, which gives us
some flexibility of choosing the identification condition. In our
theoretical analysis we use the third condition (IC3) in
\citet{bai2012},  which imposes the constraints that $(n-1)^{-1} (\tilde{\bm{Z}}_{-1}^{(0)})^T \tilde{\bm{Z}}_{-1}^{(0)} =
\bm{I}_r$ and $p^{-1} \tilde{\bm{\Gamma}}^{(0)T} \bm{\Sigma}^{-1}
\bm{\Gamma}^{(0)}$ is diagonal. Therefore, the
rotation matrix $\bm{R}$ satisfies $\bm{R}
\bm{R}^T = (n-1)^{-1} \tilde{\bm{Z}}_{-1}^T \tilde{\bm{Z}}_{-1}$.

The quasi-log-likelihood being maximized in \citet{bai2012} is
\begin{align}\label{eq:qmle}
  - \frac{1}{2p} \log \mathrm{det}\left(\bm{\Gamma}^{(0)} (\bm{\Gamma}^{(0)})^T + \bm{\Sigma}\right) - \frac{1}{2p}
  \mathrm{tr}\left\{\bm{S} \left[\bm{\Gamma}^{(0)} (\bm{\Gamma}^{(0)})^T + \bm{\Sigma}\right]^{-1}\right\}
\end{align}
where $\bm{S}$ is the sample covariance matrix of $\tilde{\bm{Y}}_{-1}$.

The theoretical results in this section rely heavily on
recent findings in \citet{bai2012}. They use these three assumptions.

\begin{assumption} \label{assumption:error}
  The noise matrix $\bm{E}$ follows the matrix normal distribution $\bm{E} \sim \mathrm{MN}(\bm{0}, \bm{I}_n,
  \bm{\Sigma})$ and $\bm{\Sigma}$ is a diagonal matrix.
\end{assumption}

\begin{assumption} \label{assumption:bounded}
  There exists a positive constant $D$ such that $\|\bm{\Gamma}_{j}\|_2
  \le D$, $D^{-2} \le \sigma_j^2 \le D^2$ for all $j$, and the estimated
  variances $\hat{\sigma}_j^2 \in [D^{-2}, D^2]$ for all $j$.
\end{assumption}

\begin{assumption} \label{assumption:large-factor}
  The limits $\lim_{p \to \infty} p^{-1} \bm{\Gamma}^T
  \bm{\Sigma}^{-1} \bm{\Gamma}$ and $\lim_{p \to \infty} \sum_{j=1}^p
  \sigma_j^{-4} (\bm{\Gamma}_j \otimes \bm{\Gamma}_{j})
  (\bm{\Gamma}_j^T \otimes \bm{\Gamma}_{j}^T)$ exist and are positive
  definite matrices.
\end{assumption}


\begin{lemma}[\citet{bai2012}]\label{lem:ml-baili}
  Under
  \Cref{assumption:error,assumption:bounded,assumption:large-factor},
the maximizer $(\hat{\bm{\Gamma}},\hat{\bm{\Sigma}})$
of the quasi-log-likelihood~\eqref{eq:qmle} satisfies
\[
\sqrt n(\hat{\bm{\Gamma}}_j - \bm{\Gamma}_j^{(0)}) \overset{d}{\to}
  \mathrm{N}(\bm{0}, \sigma_j^2 \bm{I}_r),
\quad\text{and}\quad\sqrt n(\hat\sigma^2_j -
  \sigma^2_j)  \overset{d}{\to} \mathrm{N}(0, 2\sigma_j^4).
\]
\end{lemma}
In \Cref{app:lem-ml}, we prove some strengthened
technical results of \Cref{lem:ml-baili} that are used in the proof of
subsequent theorems.

\begin{remark}
  Assumption~\ref{assumption:bounded} is Assumption D
from \cite{bai2012}.  It requires that the diagonal elements of the
quasi-MLE $\hat\Sigma$ be uniformly bounded away from zero and infinity.
We would prefer boundedness to be a consequence of some assumptions
on the distribution of the data,
but at present we are unaware of any other results like
Lemma~\ref{lem:ml-baili}
which do not use this assumption. In practice, the
quasi-likelihood problem \eqref{eq:qmle} is commonly solved by the
Expectation-Maximization (EM) algorithm. Similar to
\citet{bai2012,bai2014theory}, we do not find it necessary
to impose an upper or lower bound for the parameters in the EM
algorithm in the numerical experiments.

\end{remark}

\subsection{Inference for \texorpdfstring{$\bm \alpha$}{TEXT} and \texorpdfstring{$\bm \beta$}{TEXT}}
\label{sec:inference-beta-alpha}


The estimation of $\bm{\alpha}$ and $\bm{\beta}$ is based on the
first row of the rotated outcome $\tilde{\bm{Y}}_1$ in \eqref{eq:Y-tilde-1}, which can be
rewritten as
\begin{equation} \label{eq:y-tilde-1}
  \tilde{\bm{Y}}_1^T / \|\bm{X}\|_2 = \bm{\beta} + \bm{\Gamma}
  (\bm{\alpha} + \tilde{\bm{W}}_1 / \|\bm{X}\|_2) + \tilde{\bm{E}}_{1}^T / \|\bm{X}\|_2
\end{equation}
where $\tilde{\bm{W}}_1 \sim \mathrm{N}(0, \bm{I}_p)$ is from
\cref{eq:Z-tilde} and $\tilde{\bm{W}}_1$ is independent of $\tilde{\bm{E}}_1 \sim
\mathrm{N}(0, \bm{\Sigma})$. Note that $\tilde{\bm
  Y}_{1}/\|\bm{X}\|_2$ is proportional to the sample covariance
between $\bm Y$ and $\bm X$.
All the methods described in this section
first try to find a good estimator $\hat{\bm{\alpha}}$.
They then use $\hat{\bm{\beta}} = \tilde{\bm{Y}}_1^T/\|\bm{X}\|_2 -
\hat{\bm{\Gamma}} \hat{\bm{\alpha}}$ to estimate $\bm{\beta}$.

To reduce variance, we choose to estimate \eqref{eq:y-tilde-1} conditional on $\tilde{\bm{W}}_1$.
 Also, to use the results in \Cref{lem:ml-baili}, we replace $\bm{\Gamma}$ by $\bm{\Gamma}^{(0)}$.
Then, we can rewrite \eqref{eq:y-tilde-1} as
\begin{equation} \label{eq:y-tilde-1-0}
  \tilde{\bm{Y}}_1^T / \|\bm{X}\|_2 = \bm{\beta} + \bm{\Gamma}^{(0)}
  \bm{\alpha}^{(0)} + \tilde{\bm{E}}_1^T / \|\bm{X}\|_2
\end{equation}
where $\bm{\Gamma}^{(0)} = \bm{\Gamma} \bm{R}$ and $\bm{\alpha}^{(0)} =
\bm{R}^{-1} (\bm{\alpha} + \tilde{\bm{W}}_1 / \|\bm{X}\|_2)$. Notice that the random $\bm{R}$
only depends on $\tilde{\bm{Y}}_{-1}$ and
thus is independent of $\tilde{\bm{Y}}_1$.
In the proof of the results in this section, we first consider the estimation
of $\bm{\beta}$ for fixed $\tilde{\bm{W}}_1$, $\bm{R}$ and $\bm{X}$,
and then show the asymptotic distribution of $\hat{\bm{\beta}}$
indeed does not depend on $\tilde{\bm{W}}_1$, $\bm{R}$ or $\bm{X}$, and thus also holds unconditionally.

\subsubsection{Negative control scenario}
\label{sec:known-zero-indices}

If we know a set $\mathcal{C}$ such that $\bm{\beta}_{\mathcal{C}} = 0$
(so $\bm{\Theta} \subset \bm{\Theta}_1$), then $\tilde{\bm{Y}}_1$ can be correspondingly separated into two parts:
\begin{equation} \label{eq:y-tilde-1-separate}
  \begin{split}
  \tilde{\bm{Y}}_{1,\mathcal{C}}^T / \|\bm{X}\|_2 &=
  \bm{\Gamma}_{\mathcal{C}}^{(0)} \bm{\alpha}^{(0)} +
  \tilde{\bm{E}}_{1,\mathcal{C}}^T / \|\bm{X}\|_2, \quad\text{and}\\
  \tilde{\bm{Y}}_{1,-\mathcal{C}}^T / \|\bm{X}\|_2 &= \bm{\beta}_{-\mathcal{C}} +
  \bm{\Gamma}_{-\mathcal{C}}^{(0)} \bm{\alpha}^{(0)} + \tilde{\bm{E}}_{1,-\mathcal{C}}^T / \|\bm{X}\|_2.
  \end{split}
\end{equation}
This estimator matches the RUV-4 estimator of \cite{gagnon2013}
except that it uses quasi-maximum likelihood estimates of $\bm{\Sigma}$
and $\bm{\Gamma}$ instead of using PCA, and generalized linear squares
instead of ordinary linear squares regression. The details are in Section~\ref{sec:ruv}.

The number of negative controls $|\mathcal{C}|$ may grow as $p \to \infty$.
We impose an additional assumption on the latent factors of the negative controls.

\begin{assumption}\label{assumption:nc-factor}
 $\lim_{p \to \infty}|\mathcal{C}|^{-1}
 \bm{\Gamma}^T_\mathcal{C}\bm{\Sigma}_{\mathcal{C}}^{-1}\bm{\Gamma}_\mathcal{C}$
 exists and is positive definite.
\end{assumption}

We consider the following negative control (NC) estimator where $\bm{\alpha}^{(0)}$ is estimated by generalized least squares:
\begin{align}
  &\hat{\bm{\alpha}}^{\mathrm{NC}} = (\hat{\bm{\Gamma}}_{\mathcal{C}}^T
  \hat{\bm{\Sigma}}_{\mathcal{C}}^{-1}
  \hat{\bm{\Gamma}}_{\mathcal{C}})^{-1}
  \hat{\bm{\Gamma}}_{\mathcal{C}}^T \hat{\bm{\Sigma}}_{\mathcal{C}}^{-1}
  \tilde{\bm{Y}}_{1,\mathcal{C}}^T / \| \bm{X} \|_2,~\mathrm{and} \label{eq:alpha-nc} \\
  &\hat{\bm{\beta}}^{\mathrm{NC}} = \tilde{\bm{Y}}_{1,-\mathcal{C}}^T / \| \bm{X} \|_2 -
  \hat{\bm{\Gamma}}_{-\mathcal{C}}
  \hat{\bm{\alpha}}^{\mathrm{NC}}. \label{eq:beta-nc}
\end{align}

Our goal is to show consistency and asymptotic variance of $\hat{\bm{\beta}}^{\mathrm{NC}}_{-\mathcal{C}}$.
Let $\bm \Sigma_\mathcal{C}$ represents the noise covariance matrix of the variables in $\mathcal{C}$. 
 We then have


\begin{theorem} \label{thm:asymptotics-nc}
  Under
  \Cref{assumption:error,assumption:bounded,assumption:large-factor,assumption:nc-factor},
  if $n,p \to \infty$ and $p/n^k \to 0$ for some
  $k > 0$, then for any fixed index set $\mathcal{S}$
with finite cardinality and
 $\mathcal{S} \cap \mathcal{C} = \emptyset$, we have
\begin{equation} \label{eq:nc-var-delta}
\sqrt{n} (\hat{\bm{\beta}}^{\mathrm{NC}}_{\mathcal{S}} - \bm{\beta}_{\mathcal{S}}) \overset{d}{\to}
\mathrm{N}(\bm{0}, (1 + \|\bm{\alpha}\|_2^2) (\bm{\Sigma}_{\mathcal{S}} + \bm{\Delta}_{\mathcal{S}}))
\end{equation}
where $\bm{\Delta}_{\mathcal{S}} = \bm{\Gamma}_{\mathcal{S}} (\bm{\Gamma}_{\mathcal{C}}^T \bm{\Sigma}_\mathcal{C}^{-1}
  \bm{\Gamma}_{\mathcal{C}})^{-1} \bm{\Gamma}_{\mathcal{S}}^T$.

If in addition, $|\mathcal{C}| \to \infty$, the minimum eigenvalue
of $\bm{\Gamma}_{\mathcal{C}}^T \bm{\Sigma}_\mathcal{C}^{-1}
\bm{\Gamma}_{\mathcal{C}} \to \infty$ by \Cref{assumption:nc-factor}, then
the maximum entry of $\bm{\Delta}_{\mathcal{S}}$ goes to $0$. Therefore in this case
\begin{equation} \label{eq:nc-var-oracle}
\sqrt{n} (\hat{\bm{\beta}}^{\mathrm{NC}}_{\mathcal{S}} - \bm{\beta}_{\mathcal{S}}) \overset{d}{\to}
\mathrm{N}(\bm{0}, (1 + \|\bm{\alpha}\|_2^2) \bm{\Sigma}_{\mathcal{S}}).
\end{equation}
\end{theorem}

The asymptotic variance in \eqref{eq:nc-var-oracle} is the same as the
variance of the oracle least squares in \eqref{eq:var-beta-oracle}.
Comparable oracle efficiency statements can be found in the econometrics literature
\citep{bai2006,wang2015}.
This is also the variance
used implicitly in RUV-4 as it treats the estimated
$\bm Z$ as given when deriving test statistics for $\bm \beta$.
When the number
of negative controls is not too large, say $|\mathcal{C}| = 30$, the
correction term $\bm{\Delta}_S$ is nontrivial and gives more accurate
estimate of the variance of $\hat{\bm{\beta}}^{\mathrm{NC}}$. See \Cref{sec:simulations} for more simulation results.


\subsubsection{Sparsity scenario}
\label{sec:unkn-zero-indic}

When the zero indices in $\bm{\beta}$ are unknown but sparse (so $\bm{\Theta}
\subseteq \bm{\Theta}_2$),
the estimation of
$\bm{\alpha}$ and $\bm{\beta}$ from $ \tilde{\bm{Y}}_1^T / \|\bm{X}\|_2 = \bm{\beta} + \bm{\Gamma}^{(0)} \bm{\alpha}^{(0)} + \tilde{\bm{E}}_1^T / \|\bm{X}\|_2$
can be cast as a robust regression by viewing $\tilde{\bm{Y}}_1^T$ as observations and $\bm{\Gamma}^{(0)}$ as design
matrix. The nonzero entries in $\bm{\beta}$ correspond to
outliers in this linear regression. 

The problem here has two nontrivial differences compared to
classical robust regression. First, we expect some entries of
$\bm{\beta}$ to be nonzero, and our goal is to
make inference on the outliers; second, we don't observe
the design matrix $\bm{\Gamma}^{(0)}$ but only have its estimator
$\hat{\bm{\Gamma}}$. In fact, if $\bm{\beta} = \bm{0}$ and
$\bm{\Gamma}^{(0)}$ is observed, the ordinary least squares
estimator of $\bm{\alpha}^{(0)}$ is unbiased and has variance of order
$1/(np)$, because the noise in \cref{eq:y-tilde-1} has variance $1/n$
and there are $p$ observations. Our main conclusion is
that $\bm{\alpha}^{(0)}$ can still be estimated very accurately given
the two technical difficulties.

Given a robust loss function $\rho$, we consider the following estimator:
\begin{align}
  \label{eq:alpha-rr}
  &\hat{\bm{\alpha}}^{\mathrm{RR}} = \arg \min \sum_{j=1}^p
  \rho \left( \frac{\tilde{Y}_{1j} / \|\bm{X}\|_2 - \hat{\bm{\Gamma}}_j^T \bm{\alpha}}{\hat{\sigma}_j}\right), ~\text{and}\\
  \label{eq:beta-rr}
  &\hat{\bm{\beta}}^{\mathrm{RR}} = \tilde{\bm{Y}}_1 / \|\bm{X}\|_2 - \hat{\bm{\Gamma}} \hat{\bm{\alpha}}^{\mathrm{RR}}.
\end{align}
For a broad class of loss functions $\rho$, estimating $\bm{\alpha}$ by \cref{eq:alpha-rr} is equivalent to
\begin{equation} \label{eq:sparse-penalty}
(\hat{\bm{\alpha}}^{\mathrm{RR}}, \tilde{\bm{\beta}}) = \arg
\min_{\bm{\alpha}, \bm{\beta}} ~ \sum_{j=1}^p \frac{1}{\hat{\sigma}_j^2} (\tilde{Y}_{1j}/\|\bm X\|_2 -
\beta_j - \hat{\bm{\Gamma}}_j^T \bm{\alpha})^2 + P_{{\lambda}}(\bm{\beta}),
\end{equation}
where $P_{{\lambda}}(\bm{\beta})$ is a penalty to promote sparsity
of $\bm{\beta}$  \citep{she2011}.
However $\hat{\bm{\beta}}^{\mathrm{RR}}$ is not identical to
$\tilde{{\bm{\beta}}}$, which is a sparse vector that does not have an asymptotic normal
distribution. The LEAPP algorithm \citep{sun2012} uses the
form~\eqref{eq:sparse-penalty}. Replacing it by the robust regression \cref{eq:alpha-rr,eq:beta-rr} allows us to derive significance tests of $H_{0j}:\beta_j = 0$.

We assume a smooth loss $\rho$
for the theoretical analysis:
\begin{assumption} \label{assumption:rho}
The penalty $\rho:\real\to[0,\infty)$
with $\rho(0) = 0$.  The function $\rho(x)$ is non-increasing when $x \leq 0$ and is
  non-decreasing when $x > 0$. The derivative
  $\psi = \rho'$ exists and $|\psi|\le D$ for some $D<\infty$.
Furthermore, $\rho$ is strongly convex in a neighborhood of $0$.
\end{assumption}

A sufficient condition for the local strong convexity is that $\psi' > 0$
exists in a neighborhood of $0$.
The next theorem establishes the consistency of $\hat{\bm{\beta}}^{\mathrm{RR}}$.



\begin{theorem} \label{thm:consistency-rr}
    Under \Cref{assumption:error,assumption:bounded,assumption:large-factor,assumption:rho
    },
    if $n,p \to \infty$, $p/n^k \to 0$ for some
  $k > 0$ and $\|\bm{\beta}\|_1/p \to 0$, then $\hat{\bm{\alpha}}^{\mathrm{RR}}
  \overset{p}{\to} \bm{\alpha}$. As a consequence, for any $j$,
  $\hat{{\beta}}_j^{\mathrm{RR}} \overset{p}{\to} {\beta}_j$.
\end{theorem}


To derive the asymptotic distribution, we consider the estimating equation corresponding to \eqref{eq:alpha-rr}. By taking the derivative of \eqref{eq:alpha-rr},
$\hat{\bm{\alpha}}^{\mathrm{RR}}$ satisfies
\begin{equation} \label{eq:Psi}
\bm{\Psi}_{p,\hat{\bm{\Gamma}},\hat{\bm{\Sigma}}}(\hat{\bm{\alpha}}^{\mathrm{RR}})
= \frac{1}{p} \sum_{j=1}^p \psi\left( \frac{\tilde{Y}_{1j} /
    \|\bm{X}\|_2 - \hat{\bm{\Gamma}}_j^T
    \hat{\bm{\alpha}}^{\mathrm{RR}}}{\hat{\sigma}_j}\right) \hat{\bm{\Gamma}}_j / \hat{\sigma}_j = \bm{0}.
\end{equation}

The next assumption is used to control the higher order term in a Taylor expansion of $\bm{\Psi}$.
\begin{assumption} \label{assumption:psi-2nd-derivative}
The first two derivatives of $\psi$ exist and both $|\psi'(x)|\le D$ and $|\psi''(x)|\le D$
hold at all $x$ for some $D<\infty$.
\end{assumption}

Examples of loss functions $\rho$ that satisfy
\Cref{assumption:rho,assumption:psi-2nd-derivative} include smoothed
Huber loss and Tukey's bisquare.

The next theorem gives the asymptotic distribution of
$\hat{\bm{\beta}}^{\mathrm{RR}}$ when the nonzero entries of
$\bm{\beta}$ are sparse enough. The asymptotic variance of
$\hat{\bm{\beta}}^{\mathrm{RR}}$ is, again, the oracle variance in \eqref{eq:var-beta-oracle}.

\begin{theorem} \label{thm:asymptotics-rr}
  Under \Cref{assumption:error,assumption:bounded,assumption:large-factor,assumption:rho
    ,assumption:psi-2nd-derivative}, if $n,p \to \infty$, with $p/n^k \to 0$ for some
  $k > 0$ and $\|
  \bm{\beta} \|_1 \sqrt{n}/p \to 0$, then
\[
\sqrt{n} (\hat{\bm{\beta}}_\mathcal{S}^{\mathrm{RR}} - \bm{\beta}_\mathcal{S}) \overset{d}{\to}
\mathrm{N}(\bm{0}, (1 + \|\bm{\alpha}\|_2^2) \bm{\Sigma}_\mathcal{S})
\]
for any fixed index set $S$ with finite cardinality.
\end{theorem}
If $n/p \to 0$,  then a sufficient condition for $\| \bm{\beta} \|_1
\sqrt{n}/p \to 0$ in \Cref{thm:asymptotics-rr} is $\|\bm{\beta} \|_1 = O(\sqrt{p})$.
If instead $n/p \to c \in (0,\infty)$, then $\|\bm{\beta} \|_1 = o(\sqrt{p})$ suffices.

\subsection{Hypothesis Testing}
\label{sec:hypothesis-testing}

In this section, we construct significance tests for $\bm{\beta}$ and
$\bm{\alpha}$ based on the asymptotic normal distributions in the previous section.

\subsubsection{Test of the primary effects}
\label{sec:test-primary}

We consider the asymptotic test for $H_{0j}:\beta_j =
0,~j=1,\dotsc,p$ resulting from the asymptotic distributions of
$\hat{\beta}_j$ derived in \Cref{thm:asymptotics-nc,thm:asymptotics-rr}. 
\begin{equation}
  \label{eq:test-statistic}
  t_j = \frac{\|\bm{X}\|_2\hat{\beta}_j}{\hat{\sigma}_j \sqrt{1 +
      \|\hat{\bm{\alpha}}\|^2}},\quad j = 1,\dotsc,p
\end{equation}
Here we require $|\mathcal{C}| \to \infty$ for the NC estimator.
The null hypothesis $H_{0j}$ is rejected at level-$\alpha$ if $|t_j| > z_{\alpha/2} =
\mathrm{\Phi}^{-1}(1-\alpha/2)$ as usual, where $\mathrm{\Phi}$ is the cumulative
distribution function of the standard normal. Note that here we slightly
abuse the notation $\alpha$ to represent the significance level and
this should not be confused with the model parameter $\bm{\alpha}$.

The next theorem shows that the overall type-I error and the family-wise error rate
(FWER) can be asymptotically controlled by using the test statistics $t_j,j=1,\dotsc,p$.

\begin{theorem} \label{cor:type-I}
Let $\mathcal{N}_p = \{j| \beta_j = 0, j = 1,\dotsc,p\}$ be all the
true null hypotheses.  Under the assumptions of \Cref{thm:asymptotics-nc} or
  \Cref{thm:asymptotics-rr}, $|\mathcal{C}| \to \infty$ for the NC
  scenario, as $n,p,|\mathcal{N}_p| \to \infty$
\begin{equation} \label{eq:type-I}
\frac{1}{|\mathcal{N}_p|} \sum_{j \in \mathcal{N}_p} I(|t_j| >
z_{\alpha/2}) \overset{p}{\to} \alpha, ~\mathrm{and}
\end{equation}
\begin{equation}
  \label{eq:fwer}
  \lim\sup \, \mathrm{P}\Big(\sum_{j \in \mathcal{N}_p} I(|t_j| > z_{\alpha/(2p)}) \ge
  1\Big) \le \alpha.
\end{equation}
\end{theorem}

Although the individual
test is asymptotically valid as 
$t_j \overset{d}{\to} \mathrm{N}(0,1)$, 
\Cref{cor:type-I} is not a trivial corollary of the asymptotic normal
distribution in \Cref{thm:asymptotics-nc,thm:asymptotics-rr}. This is because $t_j,j=1,\dotsc,p$
are not independent for finite samples. The proof of \Cref{cor:type-I}
investigates how the dependence of the test statistics diminishes when
$n, p \to \infty$. The proof of
\Cref{cor:type-I} already requires a careful investigation of the convergence
of $\hat{\bm{\beta}}$ in \Cref{thm:asymptotics-rr}. It is more
cumbersome to prove FDR control using our test statistics. 
In
\Cref{sec:numerical-experiments} we show that FDR is
usually well controlled in simulations for the
Benjamini-Hochberg procedure when the sample size is large enough.

\begin{remark} \label{rmk:calibration}
We find a calibration technique in \citet{sun2012} very useful
to improve the type I error and FDR control for finite sample size.
Because the asymptotic variance used in
\cref{eq:test-statistic} is the variance of an oracle OLS
estimator, when the sample size is not sufficiently large, the
variance of $\hat{\beta}^{\mathrm{RR}}$ should be slightly larger than
this oracle variance. To correct for this inflation, one can use median
absolute deviation (MAD) with customary scaling
to match the standard deviation for a Gaussian distribution to estimate the empirical standard error of
$t_j,j=1,\dotsc,p$ and divide $t_j$ by the estimated standard
error. The performance of this empirical calibration is studied in the simulations in \Cref{sec:simulations}.
\end{remark}

\subsubsection{Test of confounding}
\label{sec:test-confound}
We also consider a significance test for $H_{0, \bm \alpha}: \bm \alpha
= \bm 0$, under which the latent factors are not confounding.

\begin{theorem} \label{thm:alpha-test}
  Let the assumptions of \Cref{thm:asymptotics-nc} or
  \Cref{thm:asymptotics-rr} and  $|\mathcal{C}| \to \infty$ for the NC
  scenario be given. Under the null hypothesis that $\bm \alpha = \bm
  0$, for $\hat{\bm \alpha} = \hat{\bm \alpha}^{\mathrm{NC}}$ in
\eqref{eq:alpha-nc} or $\hat{\bm \alpha} = \hat{\bm
  \alpha}^{\mathrm{RR}}$ in \eqref{eq:alpha-rr}, we have
  $$n \cdot \hat {\bm \alpha}^T \hat{\bm{\alpha}} \overset{d}{\to} \chi_r^2$$
  where $\chi_r^2$ is the chi-square distribution with $r$ degree of freedom.
\end{theorem}

Therefore, the null hypothesis $H_{0, \bm \alpha}: \bm \alpha
= \bm 0$ is rejected if $n\cdot \hat {\bm \alpha}^T
\hat{\bm{\alpha}} > \chi_{r, \alpha}^2$ where $\chi_{r, \alpha}^2$ is
the upper-$\alpha$ quantile of $\chi_r^2$. This test, combined with
exploratory factor analysis, can be used as a
diagnosis tool for practitioners to check whether the data gathering
process has any confounding factors that can bias the multiple
hypothesis testing.

\section{Extension to Multiple Regression}
\label{sec:extens-mult-regr}


In \Cref{sec:model,sec:stat-infer} we assume that there is only one primary variable $\bm{X}$
and all the random variables $\bm{X}$, $\bm{Y}$ and $\bm{Z}$ have mean
$\bm{0}$. In practice, there may
be several predictors, or we may want to include an intercept term
in the regression model.
Here we  develop a multiple regression extension to the
original model \cref{eq:linear-model}.

Suppose we observe in total $d = d_0 + d_1$ random predictors
that can be separated into two groups:
\begin{enumerate}
  \item $\bm{X}_0$: $n \times d_0$ nuisance covariates that we would like
    to include in the regression model, and
  \item $\bm{X}_1$: $n \times d_1$ primary variables whose
    effects we want to study.
\end{enumerate}
For example, the intercept term can be included in $\bm{X}_0$ as a $n
\times 1$ vector of $1$ (i.e. a random variable with mean $1$ and
variance $0$).

\citet{leek2008} consider the case $d_0 = 0$ and $d_1
\ge 1$ for SVA and \citet{sun2012} consider the case $d_0
\ge 0$ and $d_1 = 1$ for LEAPP. Here we study the confounder adjusted
multiple regression in full generality, for any $d_0\ge0$ and $d_1\ge1$.
Our model is
\begin{subequations}\label{eq:linear-model-ext}
\begin{align}
  &
  {\bm{Y}}={\bm{X}_0}{\bm{\Beta}}^T_0
  +{\bm{X}_1}{\bm{\Beta}}^T_1
  +{\bm{Z}}{\bm{\Gamma}}^T
  +{\bm{E}},\\
  &  \begin{pmatrix}
    {\bm{X}}_{0i}\\
    {\bm{X}}_{1i} \\
  \end{pmatrix} \text{ are } \mathrm{i.i.d.}
  \text{ with }  \mathrm{E} \left[
  \begin{pmatrix}
    {\bm{X}}_{0i}\\
    {\bm{X}}_{1i}
  \end{pmatrix}
  \begin{pmatrix}
    {\bm{X}}_{0i}\\
    {\bm{X}}_{1i}
  \end{pmatrix}^T
  \right] ={\bm{\Sigma}}_{\bm{X}},\\
  &  {\bm{Z}} \mid ({\bm{X}}_0,
  {\bm{X}}_1) \sim \mathrm{MN}
  ({\bm{X}}_0{\bm{\Alpha}}_0^T
  +{\bm{X}}_1{\bm{\Alpha}}_1^T,
  {\bm{I}}_n, {\bm{I}}_r),\quad\text{and}\\
  &  {\bm{E}} \independent ({\bm{X}}_0,
  {\bm{X}}_1, {\bm{Z}}) ,~
  {\bm{E}} \sim \mathrm{MN} ({\bm{0}},
  {\bm{I}}_n, {\bm{\Sigma}}).
\end{align}
\end{subequations}
The model does not specify means for ${\bm{X}_{0i}}$
and ${\bm{X}_{1i}}$; we do not need them.
The parameters in this model are, for $i =0$ or $1$, ${\bm{\Beta}}_i \in \mathbb{R}^{p \times
d_i}$, ${\bm{\Gamma}} \in \mathbb{R}^{p \times r}$,
${\bm{\Sigma}}_{\bm{X}} \in \mathbb{R}^{d \times d}$, and
${\bm{\Alpha}}_i \in \mathbb{R}^{r \times d_i}$. The parameters
$\bm{\Alpha}$ and $\bm{\Beta}$ are the matrix versions of
$\bm{\alpha}$ and $\bm{\beta}$ in model
\cref{eq:linear-model}. Additionally, we assume
${\bm{\Sigma}}_{\bm{X}}$ is invertible. 
 To clarify our purpose, we are
primarily interested in estimating and testing for the significance of
$\bm{\Beta}_1$.

For the multiple regression model (\ref{eq:linear-model-ext}), we
again consider the rotation matrix $\bm{Q}^T$ that is given by
the QR decomposition
$
\begin{pmatrix}
  \bm{X}_0 & \bm{X}_1
\end{pmatrix}
={\bm{Q}}{\bm{U}}$ where
${\bm{Q}} \in \mathbb{R}^{n \times n}$ is an orthogonal matrix and
${\bm{U}}$ is an upper triangular matrix of size $n \times
d$. Therefore we have
\[
{\bm{Q}}^T
\begin{pmatrix}
  \bm{X}_0 & \bm{X}_1
\end{pmatrix}
={\bm{U}}=
\begin{pmatrix}
  \bm{U}_{00} & \bm{U}_{01} \\
  \bm{0} & \bm{U}_{11} \\
  \bm{0} & \bm{0} \\
\end{pmatrix}
\]
where $\bm{U}_{00}$ is a $d_0 \times d_0$ upper triangular matrix and
$\bm{U}_{11}$ is a $d_1 \times d_1$ upper triangular matrix. Now let the
rotated $\bm{Y}$ be
\begin{equation} \label{eq:Y-tilde-ext}
  \tilde{\bm{Y}}
  ={\bm{Q}}^T{\bm{Y}}=
  \begin{pmatrix}
    \tilde{\bm{Y}}_0 \\
    \tilde{\bm{Y}}_1 \\
    \tilde{\bm{Y}}_{-1}
  \end{pmatrix}
\end{equation}
where $\tilde{{\bm{Y}}}_0$ is $d_0 \times p$, $\tilde{{\bm{Y}}}_1$ is $d_1 \times p$ and
$\tilde{{\bm{Y}}}_{- 1}$ is $(n - d) \times p$, then we can
partition the model into three parts: conditional on both
${\bm{X}}_0$ and ${\bm{X}}_1$ (hence $\bm{U}$),
\begin{align}
  &\tilde{{\bm{Y}}}_0
  = \bm{U}_{00} \bm{\Beta}_0^T + \bm{U}_{01} \bm{\Beta}_1^T + \tilde{{\bm{Z}}}_0
  {\bm{\Gamma}}^T +
  \tilde{{\bm{E}}}_0, \label{eq:Y-tilde-0-ext}\\
  &\tilde{{\bm{Y}}}_1
  =\bm{U}_{11}
  {\bm{\Beta}}_1^T + \tilde{{\bm{Z}}}_1
  {\bm{\Gamma}}^T + \tilde{{\bm{E}}}_1
  \sim \mathrm{MN}
  ({\bm{U}}_{11}
  ({\bm{\Beta}_1}+{\bm{\Gamma}}{\bm{\Alpha}}_1)^T,
  {\bm{I}}_{d_1},
  {\bm{\Gamma}}{\bm{\Gamma}}^T
  +{\bm{\Sigma}})   \label{eq:Y-tilde-1-ext}\\
  &\tilde{{\bm{Y}}}_{- 1} =
  \tilde{{\bm{Z}}}_{- 1} {\bm{\Gamma}}^T
  + \tilde{{\bm{E}}}_{- 1} \sim \mathrm{\mathrm{MN}}
  ({\bm{0}}, {\bm{I}}_{n - d},
  {\bm{\Gamma}}{\bm{\Gamma}}^T
  +{\bm{\Sigma}}) \label{eq:Y-tilde-m1-ext}
\end{align}
where $\tilde{\bm{Z}} = \bm{Q}^T \bm{Z}$ and $\tilde{\bm{E}} =
\bm{Q}^T \bm{E}
\overset{d}{=} \bm{E}$. Equation \cref{eq:Y-tilde-0-ext} corresponds to the
nuisance parameters $\bm{\Beta}_0$ and is discarded according to the
ancillary principle. Equation \cref{eq:Y-tilde-1-ext} is the multivariate
extension to \cref{eq:Y-tilde-1} that is used to estimate
$\bm{\Beta}_1$ and equation \cref{eq:Y-tilde-m1-ext} plays the same role as
\cref{eq:Y-tilde-m1} to estimate $\bm{\Gamma}$ and $\bm{\Sigma}$.

We consider the asymptotics when $n,p \to \infty$
and $d,r$ are fixed and known. Since $d$ is fixed, the estimation of $\bm{\Gamma}$ is not different
from the simple regression case and we can use the maximum likelihood
factor analysis described in \Cref{sec:inference-gamma-sigma}.
Under
\Cref{assumption:error,assumption:bounded,assumption:large-factor},
the precision results of $\hat{\bm{\Gamma}}$ and $\hat{\bm{\Sigma}}$ (\Cref{lem:gamma-ml})
still hold.

Let
${\bm{\Sigma}}_{\bm{X}}^{- 1} = \bm{\Omega} =
\bigl(\begin{smallmatrix}
  {\bm{\Omega}}_{00} &  {\bm{\Omega}}_{01}\\
  {\bm{\Omega}}_{10} &  {\bm{\Omega}}_{11}
\end{smallmatrix}\bigr)$.
In the proof of
\Cref{thm:asymptotics-nc,thm:asymptotics-rr}, we consider a fixed
sequence of $\bm{X}$ such that ${\|\bm{X}\|_2}/{\sqrt{n}} \to 1$.
Similarly, we have the following lemma in the multiple regression scenario:
\begin{lemma}\label{lem:part-asym}
  As $n \to \infty$, $\frac{1}{n} \bm{U}_{11}^T
  \bm{U}_{11}
  \overset{a.s.}{\to}
  {\bm{\Omega}}_{11}^{-1}$.
\end{lemma}

Similar to \eqref{eq:y-tilde-1}, we can rewrite \eqref{eq:Y-tilde-1-ext} as
\begin{equation*} 
  \tilde{\bm{Y}}_1^T \bm{U}_{11}^{-T} = \bm{\Beta}_1 + \bm{\Gamma}
  (\bm{\Alpha}_1 + \tilde{\bm{W}}_1\bm{U}_{11}^{-T}) + \tilde{\bm{E}}_{1}\bm{U}_{11}^{-T}
\end{equation*}
where $\tilde{\bm{W}}_1\sim \mathrm{MN}(\bm{0}, \bm{I}_{d_1}, \bm{I}_p)$ is independent
from $\tilde{\bm{E}}_1$. As in \Cref{sec:inference-beta-alpha}, we derive statistical
properties of the estimate of
$\bm \Beta_1$ for a fixed sequence of $\bm X$, $\tilde{\bm{W}}_1$ and $\bm Z$,
which also hold unconditionally.
For simplicity, we assume that the negative controls are a known
set of variables $\mathcal{C}$ with
$\bm{\Beta}_{1,\mathcal{C}} = \bm{0}$.
We can then estimate each column of ${\bm{\Alpha}}_1$ by applying the
negative control (NC) or robust regression (RR) we discussed in
\Cref{sec:unkn-zero-indic,sec:known-zero-indices} to the corresponding row of $\tilde{\bm{Y}}_1\bm{U}_{11}^{-T}$, and then
estimate ${\bm{\Beta}_1}$ by
\[ \hat{{\bm{\Beta}}}_1
   =   \tilde{{\bm{Y}}}_1^T \bm{U}_{11}^{- T}-
   \hat{{\bm{\Gamma}}}
   \hat{{\bm{\Alpha}}}_1.
\]
Notice that $\tilde{\bm{E}}_1\bm{U}_{11}^{-T}\sim
\mathrm{MN}\big(\bm{0}, \bm{\Sigma}, \bm{U}_{11}^{-1}\bm{U}_{11}^{-T}\big)$. Thus the
``samples'' in the robust regression,
which are actually the $p$ variables in the original problem are still
independent within each column. Though the estimates of each column of $\bm{\Alpha}_1$
may be correlated, we will show that the correlation won't affect inference on
$\bm{\Beta}_1$. As a result,  we still get
asymptotic results similar to
\Cref{thm:asymptotics-rr} for the multiple regression model \eqref{eq:linear-model-ext}:

\begin{theorem} \label{thm:multiple-rr}
  Under \Cref{assumption:error,assumption:bounded,assumption:large-factor,assumption:nc-factor,assumption:rho,
    assumption:psi-2nd-derivative}, if $n, p
    \to \infty$, with $p/n^k \to 0$ for some $k > 0$,
    and $\|\mathrm{vec}(\bm{\Beta}_1)\|_1 \sqrt{n} / p
  \to 0$, then for any fixed index set $\mathcal{S}$ with finite cardinality $|\mathcal{S}|$,
  \begin{align}
    \label{eq:nc-asym-ext}
\sqrt{n}
     (\hat{{\bm{\Beta}}}_{1, \mathcal{S}}^{\mathrm{NC}}
     -{\bm{\Beta}_{1, S}}) &\overset{d}{\to} \mathrm{MN}
     ({\bm{0}}_{|\mathcal{S}| \times k_1},
     {\bm{\Sigma}_\mathcal{S}} + \bm{\Delta}_{\mathcal{S}},
     {\bm{\Omega}}_{11}
     +{\bm{\Alpha}}_1^T
     {\bm{\Alpha}}_1),\quad\text{and}\\
\label{eq:rr-asym-ext}
\sqrt{n}
     (\hat{{\bm{\Beta}}}_{1, \mathcal{S}}^{\mathrm{RR}}
     -{\bm{\Beta}_{1, S}}) &\overset{d}{\to} \mathrm{MN}
     ({\bm{0}}_{|\mathcal{S}| \times k_1},
     {\bm{\Sigma}_\mathcal{S}},
     {\bm{\Omega}}_{11}
     +{\bm{\Alpha}}_1^T
     {\bm{\Alpha}}_1)
  \end{align}
where $\bm \Delta_\mathcal{S}$ is defined in \Cref{thm:asymptotics-nc}.
\end{theorem}

As for the asymptotic efficiency of this estimator, we again compare
it to the oracle OLS estimator of $\bm{\Beta}_1$ which observes
confounding variables ${\bm{Z}}$ in \eqref{eq:linear-model-ext}.
In the multiple regression model, we claim that
$\hat{{\bm{\Beta}}}_1^{\mathrm{RR}}$ still
reaches the oracle asymptotic efficiency. In fact, let $\bm{\Beta}= \begin{pmatrix}
  \bm{\Beta}_0 & \bm{\Beta}_1 & \bm{\Gamma}
\end{pmatrix}$. The
oracle OLS estimator of $\bm{\Beta}$,
$\hat{{\bm{\Beta}}}^{\mathrm{OLS}}$, is unbiased and its vectorization
has variance
$\bm{V}^{-1} \otimes \bm{\Sigma}/n$ where
\[
  \bm{V} =
  \begin{pmatrix}
  \bm{\Sigma}_{\bm{X}} & \bm{\Sigma}_{\bm{X}} \bm{\Alpha}^T\\
  \bm{\Alpha}\bm{\Sigma}_{\bm{X}} & \bm{I}_r
  +\bm{\Alpha}\bm{\Sigma}_{\bm{X}} \bm{\Alpha}^T
  \end{pmatrix},~ \mathrm{for}~
  \bm{\Alpha} =
  \begin{pmatrix}
    \bm{\Alpha}_0 & \bm{\Alpha}_1
  \end{pmatrix}.
\]
By the block-wise matrix inversion formula, the top left $d \times d$
block of $\bm{V}^{-1}$ is $\bm{\Sigma}_{\bm{X}}^{-1} +\bm{\Alpha}^T
\bm{\Alpha}$. The variance of $\hat{\bm{\Beta}}_1^{\mathrm{OLS}}$
only depends on the bottom right $d_1 \times d_1$ sub-block of this $d
\times d$ block, which is simply $\bm{\Omega}_{11} +\bm{\Alpha}_1^T
\bm{\Alpha}_1$. Therefore 
$\hat{{\bm{\Beta}}}_1^{\mathrm{OLS}}$ is unbiased and its
vectorization has variance $ ({\bm{\Omega}}_{11}
     +{\bm{\Alpha}}_1^T
     {\bm{\Alpha}}_1) \otimes {\bm{\Sigma}}/n$, matching the asymptotic variance
     of $\hat{\bm{\Beta}}^{\mathrm{RR}}_1$ in \Cref{thm:multiple-rr}.



\section{Discussion}
\label{sec:discussions}

\subsection{Confounding vs.\ unconfounding}
\label{sec:latent-but-unconf}

The issue of multiple testing dependence arises because $\bm{Z}$ in
the true model \eqref{eq:previous-model} is unobserved.
We have focused on the case where $\bm{Z}$ is confounded with
the primary variable.
Some similar results were obtained earlier
for the unconfounded case, corresponding to $\bm{\alpha}=0$ in our notation.
For example, \citet{lan2014} used a factor model to improve the
efficiency of significance tests of the regression intercepts.
\citet{jin2012comment,li2014rate} developed more powerful procedures for testing
$\bm{\beta}$ while still controlling FDR under unconfounded dependence.

In another related work, \citet{fan2012} imposed a factor structure
on the unconfounded test statistics, whereas this paper and the articles discussed later in
\Cref{sec:comp-with-exist} assume a factor structure on the raw data.
\citet{fan2012} used an approximate factor model to accurately estimate the false
discovery proportion. Their correction procedure also includes a
step of robust regression. Nevertheless, it is often difficult to interpret the factor
structure of the test statistics. In comparison, the latent variables
$\bm{Z}$ in our model \eqref{eq:linear-model}, whether confounding or
not, can be interpreted as batch effects, laboratory conditions, or
other systematic bias. Such problems are widely observed in genetics
studies (see e.g.\ the review article \citep{leek2010}).

As a final remark, some of the models and methods developed in
the context of unconfounded hypothesis testing may be useful for
confounded problems as well. For example,
the relationship between $\bm{Z}$ and $\bm{X}$  needs not be linear
as in~\eqref{eq:linear-model-z-sec1}. In certain applications, it may be
more appropriate to use a time-series model \citep{sun2009} or
a mixture model \citep{efron2010}.

\subsection{Marginal effects vs.\ direct effects}
\label{sec:marginal-effects-vs}

In \Cref{sec:introduction}, we switched our interest from
the marginal effects $\bm{\tau}$ in \eqref{eq:marginal-effect} to the
direct effects $\bm{\beta}$. We believe that they are usually more
scientifically meaningful and interpretable than the marginal effects.
For instance, if the treated (control) samples are analyzed by machine
A (machine B), and the machine A outputs
higher values than B, we certainly do not
want to include the effects of this machine to machine variation on the outcome
measurements.

When model
\eqref{eq:linear-model} is interpreted as a ``structural
equations model'' \citep{bollen1989}, $\bm \beta$ is indeed the causal effect of
$\bm{X}$ on $\bm{Y}$ \citep{pearl2009}. In this paper we do not make
such structural assumptions about the data generating
process. Instead, we use \eqref{eq:linear-model} to describe the screening
procedure commonly applied in high throughput data analysis. The model
\eqref{eq:linear-model} also describes how we think the marginal
effects can be confounded and hence different from the more meaningful direct effects
$\bm{\beta}$. Additionally, the asymptotic setting in this paper is quite
different from that in the traditional structural equations model.

\subsection{Comparison with existing confounder adjustment methods}
\label{sec:comp-with-exist}

We discuss in more detail how previous methods of confounder
adjustment, namely SVA \citep{leek2007,leek2008},
RUV-4 \citep{gagnon2013,gagnon2012}
and
LEAPP \citep{sun2012}, fit in the framework \cref{eq:linear-model}. See
\citet{perry2013degrees} for an alternative approach of bilinear
regression with latent factors that is also motivated by
high-throughput data analysis.

\subsubsection{SVA}
\label{sec:sva}

There are two versions of SVA: the reduced subset SVA (subset-SVA) of
\citet{leek2007} and the iteratively reweighted SVA (IRW-SVA) of
\citet{leek2008}. Both of them can be interpreted as the two-step
statistical procedure in the framework \cref{eq:linear-model}. In the first step, SVA
estimates the confounding factors by applying PCA to the residual matrix
$({\bm{I}}-{\bm{H}}_{{\bm{X}}}) {\bm{Y}}$ where
${\bm{H}}_{{\bm{X}}}
={\bm{X}} ({\bm{X}}^T
{\bm{X}})^{- 1} {\bm{X}}^T$ is the
projection matrix of ${\bm{X}}$.
In contrast, we
applied factor analysis to the rotated residual matrix
$(\bm{Q}^T{\bm{Y}})_{- 1}$, where $\bm{Q}$ comes from
the QR decomposition of $\bm X$ in \Cref{sec:extens-mult-regr}.
To see why these two approaches lead to the same estimate
of $\bm{\Gamma}$, we introduce the block form of $\bm Q = \begin{pmatrix} \bm Q_1 & \bm Q_2\end{pmatrix}$ where
$\bm Q_1 \in \mathbb{R}^{n \times d}$ and $\bm Q_2 \in \mathbb{R}^{n
  \times (n-d)}$. It is easy to show that $(\bm{Q}^T{\bm{Y}})_{- 1} = \bm
Q_2^T \bm Y$ and $({\bm{I}}-{\bm{H}}_{{\bm{X}}}) {\bm{Y}} = \bm Q_2 \bm Q_2^T \bm Y$.
Thus our rotated matrix $(\bm{Q}^T{\bm{Y}})_{- 1}$ decorrelates the
residual matrix by left-multiplying by $\bm{Q}_2$ (because $\bm Q_2^T\bm
Q_2 = \bm I_{n - d}$).
Because $(\bm Q_2^T \bm Y)^T\bm Q_2^T \bm Y = (\bm Q_2 \bm
Q_2^T \bm Y)^T\bm Q_2 \bm
Q_2^T \bm Y$,
$(\bm Q^T\bm Y)_{-1}$ and $(\bm I - \bm H_{\bm X})\bm Y$
have the same sample covariance matrix, they will yield the same
factor loading estimate under PCA and also under MLE.
The main advantage
of using the rotated matrix is theoretical: the rotated residual matrices
have independent rows.  


Because SVA doesn't assume an
explicit relationship between the primary variable $\bm{X}$ and the
confounders $\bm{Z}$, it cannot use the regression \cref{eq:y-tilde-1}
to estimate $\bm{\alpha}$ (not even defined) and $\bm{\beta}$. Instead, the two SVA algorithms try to reconstruct the
surrogate variables, which are essentially the confounders
$\bm{Z}$ in our framework. Assuming the true primary effect
$\bm{\beta}$ is sparse, the subset-SVA algorithm finds the outcome
variables $\bm{Y}$ that have the smallest marginal correlation with $\bm{X}$
and uses their principal scores as $\bm{Z}$. Then, it computes the p-values
by F-tests comparing the linear regression models with and without $\bm{Z}$. This procedure can
easily fail because a small marginal correlation does not imply no real effect of $\bm X$
due to the confounding factors. For example, most of
the marginal effects in the gender study in
\Cref{fig:gender-naive} are very small, but after confounding
adjustment we find some are indeed significant (see \Cref{sec:gender-study}).

The IRW-SVA algorithm modifies subset-SVA by iteratively choosing the subset.
At each step, IRW-SVA gives a weight to each outcome
variable based on how likely $\beta_j = 0$ the current estimate of
surrogate variables. The weights are then used in
a weighted PCA algorithm
to update the estimated surrogate variables. IRW-SVA may be
related to our robust regression estimator in
\Cref{sec:unkn-zero-indic} in the sense that an M-estimator is
commonly solved by Iteratively Reweighted Least Squares (IRLS) and the
weights also represents how likely the data point is an outlier.
However, unlike IRLS, the iteratively reweighted PCA algorithm is not even guaranteed to converge.
Some previous articles \citep{gagnon2013,sun2012} and our
experiments in \Cref{sec:simulations} and Supplement
\citep{wang2015supplement} show that SVA is outperformed by the
NC and RR estimators in most confounded examples. 

\subsubsection{RUV}\label{sec:ruv}

\citet{gagnon2013} derived the RUV-4
estimator of $\bm \beta$ via a sequence of heuristic calculations.
In \Cref{sec:known-zero-indices}, we derived an analytically more tractable
estimator $\hat{\bm \beta}^{\mathrm{NC}}$ which is actually the same as RUV-4,
 with the
only difference being that we use MLE instead of PCA to estimate the factors and GLS instead of OLS in
\cref{eq:alpha-nc}. To see why $\hat{\bm\beta}^{\mathrm{NC}}$ is essentially the same as
$\hat{\bm\beta}^{\mathrm{RUV-4}}$,
in the first step of RUV-4 it uses the residual matrix to estimate
$\bm \Gamma$ and $\bm Z$, which yields the same estimate as using
the rotated matrix (\Cref{sec:sva}). In the second step,
RUV-4 estimates $\bm \beta$ via a regression of $\bm{Y}$ on
$\bm X$ and $\hat{\bm Z} = \bm Q\begin{pmatrix}\tilde{\bm Z}_{-1}^T&
  \hat{\bm \alpha}^T\end{pmatrix}^T$. This is equivalent to using
ordinary least squares (OLS) to estimate $\bm{\alpha}$ in \eqref{eq:y-tilde-1-separate}.
Based on more heuristic
calculations, the authors claim that the RUV-4 estimator has approximately
the oracle variance. We rigorously prove
this statement in \Cref{thm:asymptotics-nc} when the number of negative
controls is large and give a finite sample correction when the negative
controls are few. In \Cref{sec:simulations} we show this correction is
very useful to control the type I error and FDR in simulations.

\subsubsection{LEAPP}
\label{sec:leapp}

We follow the two-step procedure and robust regression framework in
LEAPP \citep{sun2012} in this paper, thus
the test statistics $t_j^{\mathrm{RR}}$ are very similar to the test statistics in LEAPP.
The difference is that LEAPP uses the
$\Theta$-IPOD algorithm of \citet{she2011} for outlier detection,
which is robust against outliers at leverage points but is not easy
to analyze.  Indeed \citet{sun2012} replaced it by the Dantzig selector
in its theoretical appendix.
The classical M-estimator, although not robust to leverage
points \citep{yohai1987high}, allows us to study the theoretical properties more easily. In practice,
LEAPP and RR estimator usually produce very similar results; see
\Cref{sec:simulations} for a numerical comparison.

\subsection{Inference when \texorpdfstring{$\bm \Sigma$}{TEXT} is nondiagonal}
\label{sec:inference-when-sigma}

Our analysis is based on the assumption that the noise covariance matrix
$\bm\Sigma$ is diagonal,
though in many applications, the researcher might suspect that the
outcome variables $\bm{Y}$ in model \cref{eq:linear-model} are still
correlated after conditioning on the latent factors.
Typical examples include gene regulatory networks \citep{de2004} and
cross-sectional panel data \citep{pesaran2004}, where the
variable dependence sometimes cannot be fully explained by the
latent factors or may simply require too many of them.
\citet{bai2012approximate} extend the theoretical results in
\citet{bai2012} to approximate factor models allowing for weakly correlated noise.
Approximate factor models have also been discussed in \citet{fan2013}.

\section{Numerical Experiments}
\label{sec:numerical-experiments}

\subsection{Simulations}
\label{sec:simulations}

We have provided theoretical guarantees of confounder adjusting
methods in various settings and the asymptotic regime of $n, p \to \infty$ (e.g.\
\Cref{thm:asymptotics-nc,thm:consistency-rr,thm:asymptotics-rr,cor:type-I,thm:multiple-rr}).
Now we use numerical simulations to verify these results and further
study the finite sample properties of our estimators and tests statistics.

The simulation data are generated from the single primary variable model
\eqref{eq:linear-model}. More
specifically, $X_i$ is a centered binary variable
$
(X_i + 1)/2
\overset{\mathrm{i.i.d.}}{\sim} \mathrm{Bernoulli}(0.5),
$
and $\bm{Y}_i$, $\bm{Z}_i$ are generated according to
\cref{eq:linear-model}.

For the parameters in the
model, the noise variances are generated by
$\sigma_j^2 \overset{\mathrm{i.i.d.}}{\sim}\mathrm{InvGamma}(3, 2),~j=1,\dotsc,p$, and so
$\mathbb{E}(\sigma_j^2) = \mathrm{Var}(\sigma_j^2) = 1$. We set each
$\alpha_k = \|\bm{\alpha}\|_2/\sqrt r$ equally for $k = 1, 2, \cdots,
r$ where $\|\bm{\alpha}\|_2^2$ is set to $1$, so the variance of $X_i$
explained by the confounding factors is $R^2 = 50\%$. (Additional
results for $R^2 = 5\%$ and $0$ are in the Supplement.)
The primary effect $\bm{\beta}$ has
independent components $\beta_i$ taking the values $3 \sqrt{1 + \|\bm{\alpha}\|_2^2}$ and $0$ with probability
$\pi = 0.05$ and $1 - \pi = 0.95$, respectively, so the nonzero
effects are sparse and have effect size $3$. This implies that the oracle estimator has
power approximately $\mathrm{P}(\mathrm{N}(3,1) > z_{0.025}) =
0.85$ to detect the signals at a significance level of $0.05$. We set the number of latent
factors $r$ to be either $2$ or $10$. For the latent factor loading matrix $\bm \Gamma$, we take
$\bm\Gamma = \tilde{\bm\Gamma}\bm{D}$ where $\tilde{\bm\Gamma}$ is a
$p \times r$ orthogonal matrix sampled uniformly from the Stiefel
manifold $V_r(\mathbb{R}^p)$, the set of all $p \times r$ orthogonal
matrix. Based on \Cref{assumption:large-factor}, we set the latent factor strength $\bm{D} =
\sqrt p\cdot \mathrm{diag}(d_1, \cdots, d_r)$ where $d_k = 3 -
2(k-1)/(r - 1)$ thus $d_1$ to $d_r$ are distributed evenly inside the interval $[3, 1]$.
As the number of factors $r$ can be easily estimated for this strong factor setting
(more discussions can be found in \citet{owen2015}),
we assume that the number $r$ of factors is known to all of the algorithms in this simulation.


We set $p = 5000$, $n = 100$ or $500$ to mimic the data size of many
genetic studies. For the negative control scenario, we choose
$|\mathcal C|= 30$ negative controls at random from the zero positions of
$\bm \beta$.
We expect that negative control methods would perform better with a larger value of
$|\mathcal{C}|$ and worse with a smaller value.
The choice $|\mathcal{C}|=30$ is around the size of the spike-in controls in
many microarray experiments \citep{gagnon2012}.
For the loss function in
our sparsity scenario, we use Tukey's bisquare which is optimized via IRLS
with an ordinary least-square fit as the starting values of the coefficients.
Finally, each of the four combinations of $n$ and $r$ is randomly repeated $100$ times.

We compare the performance of nine
different approaches. There are two baseline methods:
the ``naive'' method estimates $\bm\beta$ by a
linear regression of $\bm{Y}$ on just the observed primary variable
$\bm{X}$ and calculates p-values using the classical t-tests, while the
``oracle'' method regresses $\bm{Y}$ on both $\bm{X}$ and the
confounding variables $\bm{Z}$ as described in \Cref{sec:linear-model-with}.
There are three methods in the RUV-4/negative controls family:
the RUV-4 method \citep{gagnon2013},
our ``NC'' method which computes test statistics using
$\hat{\bm\beta}^{\mathrm{NC}}$ and its variance estimate $(1 + \|\hat{\bm{\alpha}}\|_2^2)(\hat{\bm{\Sigma}}+\hat{\bm{\Delta}})$,  and our ``NC-ASY'' method
which uses the same $\hat{\bm \beta}^{\mathrm{NC}}$ but estimates its variance by
$(1 + \|\hat{\bm{\alpha}}\|_2^2)\hat{\bm{\Sigma}}$.
We compare four methods in the SVA/LEAPP/sparsity family:
these are ``IRW-SVA'' \citep{leek2008},  ``LEAPP'' \citep{sun2012},
the ``LEAPP(RR)'' method which is our RR estimator using M-estimation at the robustness stage and
computes the test-statistics using \eqref{eq:test-statistic}, and
 the ``LEAPP(RR-MAD)'' method which uses the median absolute deviation (MAD) of
the test statistics in \eqref{eq:test-statistic} to calibrate them.
(see \Cref{sec:hypothesis-testing})

To measure the performance of these methods, we report the type I
error (\Cref{cor:type-I}), power, false discovery proportion (FDP)
and precision of hypotheses with the smallest $100$ p-values
in the $100$ simulations. For both the type I error and power, we set
the significance level to be $0.05$. For FDP, we use
Benjamini-Hochberg procedure with FDR controlled at $0.2$. These
metrics are plotted in \Cref{fig:p-5000-methods} under different
settings of $n$ and $r$.


\begin{figure}[hp]
  \centering
   \includegraphics[width = \textwidth]{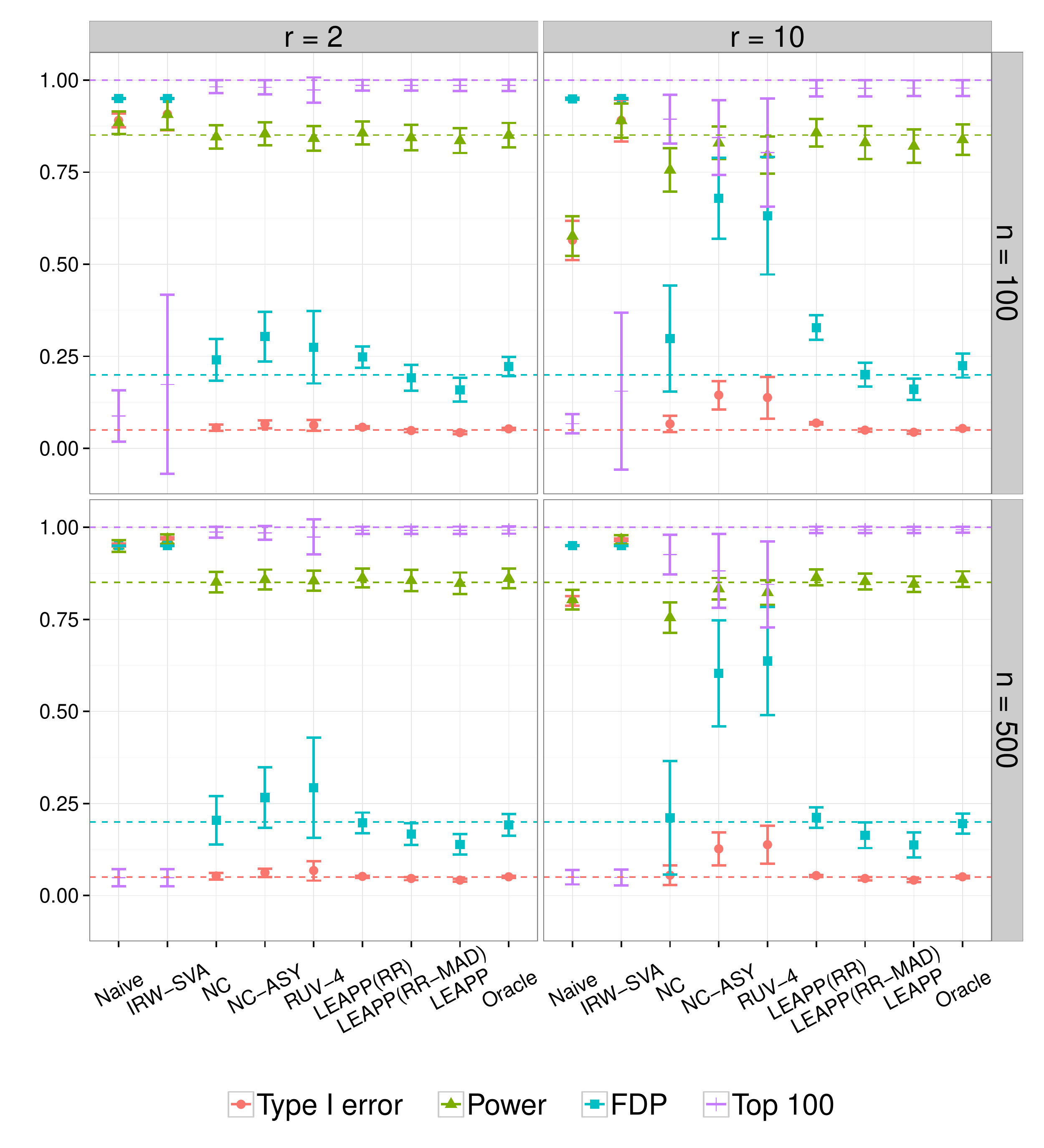}
  \caption{\small Compare the performance of nine different approaches (from
  left to right): naive regression ignoring the confounders (Naive),
  IRW-SVA, negative control with finite sample
  correction (NC) in \cref{eq:nc-var-delta}, negative control with asymptotic oracle variance (NC-ASY) in
  \cref{eq:nc-var-oracle},  RUV-4,  robust regression (LEAPP(RR)),
  robust regression with calibration (LEAPP(RR-MAD)),
   LEAPP, oracle regression which observes the confounders (Oracle).
   The error bars are one
  standard deviation over $100$ repeated simulations. The
  three dashed horizontal lines from bottom to top are the nominal significance level, FDR level and oracle power, respectively.}
  \label{fig:p-5000-methods}
\end{figure}


First, from \Cref{fig:p-5000-methods}, we see that the oracle method has exactly the same type I error
and FDP as specified, while the naive method and SVA fail drastically.
SVA performs performs better than the naive method in terms of the precision of the smallest $100$
p-values, but is still much worse than other methods.
Next, for the
negative control scenario, as we only have $|\mathcal{C}| = 30$ negative controls, ignoring the inflated variance term
$\bm{\Delta}_S$ in \Cref{thm:asymptotics-nc} will lead to
overdispersed test statistics, and that's why
the type I error and FDP of both NC-ASY and RUV-4 are much larger than the
nominal level. 
By contrast, the NC method correctly controls type I error and FDP
by considering the variance inflation, 
though as expected it loses some power compared with the oracle.
For the sparsity scenario, the ``LEAPP(RR)'' method performs as the
asymptotic theory predicted when $n = 500$,
while when $n = 100$ the p-values seem a bit too small.
This is not surprising because the
asymptotic oracle variance in \Cref{thm:asymptotics-rr} can be
optimistic when the sample size is not sufficiently large, as we discussed in \Cref{rmk:calibration}.
On the other hand, the methods which use empirical calibration
for the variance of test statistics, namely the original LEAPP and ``LEAPP(RR-MAD)'',  control
both FDP and type I error for data of small sample size in our simulations.
The price for the finite sample calibration is that it tends to be slightly
conservative, resulting in a loss of power to some extent.

In conclusion, the simulation results
are consistent with our theoretical guarantees when $p$ is as large as $5000$
and $n$ is as large as $500$. When $n$ is small, the variance of the test statistics will
be larger than the asymptotic variance for the sparsity scenario and we can use empirical calibrations (such as MAD)
to adjust for the difference.

\subsection{Real data examples}
\label{sec:gender-study}

In this section, we return to the three motivating real data examples in
\Cref{sec:introduction}. The main goal here is to demonstrate a practical
procedure for confounder adjustment and show that our asymptotic results
are reasonably accurate in real data. In an open-source R package \texttt{cate}
(available on CRAN), we also provide the necessary tools to carry
out the procedure.

\subsubsection{The datasets}
\label{sec:datasets}

First we briefly describe the three datasets. The first dataset
\citep{singh2011} is tries to identify candidate genes associated with
the extent of emphysema and can be downloaded from the GEO database
(Series GSE22148). We preprocessed the data using the standard
Robust Multi-array Average (RMA) approach \citep{irizarry2003}. The primary variable of interest is
the severity (moderate or severe) of the Chronic Obstructive Pulmonary Disease (COPD). The dataset also
include age, gender, batch and date of the $143$ sampled patients which are served as nuisance covariates.

The second and third datasets are taken from \citet{gagnon2013} where they used them to
compare RUV methods with other methods
such as SVA and LEAPP. The original scientific studies are
\citet{vawter2004} and \citet{blalock2004}, respectively.
The primary variable of interest is gender in both datasets, though
the original objective in \citet{blalock2004} is to identify genes
associated with Alzheimer's disease. \citet{gagnon2013} switch the
primary variable to gender in order to have a gold standard: the
differentially expressed genes should mostly come from or relate to
the X or Y chromosome. We follow their suggestion and use this
standard to study the performance of our RR estimator. In addition, as the first COPD dataset
also contains gender information of the samples, we apply this suggestion and
use gender as the primary variable for the COPD data as a supplementary dataset.

Finally, we want to mention that the second dataset has repeated samples from the same
individuals while the individual information is lost. We
suspect that the individual information are then strong latent factors
which caused the atypical
concentration of the histograms in \Cref{fig:gender-naive} and \Cref{fig:gender-batch}.
This suggests necessity of a latent factor model for this dataset.

\subsubsection{Confounder adjustment}
\label{sec:conf-adjustm}

Recall that without the confounder adjustment,
the distribution of the regression $t$-statistics in these datasets
can be skewed, noncentered, underdispersed, or overdispersed as shown
in \Cref{fig:confounding-problem}. The adjustment method used here is the
maximum likelihood factor analysis described in
\Cref{sec:inference-gamma-sigma} followed by the robust regression
(RR) method with Tukey's bisquare loss described in
\Cref{sec:unkn-zero-indic}. Since the true number of confounders is unknown, we increase
$r$ from $1$ to $n/2$ and study the empirical performance. We report the results without
empirical calibration for illustrative purposes, though in practice we suggest using
calibration for better control of type I errors and FDP.

\setlength{\tabcolsep}{3pt}
\begin{table}[hp]
\centering
\begin{subtable}[t]{\textwidth}
\centering
\footnotesize
\begin{tabular}{|$r|^r^r^r^r^r^r|^r^r|}
  \hline
r & mean & median & sd & mad & skewness & medc. & \#sig. & p-value \\
  \hline
0 & -0.16 & 0.024 & 2.65 & 2.57 & -0.104 & -0.091 & 164 &  NA\\
  1 & -0.45 & -0.39 & 2.85 & 2.52 & -0.25 & 0.00074 & 1162 & 0.0057 \\
  2 & 0.012 & -0.039 & 1.35 & 1.33 & 0.139 & 0.042 & 542 & $<$1e-10 \\
  3 & 0.014 & -0.05 & 1.43 & 1.41 & 0.169 & 0.048 & 552 & $<$1e-10 \\
  5 & -0.029 & -0.11 & 1.52 & 1.48 & 0.236 & 0.057 & 647 & $<$1e-10 \\
  7 & -0.1 & -0.14 & 1.42 & 1.35 & 0.109 & 0.027 & 837 & $<$1e-10 \\
  10 & -0.06 & -0.085 & 1.13 & 1.12 & 0.103 & 0.022 & 506 & $<$1e-10 \\
  20 & -0.083 & -0.095 &  1.2 & 1.19 & 0.0604 & 0.0095 & 479 & $<$1e-10 \\
  \rowstyle{\bfseries}
  33 & -0.099 & -0.11 & 1.33 &  1.3 & 0.0727 & 0.0056 & 579 & $<$1e-10 \\
  40 & -0.1 & -0.12 & 1.43 &  1.4 & 0.0775 & 0.0072 & 585 & $<$1e-10 \\
  50 & -0.16 & -0.17 & 1.58 & 1.53 & 0.0528 & 0.0032 & 678 & $<$1e-10 \\
   \hline
\end{tabular}
\caption{\small Dataset 1 ($n = 143$, $p = 54675$). Primary variable:
  severity of COPD.}
\label{tab:COPD}
\end{subtable}

\begin{subtable}[t]{\textwidth}
\centering
\footnotesize

\begin{tabular}{|$r|^r^r^r^r^r^r|^r^r^r^r|}
  \hline
r & mean & median & sd & mad & skewness & medc. & \#sig. & X/Y & top 100 & p-value \\
  \hline
0 & 0.11 & 0.043 & 0.36 & 0.237 & 2.99 & 0.2 & 1036 & 58 & 11 &  NA\\
  1 & -0.44 & -0.47 & 1.06 & 1.04 & 0.688 & 0.035 & 108 & 20 & 20 & 0.74 \\
  2 & -0.14 & -0.15 & 1.15 & 1.13 & 0.601 & 0.015 & 113 & 21 & 21 & 0.31 \\
  3 & 0.013 & 0.012 & 1.13 & 1.08 & 0.795 & -0.01 & 168 & 34 & 28 & 0.03 \\
  5 & 0.044 & 0.019 & 1.18 & 1.08 & 0.878 & 0.017 & 238 & 32 & 27 & 0.0083 \\
  7 & 0.03 & 0.012 & 1.26 & 1.15 & 0.784 & 0.0062 & 269 & 35 & 25 & 0.006 \\
  10 & 0.023 & 0.00066 & 1.36 & 1.24 & 0.661 & 0.011 & 270 & 38 & 27 & 0.019 \\
  15 & 0.049 & 0.022 & 1.46 & 1.31 & 0.584 & 0.012 & 296 & 36 & 29 & 0.00082 \\
  20 & 0.029 & -0.0009 & 1.53 & 1.36 & 0.502 & 0.019 & 314 & 36 & 28 & 7.2e-07
  \\
  \rowstyle{\bfseries}
  25 & 0.048 & 0.012 & 1.68 & 1.48 & 0.452 & 0.026 & 354 & 37 & 27 & 1.1e-06 \\
  30 & 0.026 & 0.012 & 1.82 & 1.61 & 0.436 & 0.0068 & 337 & 40 & 27 & 8.7e-08 \\
  40 & 0.061 & 0.046 & 2.07 & 1.79 & 0.642 & 0.0028 & 363 & 41 & 27 & 7.7e-10 \\
   \hline
\end{tabular}
%
\caption{\small Dataset 2 ($n = 84$, $p = 12600$). Primary variable: gender.}
\label{tab:gender}
\end{subtable}

\begin{subtable}[t]{\textwidth}
\centering
\footnotesize
\begin{tabular}{|$r|^r^r^r^r^r^r|^r^r^r^r|}
  \hline
r & mean & median & sd & mad & skewness & medc. & \#sig. & X/Y & top 100 & p-value \\
  \hline
0 & -1.8 & -1.8 & 0.599 & 0.513 & -3.46 & 0.082 & 418 & 39 & 20 &  NA\\
  1 & -0.55 & -0.56 & 1.09 & 1.01 & -1.53 & 0.01 & 261 & 29 & 23 & 0.00024 \\
  2 & -0.2 & -0.22 &  1.2 & 1.11 & -0.99 & 0.014 & 320 & 38 & 22 & 0.00014 \\
  3 & -0.096 & -0.12 & 1.27 & 1.18 & -0.844 & 0.017 & 311 & 42 & 25 & 0.00014 \\
  5 & -0.33 & -0.32 & 1.31 & 1.22 & -1.29 & -0.011 & 305 & 35 & 23 & 2.1e-07 \\
  7 & -0.37 & -0.36 & 1.46 & 1.36 & -0.855 & -0.0099 & 300 & 38 & 23 & 4.0e-07 \\
  \rowstyle{\bfseries}
  11 & -0.13 & -0.12 & 1.51 & 1.36 & -0.601 & -0.0051 & 432 & 48 & 31 & 1.8e-09 \\
  15 & -0.12 & -0.13 & 1.83 & 1.62 & -0.341 & 0.013 & 492 & 54 & 25 & 2.3e-08 \\
  20 & -0.13 & -0.14 & 2.61 & 2.23 & -0.327 & 0.0045 & 613 & 50 & 26 & 4.0e-06 \\
   \hline
\end{tabular}
\caption{\small Dataset 3 ($n = 31$, $p = 22283$). Primary variable: gender.}
\label{tab:alzheimers}
\end{subtable}
\caption{\small Summary of the adjusted z-statistics. The first group is
  summary statistics of the z-statistics before the empirical
  calibration. The second group is some performance metrics after the
  empirical calibration, including total number of significant genes
  of p-value less than $0.01$ in
  \Cref{rmk:calibration} (\#sig.), number of the genes on X/Y
  chromosome that have p-value less than $0.01$ (X/Y), the number
  among the $100$ most significant genes that are on the X/Y
chromosome (top 100) and the p-value of the confounding test in \Cref{sec:test-confound}.
The bold row corresponds to the $r$ selected by BCV (\Cref{fig:output}).}
\label{tab:output}
\end{table}

\begin{figure}[hp]
  \centering
  \begin{subfigure}[b]{0.42\textwidth}
  \includegraphics[width = \textwidth]{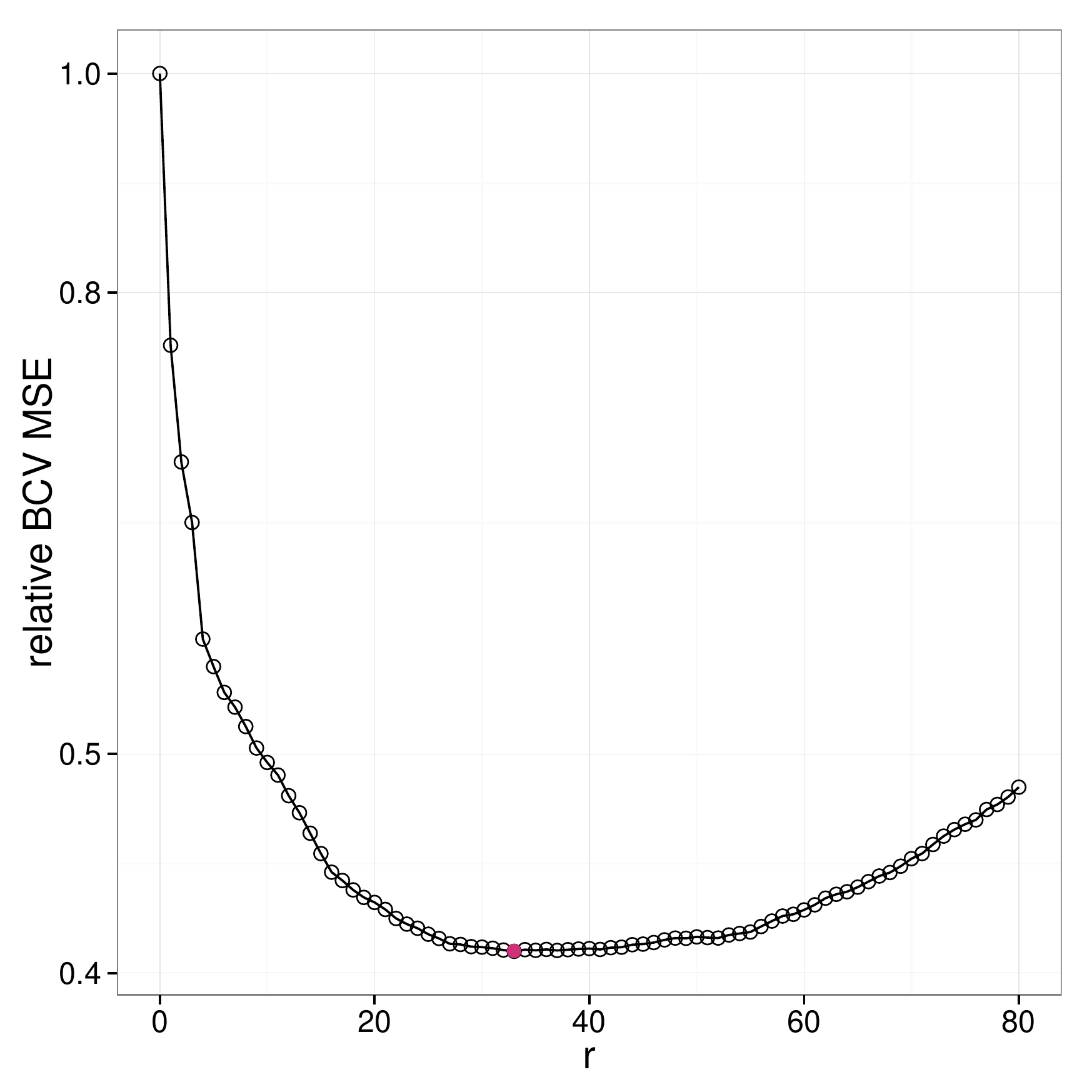}
  \caption{\small Dataset 1: BCV selects $r=33$.}
  \label{fig:COPD-bcv}
  \end{subfigure}
  ~
  \begin{subfigure}[b]{0.42\textwidth}
  \includegraphics[width = \textwidth]{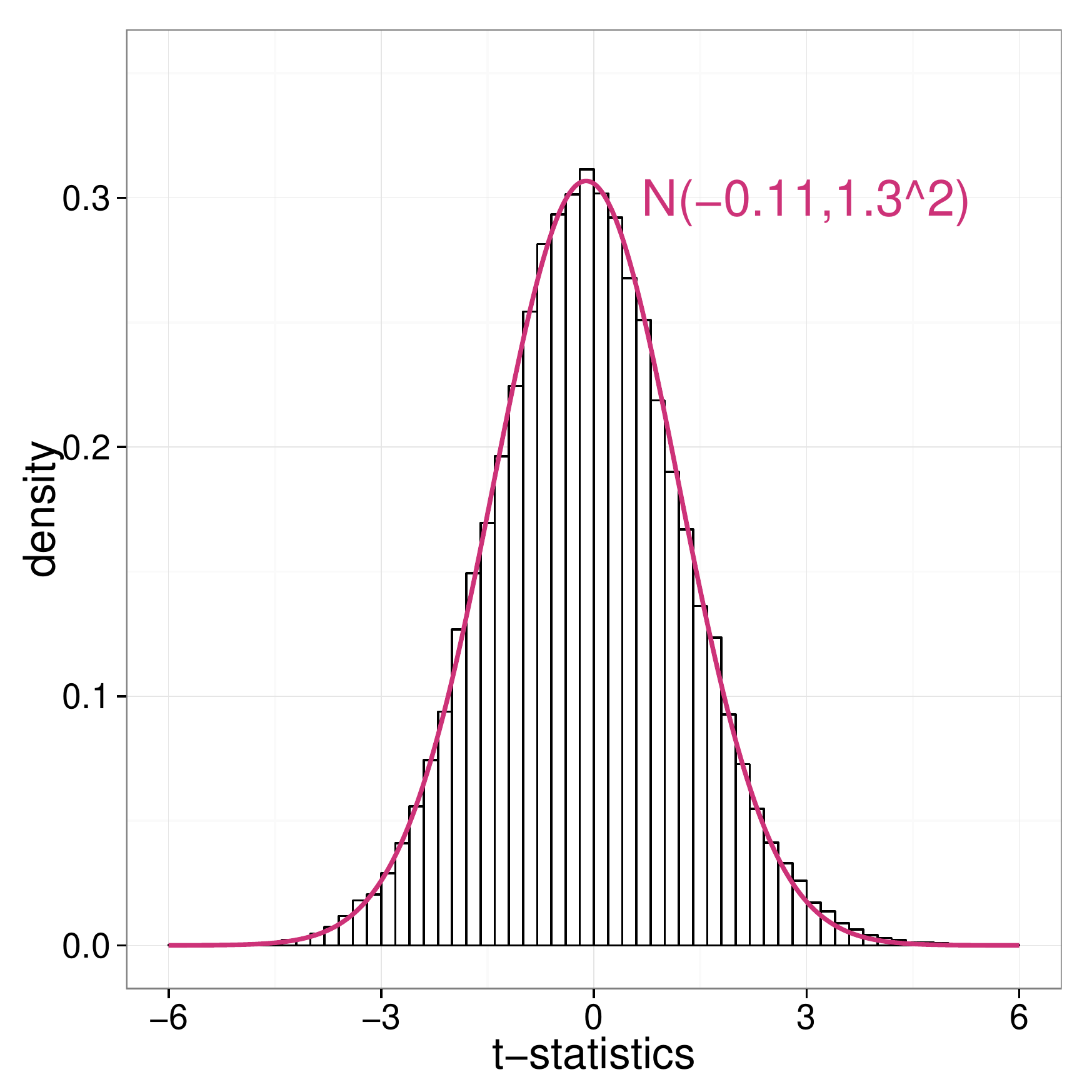}
  \caption{\small Dataset 1: histogram. }
  \label{fig:COPD-t}
  \end{subfigure} \\

  \begin{subfigure}[b]{0.42\textwidth}
  \includegraphics[width = \textwidth]{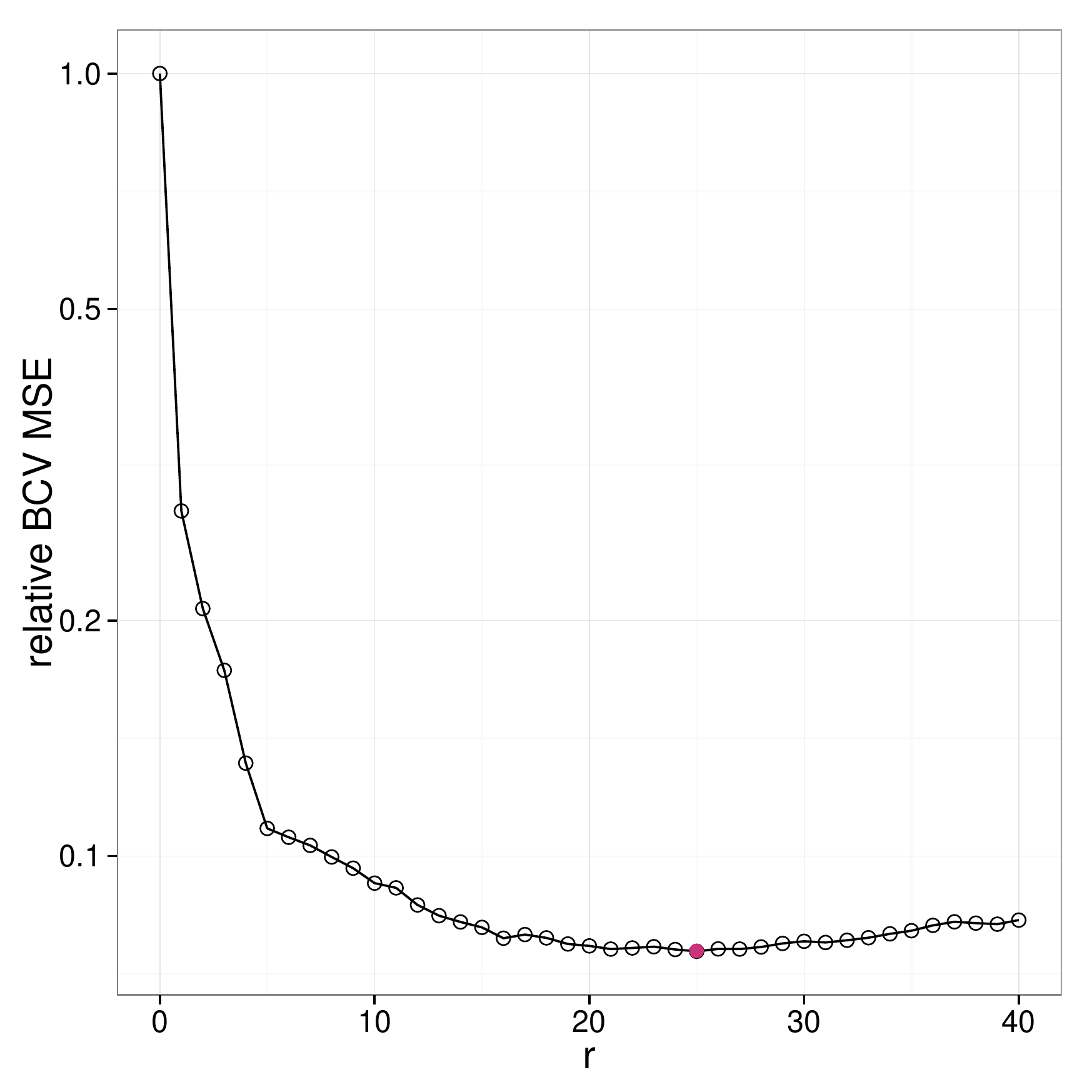}
  \caption{\small Dataset 2: BCV selects $r=25$.}
  \label{fig:gender-bcv}
  \end{subfigure}
  ~
  \begin{subfigure}[b]{0.42\textwidth}
  \includegraphics[width = \textwidth]{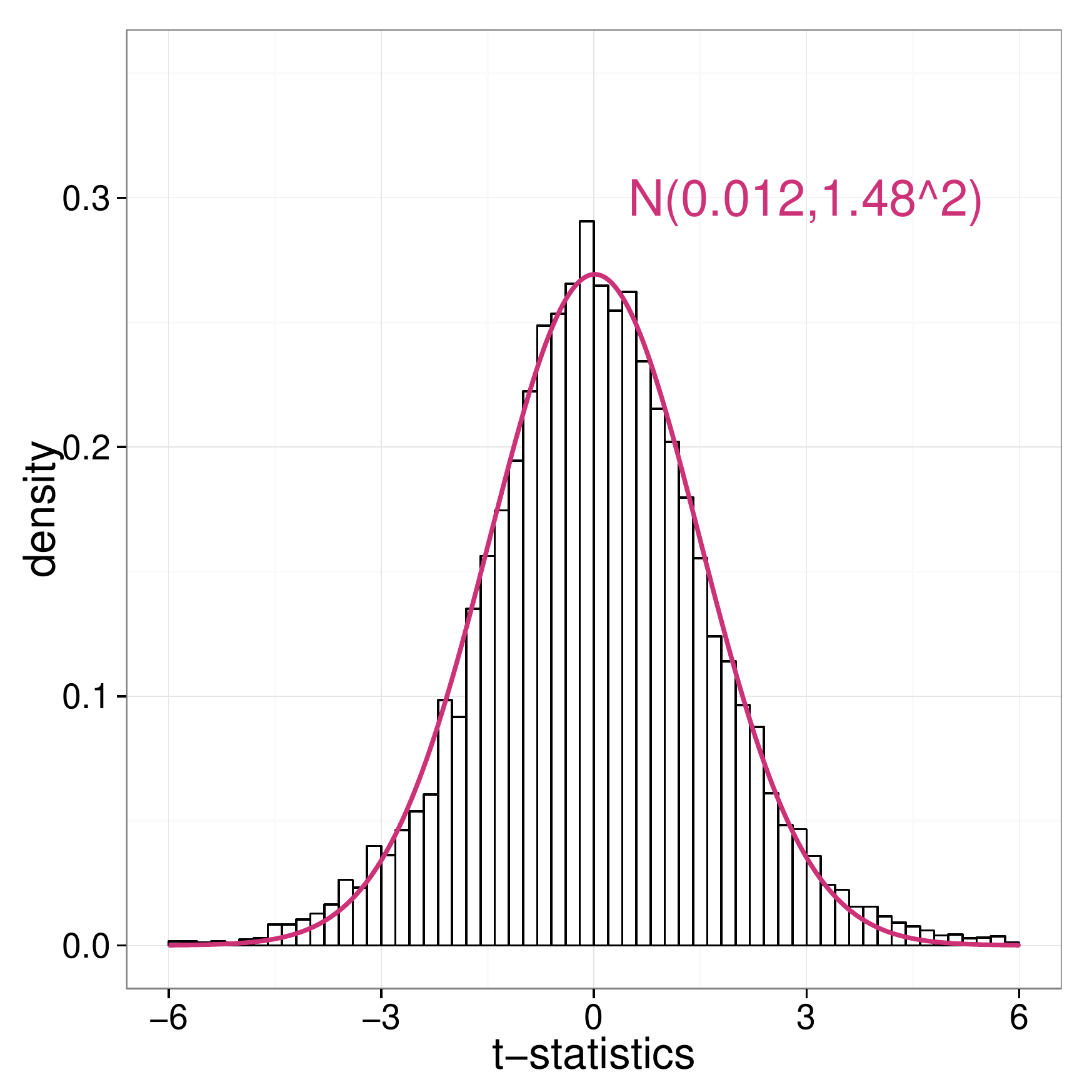}
  \caption{\small Dataset 2: histogram. }
  \label{fig:gender-t}
  \end{subfigure} \\

  \begin{subfigure}[b]{0.42\textwidth}
  \includegraphics[width = \textwidth]{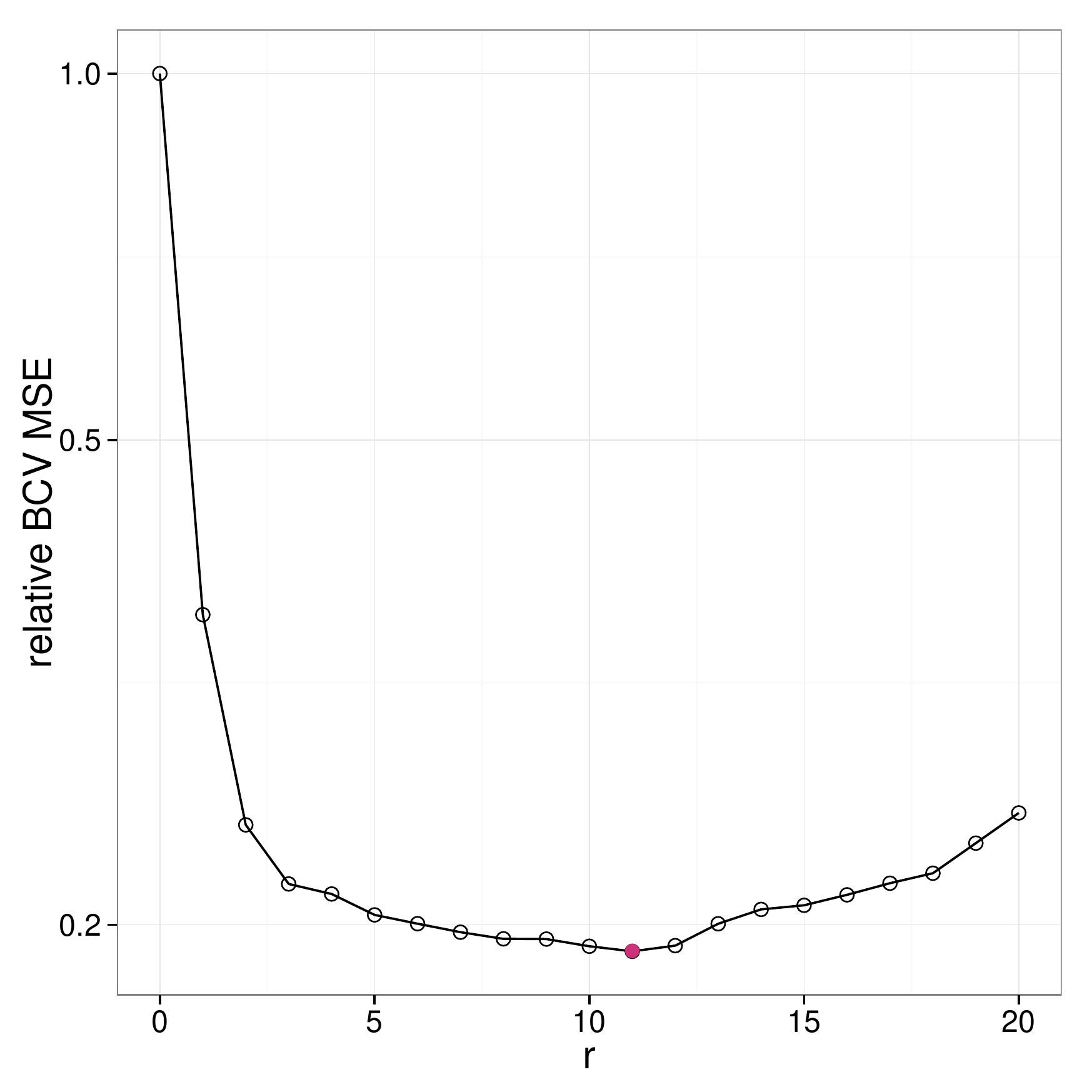}
  \caption{\small Dataset 3: BCV selects $r=11$.}
  \label{fig:alzheimers-bcv}
  \end{subfigure}
  ~
  \begin{subfigure}[b]{0.42\textwidth}
  \includegraphics[width = \textwidth]{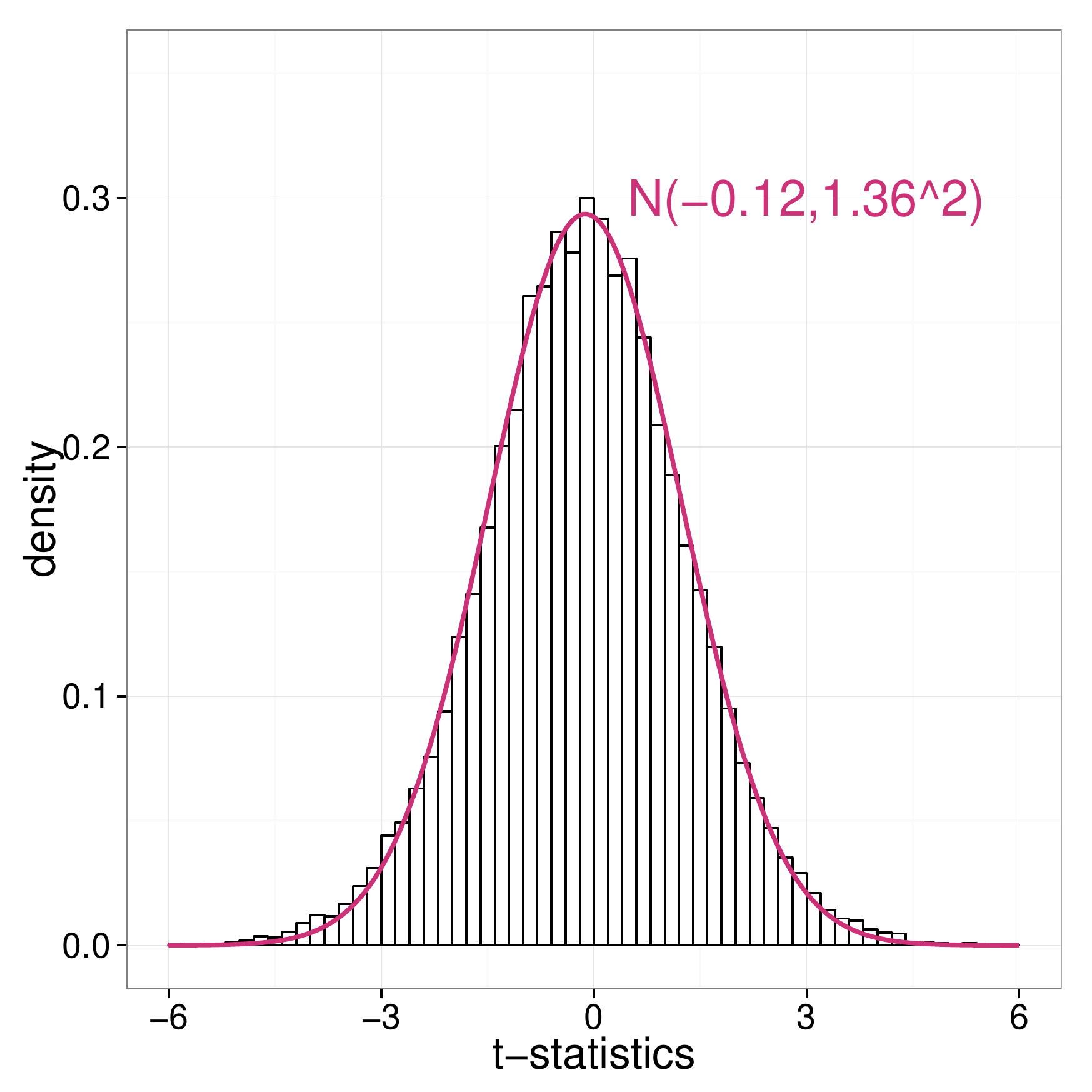}
  \caption{\small Dataset 3: histogram.}
  \label{fig:alzheimers-t}
  \end{subfigure}

  \caption{\small Histograms of z-statistics after confounder adjustment (without calibration)
    using the number of confounders $r$ selected by
    bi-cross-validation.}
  \label{fig:output}

\end{figure}

In \Cref{tab:output} and \Cref{fig:output}, we present the results after
confounder adjustment for the three datasets. 
We report
two groups of summary statistics in \Cref{tab:output}: the first group is several
summary statistics of all the z-statistics computed using
\cref{eq:test-statistic}, 
including the mean, median, standard
deviation, median absolute deviation (scaled for consistency
of normal distribution), skewness, and
the medcouple.
The medcouple \citep{brys2004}) is a robust measure of skewness.
After subtracting the median observation some positive and some
negative values remain.  For any pair of values $x_1\ge 0$ and $x_2\le0$
with $x_1+|x_2|>0$ one can compute $(x_1-|x_2|)/(x_1 +|x_2|)$. The medcouple is the median
of all those ratios.
The second group of statistics has performance metrics to
evaluate the effectiveness of the confounder adjustment. See
the caption of \Cref{tab:output} for more detail.

In all three datasets, the z-statistics become more centered at $0$ and
less skewed as we include a few confounders in the model.
Though the standard deviation (SD) suggests overdispersed variance,
the overdispersion will go away if
we add MAD calibration as SD and MAD have similar values. The
similarity between SD and MAD values also indicates that
the majority of statistics after confounder adjustment are
approximately normally distributed.
Note that 
the medcouple values shrink towards zero after adjustment,
suggesting that skewness then only arises from small fraction of the genes,
which is in accordance with our assumptions that the primary effects should be sparse.

In practice, some latent factors may be too weak to meet
\Cref{assumption:large-factor} (i.e.\ $d_j \ll \sqrt{p}$) , making it
difficult to choose an appropriate $r$.
A practical way to pick the number of confounders $r$
with presence of heteroscedastic noise we investigate here is the
bi-cross-validation (BCV) method of \citet{owen2015}, which uses randomly held-out submatrices to estimate
the mean squared error of reconstructing factor loading matrix. It is
shown in \citet{owen2015} that BCV outperforms many existing methods
in recovering the latent signal matrix and the number of factors $r$, especially in high-dimensional datasets ($n, p\to \infty$).
In \Cref{fig:output},
we demonstrate the performance of BCV on these three datasets. The $r$
selected by BCV is respectively $33$, $25$ and $11$
(\Cref{fig:COPD-bcv,fig:gender-bcv,fig:alzheimers-bcv}), and they all
result in the presumed shape of z-statistics distribution
(\Cref{fig:COPD-t,fig:gender-t,fig:alzheimers-t}). For the second and
the third datasets where we have a gold standard, the $r$ selected
by BCV has near optimal performance in selecting genes on the X/Y chromosome
(columns 3 and 4 in \Cref{tab:gender,tab:alzheimers}).
Another method we applied is proposed by
\citet{onatski2010determining} based on the empirical distribution of
eigenvalues. This method estimates $r$ as $2$, $9$ and $3$ respectively for the three datasets.
Table 3 of \citet{gagnon2013} has the ``top 100'' values for RUV-4 on the second and third dataset.
They reported 26 for LEAPP, 28 for RUV-4, and 27 for SVA in the second dataset,
and 27 for LEAPP, 31 for RUV-4, and 26 for SVA in the third dataset.
Notice that the precision of the top $100$ significant genes is relatively stable when $r$
is above certain number. 
Intuitively, the factor analysis is applied to
the residuals of $\bm{Y}$ on $\bm{X}$ and the overestimated factors
also have very small eigenvalues, thus they usually do not change
$\hat{\bm{\beta}}$ a lot. See also
\citet{gagnon2013} for more discussion on the robustness of the
negative control estimator to overestimating $r$.

Lastly we want to point out that both the small sample size of the datasets
and presence of weak factors can result in overdispersed variance of the test statistics.
The BCV plots indicate presence of many weak factors in the first two datasets.
In the third
dataset, the sample size $n$ is only $31$, so the adjustment result is
not ideal. Nevertheless, the empirical
performance (e.g.\ number of X/Y genes in top $100$) suggests it
is still beneficial to adjust for the confounders.




\appendix

\section{Proofs}
\label{sec:proofs}

\subsection{More technical results of factor analysis}
\label{app:lem-ml}


Here we prove uniform convergence of the estimated factors and noise
variances based on the results of \citet{bai2012},
which are needed to prove
\Cref{thm:asymptotics-nc,thm:consistency-rr,thm:asymptotics-rr,cor:type-I}. In
the proof of the following lemma, we intensively
use some of the technical results in \citet{bai2012} and also modify
internal parts of their proof. Before reading the proof of
\Cref{lem:gamma-ml}, we recommend that the reader first
read the original proof in \citet{bai2012,bai2012supplement}. To help the readers to follow,
the variables $N$, $T$, $\Lambda$ (or $\Lambda^\star$) and $f$ (or $f^\star$)
in \citet{bai2012}
correspond to $p$, $n$, $\bm \Gamma^{(0)}$ and
$\tilde{\bm Z}^{(0)}$ in our notation. 

\begin{lemma} \label{lem:gamma-ml}
  Under \Cref{assumption:error,assumption:bounded,assumption:large-factor},
for any fixed index set
 $S$ with finite cardinality,
\begin{equation}\label{eq:gamma-asym}
\sqrt{n} (\hat{\bm{\Gamma}}_S - \bm{\Gamma}_S^{(0)}) \overset{d}{\to}
\mathrm{MN}(0, \bm{\Sigma}_S, \bm{I}_r)
\end{equation}
where $\bm{\Sigma}_S$ is the noise covariance matrix of the variables in $S$.
Further, if there exists $k > 0$ such that $p/n^{k} \to 0$ when $p \to \infty$, then
\begin{equation}\label{eq:max-sigma-err}
\max_{1 \le j \le p}|\hat{\sigma_j^2} - \sigma_j^2| =
O_p(\sqrt{\log p / n}),~
\max_{1 \le j \le p}|\hat{\bm \Gamma}_j - \bm\Gamma_j^{(0)}| = \bm O_p(\sqrt{\log p / n}), ~\mathrm{and}
\end{equation}
\begin{equation}\label{eq:max-gamma-err}
\max_{j = 1, 2,\cdots,p} \Big|\hat{\bm{\Gamma}}_j - \bm{\Gamma}_j^{(0)}
-\frac{1}{n-1}\sum_{i=2}^n \tilde{\bm{Z}}_i^{(0)} \tilde{E}_{ij}\Big|
 = \bm{o}_p(n^{-\frac{1}{2}}).
\end{equation}
\end{lemma}

  \begin{remark} \label{rmk:rotation}
If we directly apply the results in \citet{bai2012} to prove \Cref{lem:ml-baili}, we
need uniform boundedness of $\bm{\Gamma}^{(0)}$ which is not always true. However, it is easy to show $\bm{R} \overset{a.s.}{\to}
\bm{I}_r$ by applying \citet[Lemma A.1]{bai2012}. Also, as $\bm{R}\bm{R}^T$ is the sample
covariance matrix, the maximum entry of
$|\bm{R} - \bm{I}|$ is ${O}_p(n^{-1/2})$,
thus the maximum entry of
$|\bm{\Gamma}^{(0)} - \bm{\Gamma}| = |\bm{\Gamma}(\bm R - \bm I)|$ is also ${O}_p(n^{-1/2})$. As a
consequence, although $\bm{\Gamma}^{(0)}$ is not always uniformly bounded, all the results in
\citet{bai2012} still hold as we stated in \Cref{lem:ml-baili} and \Cref{lem:gamma-ml}.
\end{remark}

\begin{proof}

Our factor model corresponds to the IC3 identification condition in
\citet{bai2012}. Equation \eqref{eq:gamma-asym}
is an immediate consequence of \citet[Theorem 5.2]{bai2012}, except here we additionally consider the asymptotic
  covariance of $\sqrt{n} (\hat{\bm{\Gamma}}_j - \bm{\Gamma}_j^{(0)})$
  and $\sqrt{n} (\hat{\bm{\Gamma}}_k - \bm{\Gamma}_k^{(0)})$.
  The asymptotic distribution of $\sqrt{n} (\hat{\bm{\Gamma}}_S -
  \bm{\Gamma}^{(0)}_S)$ immediately follows from equation (F.1) in
  \citet{bai2012supplement}:
\begin{equation} \label{eq:gamma-hat}
\sqrt{n-1} (\hat{\bm{\Gamma}}_j - \bm{\Gamma}_j^{(0)}) =
\frac{1}{\sqrt{n-1}} \sum_{i=2}^n \tilde{\bm{Z}}_i^{(0)} \tilde{E}_{ij} + \bm{o}_p(1).
\end{equation}

Now we prove \eqref{eq:max-sigma-err}.
Let $\hat {\bm{\Gamma}}_j - \bm{\Gamma}_j^{(0)} = \bm{b}_{1j} + \bm{b}_{2j} +\cdots + \bm{b}_{10,j}$
where
$\bm{b}_{kj}$ represents the $k$th term in the right hand side of
equation (A.14) in \citet{bai2012}.
Also, let  $\hat \sigma_j^2 - \sigma_j^2 = a_{1j} + a_{2j} +\cdots + a_{10,j}$ where
$a_{kj}$ represents the $k$th term in the right hand side of equation (B.9) in \citet{bai2012supplement}.
To bound each $\bm b_j$ and $a_j$ term, we extensively use Lemma C.1 of \citet{bai2012supplement}.
First, we give a clearer approximation to replace $(a)$ and $(c)$ in Lemma C.1 of \citet{bai2012supplement}:
\begin{equation}\label{eq:baiC1a}
\|\hat{\bm{H}}
\hat{\bm{\Gamma}}^T\hat{\bm{\Sigma}}^{-1}
(\hat{\bm{\Gamma}} - \bm{\Gamma}^{(0)})\|_F = O_p(n^{-1}) +O_p(n^{-1/2}p^{-1/2})
\end{equation}
and
\begin{equation}\label{eq:baiC1e}
\frac{1}{n-1}\big\|\hat{\bm{H}}
\hat{\bm{\Gamma}}^T\hat{\bm{\Sigma}}^{-1}
\tilde{\bm{E}}_{-1}^T\tilde{\bm{Z}}^{(0)}\big\|_F = O_p(n^{-1/2}p^{-1/2}) + O_p(n^{-1})
\end{equation}
where $\hat{\bm{H}} = (\hat{\bm{\Gamma}}^T\hat{\bm{\Sigma}}^{-1}\hat{\bm{\Gamma}})^{-1}$
and $\|\cdot\|_F$ is the Frobenius norm.
To show \eqref{eq:baiC1e}, one just needs to apply
$\hat {\bm H}_p = p\hat{\bm H} = \bm O_p(1)$ \citep[Corollary A.1]{bai2012},
\Cref{rmk:rotation} and $(n-1)^{-1}(\tilde {\bm Z}^{(0)}_{-1})^T\tilde {\bm Z}^{(0)}_{-1} = \bm I_r$
to simplify Lemma C.1(e) of \citet{bai2012supplement}. To prove \eqref{eq:baiC1a}, notice that under our conditions
(or the IC3 condition of \citet{bai2012}), the left hand side of (A.13) in
\citet{bai2012} is actually $\bm 0$ as the terms $\hat M_{ff}$ and $M_{ff}^\star$ in their notation are exactly $\bm I_r$.
Also, $\hat{\bm H}\hat{\bm \Gamma}^{(0)T}\hat {\bm \Sigma}^{-1}
 {\bm \Gamma}= \bm I_r + \bm o_p(1)$ from \citet[Corollary A.1]{bai2012}.
 Thus, \eqref{eq:baiC1a} holds by applying Lemma C.1 of \citet{bai2012supplement}.
 As a consequence, by applying \citet[Lemma C.1]{bai2012supplement}, \eqref{eq:baiC1a} and \eqref{eq:baiC1e},
 we now have $\max_j|\bm b_{kj}| = \bm o_p(n^{-1/2})$ for $k \neq 8, 10$ and
$\max_j|a_{kj}| = o_p(n^{-1/2})$ for $k \neq 1, 2, 8, 9, 10$. Using independence of the noise,
it's also easy to see that $\max_j|\bm b_{8j}| = O_p(\sqrt{\log p / n})$ and
$\max_j|a_{kj}| = O_p(\sqrt{\log p / n})$ for $k = 1, 10$.

Next, we show the following facts under the condition that $p/n^k \to 0$ when $p \to \infty$ for some $k > 0$.
Let
$(e_{ti})_{(n-1)\times p} = \tilde{\bm E}\bm{\Sigma}^{-1/2}$ denote a
random matrix whose entries are then i.i.d.\
$\mathrm{N}(0,1)$ variables. Then for each
$s = 1, 2, \cdots, r$,
\begin{equation}\label{eq:noise-fact}
  \max_{j = 1,2, \cdots, p}\frac{1}{(n-1)p}
  \Big|\sum_{i = 1}^p\Gamma_{is}\sum_{t = 1}^{n-1}
  [e_{ti}e_{tj} - \mathrm{E}(e_{ti}e_{tj})]\Big| = o_p(n^{-1/2}),~\mathrm{and}
\end{equation}
\begin{equation}\label{eq:noise-fact2}
  \max_{j = 1,2, \cdots, p}\frac{1}{(n-1)^2p}
  \sum_{i = 1}^p\Big(\sum_{t = 1}^{n-1}
  [e_{ti}e_{tj} - \mathrm{E}(e_{ti}e_{tj})]\Big)^2 = o_p(n^{-1/2}).
\end{equation}
To prove \eqref{eq:noise-fact},
we only need to show
$\max_j\frac{1}{(n-1)p}
  \big|\sum_{i \neq j}\sum_{t = 1}^{n-1}
  \Gamma_{is}e_{ti}e_{tj}\big| = o_p(n^{-1/2})
$ as the remaining term is $o_p(n^{-1/2})$ because of the
independence. This approximation is proven by
the union bound and  boundedness of $\bm{\Gamma}$:
for $\forall \epsilon > 0$
\begin{equation*}
\begin{split}
 &      \lim_{n,p\to \infty}\mathrm{P}\Big(\sqrt n\max_{j = 1,2,\cdots, p}\frac{1}{(n-1)p}
  \big|\sum_{i \neq j}\sum_{t = 1}^{n-1}
  \Gamma_{is}e_{ti}e_{tj}\big| > \epsilon\Big) \\
   \leq & \lim_{n,p\to \infty}2p \cdot \mathrm{P}\Big(\frac{\sqrt n D}{(n-1)p}
  \sum_{i \neq 1}\sum_{t = 1}^{n-1}e_{ti}e_{t1} > \epsilon\Big)\\
  = & \lim_{n,p\to \infty}2p \cdot \mathrm{P}\Bigg(\frac{\sqrt n}{n-1}\sum_{t = 1}^{n-1}e_{t1}
  \big(\frac{1}{\sqrt {p-1}}\sum_{i \neq 1}e_{ti}\big) >
  \frac{\epsilon}{D}\frac{p}{\sqrt{p-1}}\Bigg)\\
  \leq& \lim_{n,p\to \infty}2p \cdot \mathrm{E}\Bigg[\Big(\frac{\sqrt n}{n-1}\sum_{t = 1}^{n-1}e_{t1}
  \big(\frac{1}{\sqrt {p-1}}\sum_{i \neq 1}e_{ti}\big)\Big)^4\Bigg]/
  \Big(\frac{\epsilon}{D}\frac{p}{\sqrt{p-1}}\Big)^4 = 0
\end{split}
\end{equation*}
To see why the last equality holds,
$(p - 1)^{-1/2}\sum_{i \neq 1}e_{ti}\sim \mathrm{N}(0, 1)$
is independent from $e_{t1}$, thus the fourth moment of $(n-1)^{-1/2}\sum_{t = 1}^{n-1}e_{t1}
  \big((p-1)^{-1/2}\sum_{i \neq 1}e_{ti}\big)$ is bounded
  which enables us to use the Markov inequality.
To prove \eqref{eq:noise-fact2}, we start with the same union bound as for \eqref{eq:noise-fact},
\[
\begin{split}
   & \lim_{n,p\to \infty}\mathrm{P}\Big(\max_{j = 1,2,\cdots, p}\frac{1}{(n-1)^2p}
   \sum_{i \neq j}\big(\sum_{t = 1}^{n-1}e_{ti}e_{tj}\big)^2  >
   \epsilon\Big) \\
   \le & \lim_{n,p\to \infty} p \cdot \mathrm{P}\Big(\frac{1}{(n-1)^2p}
   \sum_{i=2}^p\big(\sum_{t = 1}^{n-1}e_{ti}e_{t1}\big)^2  >
   \epsilon\Big) \\
   \le & \lim_{n,p\to \infty} 2p^2 \cdot \mathrm{P}\Big(\frac{1}{n-1}
   \sum_{t = 1}^{n-1}e_{t2}e_{t1}  >
   \sqrt\epsilon\Big)\\
   \le & \lim_{n,p\to \infty} 2p^2 \cdot
   \mathrm{E}\Big[\Big(\frac{1}{n-1}\sum_{t=1}^{n-1}e_{t2}e_{t1}\Big)^{4k}\Big]/\epsilon^{2k}\\
   \le & \lim_{n,p\to \infty} 2C/\epsilon^{2k}\cdot\big(p^2/n^{2k}\big) = 0\\
\end{split}
\]
where $C$ is some positive constant.
The second last inequality is due to Markov inequality and last inequality holds
as $e_{t_2}e_{t1}, t= 1, 2,\cdots n-1$ are independent and have finite moments of any order.
The last limit holds when we assume $p/n^k \to 0$.

Equation \eqref{eq:noise-fact} directly implies that
\begin{equation*}
\max_{j = 1, \cdots, p}\Big|\hat{\bm{H}}\big(
\sum_{i=1}^p\frac{1}{\sigma_i\sigma_j}{\bm{\Gamma}}_j
\frac{1}{n-1}\sum_{t = 2}^n[\tilde E_{ti}\tilde E_{tj} -
\mathrm{E}(\tilde E_{ti}\tilde E_{tj})]\big)\Big|
= \bm{o}_p(n^{-1/2})
\end{equation*}
as $\hat{\bm H} = \bm O_p(p^{-1})$. Using \eqref{eq:noise-fact2} and
$p^{-1}\sum_j\|\hat{\bm \Gamma}_j - \bm\Gamma^{(0)}_j\|_2^2 = O_p(n^{-1})$ from \Cref{lem:ml-baili}, we get by
using the Cauchy-Schwartz inequality:
\begin{equation*}
\max_{j = 1, \cdots, p}\Big|\hat{\bm{H}}\big(
\sum_{i=1}^p\frac{1}{\sigma_i\sigma_j}(\hat {\bm{\Gamma}}_j - \bm\Gamma_j^{(0)})
\frac{1}{n-1}\sum_{t = 2}^n[\tilde E_{ti}\tilde E_{tj} -
\mathrm{E}(\tilde E_{ti}\tilde E_{tj})]\big)\Big|
= \bm{o}_p(n^{-1})
\end{equation*}
Similarly, combining with \Cref{rmk:rotation}, we get
\begin{equation*}
\max_{j = 1, \cdots, p}\Big|\hat{\bm{H}}\big(
\sum_{i=1}^p\frac{1}{\sigma_i\sigma_j}({\bm{\Gamma}}_j^{(0)} - \bm\Gamma_j)
\frac{1}{n-1}\sum_{t = 2}^n[\tilde E_{ti}\tilde E_{tj} -
\mathrm{E}(\tilde E_{ti}\tilde E_{tj})]\big)\Big|
= \bm{o}_p(n^{-1})
\end{equation*}
By writing $\hat{\bm{\Gamma}}_j = \bm{\Gamma}_j + \bm{\Gamma}_j^{(0)} - \bm{\Gamma}_j
+ \hat{\bm{\Gamma}}_j - \bm{\Gamma}_j^{(0)}$ and using boundedness of
both $\hat \sigma_j$ and $\sigma_j$,
\begin{equation}\label{eqn:B4-a}
 \max_{j = 1, \cdots, p}\Big|\hat{\bm{H}}\big(
\sum_{i=1}^p\frac{1}{\hat\sigma_i^2}
\hat{\bm{\Gamma}}_j
\frac{1}{n-1}\sum_{t = 2}^n[\tilde E_{ti}\tilde E_{tj} -
\mathrm{E}(\tilde E_{ti}\tilde E_{tj})]\big)\Big|
= \bm{o}_p(n^{-1/2})
\end{equation}
which indicates that $\max_j|a_{9j}| = o_p(n^{-1/2})$.

To bound the remaining terms, we use the fact that $\max_{j = 1, \cdots, p}|\hat{\bm{\Gamma}}_j| = \bm{O}_p(1)$.
To see this, first notice that because of boundedness of $\hat \sigma_j$ and $\sigma_j$ and the fact that
 $\hat{\bm{H}} = \bm{O}_p(p^{-1})$, we have
 $\max_{j}|\bm{b}_{10,j}| = \bm O_p(p^{-1}\max_j{|\hat{\bm{\Gamma}}_j|})$. Combining the previous results, we have
$\max_{j}|\hat{\bm{\Gamma}}_j - \bm{\Gamma}_j^{(0)}|
 = \bm{O}_p(\sqrt{\log p / n}) +  \bm o_p(\max_j{|\hat{\bm{\Gamma}}_j|})$ which indicates
 that $\max_j|\hat{\bm{\Gamma}}_j| = \bm{O}_p(1)$. Thus,
 $\max_j|a_{8j}| = o_p(\max_j|\hat\sigma_j^2 - \sigma_j^2|)$ is negligible
 and $\max_{j}|\hat{\bm{\Gamma}}_j - \bm{\Gamma}_j^{(0)}|
 = \bm{O}_p(\sqrt{\log p / n}) +  \bm o_p(\max_j|\hat \sigma_j^2 -
 \sigma_j^2|)$. The latter conclusion also indicates that $\max_j
 |a_{2j}| = \bm{O}_p(\sqrt{\log p / n}) +  \bm o_p(\max_j|\hat \sigma_j^2 -
 \sigma_j^2|)$. As a
 consequence, the second claim in \eqref{eq:max-sigma-err} holds.

Finally, To prove \eqref{eq:max-gamma-err},
we actually have already shown that
$\max_j|\hat {\bm{\Gamma}}_j - \bm{\Gamma}_j^{(0)}-\bm{b}_{8j}| = \bm{o}_p(n^{-1/2})$.
Then,
\begin{equation*}
  \begin{split}
    &\max_{j = 1, 2, \cdots, p}\Big|\hat {\bm{\Gamma}}_j - \bm{\Gamma}_j^{(0)} -
  \frac{1}{n-1}\sum_{i=2}^n \tilde{\bm{Z}}_i^{(0)} \tilde{E}_{ij}\Big|\\
  \leq & \max_{j = 1, 2, \cdots, p}\Big|\hat {\bm{\Gamma}}_j -
  \bm{\Gamma}_j^{(0)} - \bm{b}_{8j}\Big| + \max_{j = 1, 2, \cdots, p}\Big|\bm{b}_{8j} -
  \frac{1}{n-1}\sum_{i=2}^n \tilde{\bm{Z}}_i^{(0)} \tilde{E}_{ij}\Big|\\
  \leq & \bm{o}_p(n^{-1/2}) + \|\hat{\bm{H}}
\hat{\bm{\Gamma}}^T\hat{\bm{\Sigma}}^{-1}
(\hat{\bm{\Gamma}} - \bm{\Gamma}^{(0)})\|_F
\max_{j = 1, 2, \cdots, p}
\Big|\frac{1}{n-1}\sum_{i = 2}^n\tilde{\bm{Z}}_{i}^{(0)}\tilde
E_{ij}\Big| \\
=&
\bm{o}_p(n^{-1/2})
\end{split}
\end{equation*}
Thus, \eqref{eq:max-gamma-err} holds.

\end{proof}

\subsection{Proof of \Cref{thm:asymptotics-nc}}
\label{app:thm-nc}

First, note that by the strong law of large numbers $n^{-1/2}\|\bm{X}\|_2 =
\sqrt{n^{-1}\sum_{i=1}^n X_i^2} \overset{a.s.}{\to} 1$, and $$\bm{R}
\bm{R}^T = (n-1)^{-1} \tilde{\bm{Z}}_{-1}^T \tilde{\bm{Z}}_{-1} \overset{a.s.}{\to}
\bm{I}_r.$$ Indeed one can show that $\bm{R} \overset{a.s.}{\to} \bm{I}_r$
by applying \citet[Lemma A.1]{bai2012}. We proceed to prove our theorem by
showing the conclusion holds for any fixed $\bm{u}$ and fixed
sequences $\{\bm{X}^{(n)}\}_{n=1}^{\infty}$ and
$\{\bm{R}^{(n,p)}\}_{n=1,p=1}^{\infty}$ such that
$\|\bm{X}^{(n)}\|_2/\sqrt{n} \to 1$ and $\bm{R}^{(n,p)}
\to \bm{I}_r$ as $n,p \to \infty$. For brevity we
will write $\bm{X}$ and $\bm{R}$
instead of $\bm{X}^{(n)}$ and $\bm{R}^{(n,p)}$ for the rest of this proof.

Plugging \eqref{eq:y-tilde-1-separate} in the estimator
\eqref{eq:alpha-nc} and \eqref{eq:beta-nc}, we obtain
\[
\begin{split}
\sqrt{n} (\hat{\bm{\beta}}^{\mathrm{NC}}_{-\mathcal{C}} -
\bm{\beta}_{-\mathcal{C}}) =& \frac{\sqrt{n}}{\|\bm{X}\|_2}
(\tilde{\bm{E}}^T_{1,-\mathcal{C}} - \hat{\bm{\Gamma}}_{-\mathcal{C}}
(\hat{\bm{\Gamma}}_{\mathcal{C}}^T \hat{\bm{\Sigma}}_\mathcal{C}^{-1}
  \hat{\bm{\Gamma}}_{\mathcal{C}})^{-1} \hat{\bm{\Gamma}}_{\mathcal{C}}^T \hat{\bm{\Sigma}}_\mathcal{C}^{-1} \tilde{\bm{E}}^T_{1,\mathcal{C}}) \\
  +& \sqrt{n} \cdot (\bm{\Gamma}_{-\mathcal{C}}^{(0)}  -
  \hat{\bm{\Gamma}}_{-\mathcal{C}} ) \bm{\alpha}^{(0)} \\
  +& \sqrt{n} \cdot \hat{\bm{\Gamma}}_{-\mathcal{C}} (\hat{\bm{\Gamma}}_{\mathcal{C}}^T \hat{\bm{\Sigma}}_\mathcal{C}^{-1}
  \hat{\bm{\Gamma}}_{\mathcal{C}})^{-1} \hat{\bm{\Gamma}}_{\mathcal{C}}^T
  \hat{\bm{\Sigma}}_\mathcal{C}^{-1} (\hat{\bm{\Gamma}}_{\mathcal{C}}  -  \bm{\Gamma}_{\mathcal{C}}^{(0)}) \bm{\alpha}^{(0)}.
\end{split}
\]
As $n,p \to \infty$,
$\sqrt{n}/\|\bm{X}\|_2 \overset{a.s.}{\to} 1$.
Also, as $p/n^k \to 0$ for some $k > 0$, using
\Cref{lem:gamma-ml} and \Cref{rmk:rotation}, both $\hat{\bm \Sigma}$ and $\hat{\bm \Gamma}$ has
entrywise uniform convergence in probability to $\bm \Sigma$ and $\bm \Gamma$.
Using
\Cref{assumption:nc-factor}, we get
\begin{equation}
\begin{split} \label{eq:nc-facts}
\Big(\frac{1}{|\mathcal{C}|}\hat{\bm{\Gamma}}_{
  \mathcal{C}}^T \hat{\bm{\Sigma}}_\mathcal{C}^{-1} \hat{\bm{\Gamma}}_{\mathcal{C}}\Big)^{-1}=&
 \Big(\frac{1}{|\mathcal{C}|}\bm{\Gamma}_{
  \mathcal{C}}^T \bm{\Sigma}_\mathcal{C}^{-1}\bm{\Gamma}_{\mathcal{C}}\Big)^{-1} + \bm o_p(1)\\
\frac{1}{|\mathcal{C}|}\hat{\bm{\Gamma}}_{\mathcal{C}}^T
  \hat{\bm{\Sigma}}_\mathcal{C}^{-1}\tilde{\bm{E}}^T_{1,\mathcal{C}} =&
 \frac{1}{|\mathcal{C}|}\bm{\Gamma}_{\mathcal{C}}^T
  \bm{\Sigma}_\mathcal{C}^{-1}\tilde{\bm{E}}^T_{1,\mathcal{C}} + \bm o_p(1)\\
\frac{1}{|\mathcal{C}|}\hat{\bm{\Gamma}}_{\mathcal{C}}^T
  \hat{\bm{\Sigma}}_\mathcal{C}^{-1}\big(\sqrt n(\hat{\bm{\Gamma}}_{\mathcal{C}} - \bm{\Gamma}_{\mathcal{C}}^{(0)})\big)
  = &
 \frac{1}{|\mathcal{C}|}\bm{\Gamma}_{\mathcal{C}}^T
  \bm{\Sigma}_\mathcal{C}^{-1}\big(\sqrt n(\hat{\bm{\Gamma}}_{\mathcal{C}} - \bm{\Gamma}_{\mathcal{C}}^{(0)})\big)
  + \bm o_p(1)
  \end{split}
\end{equation}
which implies
\begin{equation}
\label{eq:beta-nc-asymptotics-intermediate}
\begin{split}
\sqrt{n} (\hat{\bm{\beta}}^{\mathrm{NC}}_{S} -
\bm{\beta}_{S}) =& \tilde{\bm{E}}^T_{1,S}  -
\bm{\Gamma}_{S} (\bm{\Gamma}_{\mathcal{C}}^T  \bm{\Sigma}_\mathcal{C}^{-1}
\bm{\Gamma}_{\mathcal{C}})^{-1} \bm{\Gamma}_{\mathcal{C}}^T \bm{\Sigma}_\mathcal{C}^{-1} \tilde{\bm{E}}^T_{1,\mathcal{C}} \\
  +& \sqrt{n} \cdot (\bm{\Gamma}_{S}^{(0)}  -
  \hat{\bm{\Gamma}}_{S} ) \bm{\alpha}^{(0)} \\
  +& \sqrt{n} \cdot \bm{\Gamma}_{S} (\bm{\Gamma}_{\mathcal{C}}^T
  \bm{\Sigma}_\mathcal{C}^{-1}  \bm{\Gamma}_{\mathcal{C}})^{-1} \bm{\Gamma}_{\mathcal{C}}^T \bm{\Sigma}_\mathcal{C}^{-1}
  (\hat{\bm{\Gamma}}_{\mathcal{C}}  -  \bm{\Gamma}_{\mathcal{C}}^{(0)}) \bm{\alpha}^{(0)} + \bm{o}_p(1).
\end{split}
\end{equation}
Note that $\tilde{\bm{E}}_1 \independent \hat{\bm{\Gamma}}$, $\tilde{\bm{E}}_{1,\mathcal{C}}
\independent \tilde{\bm{E}}_{1,S}$, and $\sqrt{n} (\hat{\bm{\Gamma}}_S -
\bm{\Gamma}^{(0)}_S) \overset{d}{\to} \mathrm{N}(0, \bm{\Sigma}_S \otimes
\bm{I}_r)$, the four main terms on the right hand side of
\eqref{eq:beta-nc-asymptotics-intermediate} are (asymptotically)
uncorrelated, so we only need to work out their individual
variances. Since $\tilde{\bm{E}}_1^T \sim \mathrm{N}(\bm{0}, \bm{\Sigma})$,
we have $\tilde{\bm{E}}^T_{1,S} \sim \mathrm{N}(\bm{0}, \bm{\Sigma}_{S})$ and
$\bm{\Gamma}_{S} (\bm{\Gamma}_{\mathcal{C}}^T  \bm{\Sigma}_\mathcal{C}^{-1} \bm{\Gamma}_{\mathcal{C}})^{-1}
\bm{\Gamma}_{\mathcal{C}}^T \bm{\Sigma}_\mathcal{C}^{-1} \tilde{\bm{E}}^T_{1,C} \sim \mathrm{N}(\bm{0}, \bm{\Delta}_S)$.
Similarly, $\sqrt{n} \cdot (\bm{\Gamma}_{S}^{(0)}  -
  \hat{\bm{\Gamma}}_{S} ) \bm{\alpha}^{(0)} \overset{d}{\to}
  \mathrm{N}(\bm{0}, \|\bm{\alpha}\|^2 \bm{\Sigma}_{S})$, and
  $$\sqrt{n} \cdot \bm{\Gamma}_{S} (\bm{\Gamma}_{\mathcal{C}}^T \bm{\Sigma}_\mathcal{C}^{-1}
  \bm{\Gamma}_{\mathcal{C}})^{-1} \bm{\Gamma}_{\mathcal{C}}^T \bm{\Sigma}_\mathcal{C}^{-1}
  (\hat{\bm{\Gamma}}_{\mathcal{C}}  -  \bm{\Gamma}_{\mathcal{C}}^{(0)}) \bm{\alpha}^{(0)}
  \overset{d}{\to} \mathrm{N}(\bm{0}, \|\bm{\alpha}\|^2 \bm{\Delta}_S).$$

\subsection{Proof of \Cref{thm:consistency-rr}}
\label{app:consistency-rr}

As in the proof of \Cref{thm:asymptotics-nc}, we prove the
conclusions in this theorem for any fixed $\tilde{\bm{W}}_1$ and fixed
sequences $\{\bm{X}^{(n)}\}_{n=1}^{\infty}$ and
$\{\bm{R}^{(n,p)}\}_{n=1,p=1}^{\infty}$ such that
$\|\bm{X}^{(n)}\|_2/\sqrt{n} \to 1$ and $\bm{R}^{(n,p)} \to \bm{I}_r$ as $n,p \to \infty$. For brevity we
will write $\bm{X}$ and $\bm{R}$ instead of $\bm{X}^{(n)}$ and $\bm{R}^{(n,p)}$ for the rest of this proof.
We abbreviate $\hat{\bm{\alpha}}^{\mathrm{RR}}$
as $\hat{\bm{\alpha}}$ in this proof. To avoid confusion, we
use $\bm \alpha$ for the true value of the parameter and
$\tilde{\bm \alpha}$ to represent a vector in $\mathbb{R}^r$.

  Because $\bm{\alpha}^{(0)} \to \bm{\alpha}$, we prove
  this theorem by showing that for any $\epsilon > 0$,
  $\mathrm{P}(\|\hat{\bm{\alpha}} -
  \bm{\alpha}^{(0)}\|_0 \ge \epsilon) \to 0$. We break down our proof
  to two key results: First, we show $\hat{\bm{\alpha}}$ and
  $\bm{\alpha}^{(0)}$ are close in the following sense
\begin{equation} \label{eq:phi-op1}
\varphi(\bm{\alpha}^{(0)} -
  \hat{\bm{\alpha}}) = \frac{1}{p} \sum_{j=1}^p \rho\left(\frac{
  \hat{\bm{\Gamma}}_j^T (\bm{\alpha}^{(0)} -
  \hat{\bm{\alpha}})}{\hat{\sigma}_j}\right) = o_p(1),
\end{equation}
and second, we show that for sufficiently small
  $\epsilon > 0$, there exists $\tau > 0$ such that as $n, p \to \infty$
\begin{equation} \label{eq:consistency-key}
  \mathrm{P} \left( \inf_{\|\tilde{\bm{\alpha}}\|_2 \ge \epsilon}
  {\varphi}(\tilde{\bm{\alpha}}) > \tau \right)
\to 1.
\end{equation}

Based on these two results and the observation that
\[
\{ \|\bm{\alpha}^{(0)} - \hat{\bm{\alpha}}\|_2 < \epsilon \} \supseteq
\left\{ {\varphi}(\bm{\alpha}^{(0)} - \hat{\bm{\alpha}}) < \tau \right\} \bigcap
\left\{ \inf_{\|\tilde{\bm{\alpha}}\|_2 \ge \epsilon} {\varphi}(\tilde{\bm{\alpha}}) > \tau \right\},
\]
we conclude that $\mathrm{P}(\|\hat{\bm{\alpha}} -  \bm{\alpha}^{(0)}\|_2 \ge \epsilon) \to 0$.

Let's start with \eqref{eq:phi-op1}. Denote
$l_p(\tilde{\bm{\alpha}}) = p^{-1}\sum_{j=1}^p
  \rho \left( {\tilde{Y}_{1j} / \|\bm{X}\|_2 -
  \hat{\bm{\Gamma}}_j^T \tilde{\bm{\alpha}}}/{\hat{\sigma}_j}\right)$. By
  \eqref{eq:alpha-rr}, we have $\hat{\bm{\alpha}}^{\mathrm{RR}} = \arg
  \min l_p(\tilde{\bm{\alpha}})$, so $l_p(\hat{\bm{\alpha}}) \le
  l_p(\bm{\alpha}^{(0)})$. We examine the difference between
  $l_p(\tilde{\bm{\alpha}})$ and $\varphi(\bm{\alpha}^{(0)} - \tilde{\bm{\alpha}})$ for any
  $\tilde{\bm{\alpha}}$, starting from
\begin{equation*}
\begin{split}
l_p(\tilde{\bm{\alpha}}) & = \frac{1}{p}\sum_{j=1}^p
  \rho \left( \frac{\tilde{Y}_{1j} / \|\bm{X}\|_2 -
      \hat{\bm{\Gamma}}_j^T \tilde{\bm{\alpha}}}{\hat{\sigma}_j}\right) \\
& = \frac{1}{p} \sum_{j=1}^p \rho\left(  \frac{ \beta_j + (\bm{\Gamma}^{(0)}_j)^T
  \bm{\alpha}^{(0)} + \tilde{E}_{1j} / \|\bm{X}\|_2 - \hat{\bm{\Gamma}}_j^T
    \tilde{\bm{\alpha}}}{\hat{\sigma}_j}\right).\\
\end{split}
\end{equation*}
Because $\rho$ has bounded derivative, $|\rho(x) - \rho(y)| \le D |x -
y|$ for any $x,y \in \mathbb{R}$. In the statement of
\Cref{thm:consistency-rr} we assume $\|\bm{\beta}\|_1/p \to 0$. This
together with $\|\bm{X}\|_2 \to 0$ implies that
\begin{equation*}
  l_p(\tilde{\bm{\alpha}}) = \frac{1}{p} \sum_{j=1}^p \rho\left(  \frac{ (\bm{\Gamma}^{(0)}_j)^T
  \bm{\alpha}^{(0)} - \hat{\bm{\Gamma}}_j^T
    \tilde{\bm{\alpha}}}{\hat{\sigma}_j}\right) + o_p(1).
\end{equation*}
Next,
\[
\left|\frac{ (\bm{\Gamma}^{(0)}_j)^T
  \bm{\alpha}^{(0)}  - \hat{\bm{\Gamma}}_j^T
    {\tilde{\bm{\alpha}}}}{\hat{\sigma}_j} - \frac{
  \hat{\bm{\Gamma}}_j (\bm{\alpha}^{(0)} -
  {\tilde{\bm{\alpha}}})}{\hat{\sigma}_j}\right|
=  \left|
  \frac{(\bm{\Gamma}^{(0)}_j - \hat{\bm{\Gamma}}_j)^T
    \bm{\alpha}^{(0)}}{\hat{\sigma}_j}
\right| \overset{p}{\to} \bm{0}.
\]
Therefore, by the same argument as before,
\begin{equation}
  \label{eq:consistency-proof-1}
  l_p(\tilde{\bm{\alpha}}) = \frac{1}{p} \sum_{j=1}^p \rho\left( \frac{
  \hat{\bm{\Gamma}}_j (\bm{\alpha}^{(0)} -
  {\tilde{\bm{\alpha}}})}{\hat{\sigma}_j} \right) + o_p(1) =
\varphi(\bm{\alpha}^{(0)} - \tilde{\bm{\alpha}}) + o_p(1).
\end{equation}
Also, $\varphi(\bm{0}) = 0$ because $\rho(0)
= 0$. Therefore $l_p(\hat{\bm{\alpha}}) \le l_p(\bm{\alpha}^{(0)}) =
o_p(1)$. Notice that the $o_p(1)$ term in
\eqref{eq:consistency-proof-1} does not depend on $\hat{\bm{\alpha}}$,
hence $\varphi(\bm{\alpha}^{(0)} - \hat{\bm{\alpha}}) =
l_p(\hat{\bm{\alpha}}) + o_p(1) = o_p(1)$.

Next we prove \eqref{eq:consistency-key}. Since $\rho(x)$ is 
non-decreasing when $x \geq 0$,
\[
\inf_{\|\tilde{\bm{\alpha}}\|_2 \ge \epsilon} \frac{1}{p} \sum_{j=1}^p \rho\left(\frac{
  \hat{\bm{\Gamma}}_j^T \tilde{\bm{\alpha}}}{\hat{\sigma}_j}\right) \ge
\inf_{\|\tilde{\bm{\alpha}}\|_2 = \epsilon}
\frac{1}{p} \sum_{j=1}^p \rho\left(\frac{
  \hat{\bm{\Gamma}}_j^T \tilde{\bm{\alpha}}}{\hat{\sigma}_j}\right).
\]
If $p/n^k \to 0$ for some $k > 0$, then using \Cref{lem:gamma-ml}, there exists 
some constant $D^\star$ that $\Prob(\max_j\|\hat {\bm \Gamma}_j\|_2 \leq D^\star) \to 1$.
Thus when $\max_j\|\hat {\bm \Gamma}_j\|_2 \leq D^\star$ holds,
there is sufficiently small $\epsilon > 0$, the $\tilde{\bm{\alpha}}$ on the right hand
side is within the neighborhood where $\rho$ is strongly convex in
\Cref{assumption:rho}, so for some $\kappa > 0$
\[
\inf_{\|\tilde{\bm{\alpha}}\|_2 = \epsilon}
\frac{1}{p} \sum_{j=1}^p \rho\left(\frac{
  \hat{\bm{\Gamma}}_j^T \tilde{\bm{\alpha}}}{\hat{\sigma}_j}\right) \ge
\inf_{\|\tilde{\bm{\alpha}}\|_2 = \epsilon} \kappa \cdot
\frac{1}{p} \sum_{j=1}^p \left(\frac{
  \hat{\bm{\Gamma}}_j^T \tilde{\bm{\alpha}}}{\hat{\sigma}_j}\right)^2 = \kappa\epsilon^2 \cdot
\lambda_{\mathrm{min}}\left( \hat{\bm{\Gamma}}^T
  \hat{\bm{\Sigma}}^{-1} \hat{\bm{\Gamma}} \right).
\]
By the uniform consistency of $\hat{\bm{\Gamma}}$ and $\hat{\bm{\Sigma}}$ using \Cref{lem:gamma-ml}, we
conclude \eqref{eq:consistency-key} is true for $\tau = \kappa\epsilon^2
\lambda_{\mathrm{min}}(\bm{\Gamma}^T \bm{\Sigma}^{-1} \bm{\Gamma}) /
2$, where $\lambda_{\mathrm{min}}(\bm{\Gamma}^T \bm{\Sigma}^{-1}
\bm{\Gamma}) > 0$ by \Cref{assumption:large-factor}.

\subsection{Proof of \Cref{thm:asymptotics-rr}}
\label{app:thm-rr}

In
fact, by using the representation of $\hat{\bm{\Gamma}}$ in
\eqref{eq:gamma-hat} and (C.4) in \citep{bai2012}.
  \begin{equation*}
    \label{eq:eq:sigma-hat}
    \hat{\sigma}_j^2 = \frac{1}{n-1} \tilde{\bm{E}}_{-1,j}^T
    \tilde{\bm{E}}_{-1,j} + o_p(n^{-\frac{1}{2}}),
  \end{equation*}
we have
\[
\begin{split}
\bm{\Psi}_{p}(\bm{\alpha}^{(0)}) =& \frac{1}{p} \sum_{j=1}^p \psi\left( \frac{\tilde{Y}_{1j} / \|\bm{X}\|_2 -
    \hat{\bm{\Gamma}}_j^T
    \bm{\alpha}^{(0)}}{\hat{\sigma}_j}\right)
\hat{\bm{\Gamma}}_j / \hat{\sigma}_j \\
=& \frac{1}{p} \sum_{j=1}^p \psi\left( \frac{\beta_j + \tilde{E}_{1j}
    / \|\bm{X}\|_2 + \bm{\Gamma}_j^T \bm{\alpha}^{(0)} - (\hat{\bm{\Gamma}}_j^0)^T
    \bm{\alpha}^{(0)}}{\hat{\sigma}_j}\right)
\hat{\bm{\Gamma}}_j / \hat{\sigma}_j \\
=& \frac{1}{p} \sum_{j=1}^p \psi\left( \frac{\beta_j + \tilde{E}_{1j}
    / \|\bm{X}\|_2 + \frac{1}{n-1} \tilde{\bm{E}}_{-1,j} \tilde{\bm{Z}}_{-1}^{(0)} \bm{\alpha}^{(0)} + o_p(n^{-\frac{1}{2}})
    }{\sqrt{\frac{1}{n-1} \tilde{\bm{E}}_{-1,j}^T
    \tilde{\bm{E}}_{-1,j} + o_p(n^{-\frac{1}{2}})}}\right)
\hat{\bm{\Gamma}}_j / \hat{\sigma}_j
\end{split}
\]
Let $W_j$ be the expression inside $\psi$ in the last equation without
the two $o_p(n^{-\frac{1}{2}})$ terms. Condition on
$\tilde{\bm{Z}}_{-1}^0$, $W_j$ and $W_k$ are independent if $j \ne
k$ and $W_j \overset{d}{=} W_k$ if $\beta_j = \beta_k = 0$. It is not hard to show the denominator of $W_j$ converges in
probability to $\sigma_j$, while the numerator of $W_j$ is $\beta_j +
n^{-\frac{1}{2}} t_j + o_p(n^{-\frac{1}{2}})$ where $t_j$ are
i.i.d. normal random variables by noticing $\tilde{\bm{Z}}_{-1}^{(0)}
\bm{\alpha}^{(0)} = \tilde{\bm{Z}}_{-1} (\bm{R} \bm{R}^T)^{-1}
(\bm{\alpha} + \bm{u}/\|\bm{X}\|_2)$ and $\bm{R} \bm{R}^T \to \bm{I}_r$. Because $\psi(W) = \psi(0) + \psi'(0) W +
\frac{1}{2} \psi''(0) W'^2$ for some $W'$ between $0$ and $W$ and
$\psi(0) = 0$, $\psi(W_j) = \psi'(0) (n^{-\frac{1}{2}} t_j + o_p(n^{-\frac{1}{2}}))$ if $\beta_j =
0$. Also $\psi$ is assumed to be bounded $|\psi(W_j)| \le C$ in
\cref{assumption:rho}. Therefore, by using $\lim_{p \to \infty}
\frac{1}{p} \bm{\Gamma}^T \bm{\Sigma}^{-1} \bm{\Gamma}$ exists and is
positive definite (in \cref{assumption:large-factor})
\[
\begin{split}
\left|\bm{\Psi}_{p}(\bm{\alpha}^{(0)})\right| &= \left| \frac{1}{p} \sum_{j=1}^p \psi\left(W_j
+ o_p(n^{-\frac{1}{2}})\right) \hat{\bm{\Gamma}}_j / \hat{\sigma}_j
\right| \\
&= \left| \frac{1}{p} \left( \sum_{\beta_j \ne 0} + \sum_{\beta_j = 0} \right) \psi\left(W_j
+ o_p(n^{-\frac{1}{2}})\right) \hat{\bm{\Gamma}}_j / \hat{\sigma}_j
\right| \\
&\le \frac{1}{p} C \|\bm{\beta}\|_0 + O_p(n^{-\frac{1}{2}}
p^{-\frac{1}{2}}) \\
&= o_p(n^{-\frac{1}{2}})
\end{split}
\]

Similarly,
\[
\begin{split}
  \left[\nabla \bm{\Psi}_p(\bm{\alpha}^{(0)})\right]^{-1} &=
\left[\frac{1}{p} \sum_{j=1}^p
  \psi'\left(W_j + o_p(n^{-\frac{1}{2}})\right)
\hat{\bm{\Gamma}}_j \hat{\bm{\Gamma}}_j^T / \hat{\sigma}_j^2
\right]^{-1} \\
&= \left[ \frac{1}{p} \sum_{j=1}^p \left(\psi'(\beta_j) +
    O_p(n^{-\frac{1}{2}})\right) \hat{\bm{\Gamma}}_j
  \hat{\bm{\Gamma}}_j^T / \hat{\sigma}_j^2 \right]^{-1} \\
&= \left[ \frac{1}{p} \sum_{j=1}^p \left(\psi'(0) +
    O_p(n^{-\frac{1}{2}})\right) \hat{\bm{\Gamma}}_j
  \hat{\bm{\Gamma}}_j^T / \hat{\sigma}_j^2 + \frac{1}{p} \sum_{\beta_j
    \ne 0}
  (\psi'(\beta_j) - \psi'(0)) \hat{\bm{\Gamma}}_j
  \hat{\bm{\Gamma}}_j^T / \hat{\sigma}_j^2  \right]^{-1} \\
&\overset{p}{\to} \left[ \psi'(0) \frac{1}{p} \sum_{j=1}^p
  \bm{\Gamma}_j \bm{\Gamma}_j^T / \sigma_j^2 \right]^{-1}
\end{split}
\]
The last probability limit follows from $\hat{\bm{\Gamma}}
\overset{p}{\to} \bm{\Gamma}$, $\hat{\bm{\sigma}}_j
\overset{p}{\to} \sigma_j$ and $\psi'$, $\bm{\Gamma}$ and $\bm{\Sigma}$ are bounded. This means all the eigenvalues of
$\left[\nabla \bm{\Psi}_p(\bm{\alpha}^{(0)})\right]^{-1}$ are converging to some constants.

\subsection{Proof of \Cref{cor:type-I}}
\label{app:thm-type-I}

We begin with a lemma regarding the test statistics.

\begin{lemma}
\label{lem:test-statistic-uniform}
The test statistics \eqref{eq:test-statistic} can be written as $t_j =
z_j + v_j$ for $j\in\mathcal{N}_p$, where $z_{j}, j\in\mathcal{N}_p$ are independent standard normal variables
and $v_j, j\in\mathcal{N}_p$ satisfy $\max_{1 \le j \le p} |v_j| =
o_p(1)$.
\end{lemma}
\begin{proof}
  We first prove this lemma for the RR estimator
  $\hat{\beta}^{\mathrm{RR}}$ and the corresponding test
  statistics. Let
\[
z_j = \frac{\tilde{E}_{1j} +\frac{1}{\sqrt{n-1}}\sum_{i=2}^n
  \tilde{E}_{ij} (\tilde{\bm{Z}}_i^{(0)})^T \bm{\alpha}}{\sigma_j \sqrt{1 +
    \|\bm{\alpha}\|^2}}.
\]
It is easy to verify that $z_j$ i.i.d.\ $\sim \mathrm{N}(0,1),~j=1,\dotsc,p$.
By using the expression of
  $\hat{\bm{\beta}}^{\mathrm{RR}}$ in \eqref{eq:beta-expand}, we can show
  for $j \in \mathcal{N}$ (so $\beta_j = 0$), that
\[
\begin{split}
&\max_{j \in \mathcal{N}_p}\left|\|\bm{X}\|_2 \hat{\beta}_j - \sigma_j \sqrt{1 +
    \|\bm{\alpha}\|^2} z_j\right| \\
=&\max_{j \in \mathcal{N}_p}\left|\|\bm{X}\|_2 \hat{\beta}_j - \left[\tilde{E}_{1j} +\frac{1}{\sqrt{n-1}}\sum_{i=2}^n
  \tilde{E}_{ij} (\tilde{\bm{Z}}_i^{(0)})^T \bm{\alpha}\right]\right| \\
=& \max_{j \in \mathcal{N}_p} \left|
\|\bm{X}\|_2 (\bm{\Gamma}_j^{(0)} - \hat{\bm{\Gamma}}_j)^T
\hat{\bm{\alpha}} - \frac{1}{\sqrt{n-1}}\sum_{i=2}^n
  \tilde{E}_{ij} (\tilde{\bm{Z}}_i^{(0)})^T \bm{\alpha} + r_j \right| \\
=& o_p(1),
\end{split}
\]
where $r_j =  \hat{\bm{\Gamma}}_j^T \left[ \nabla
  \bm{\Psi}_p(\bm{\alpha}^{(0)}) + \bm{o}_p(1) \right]^{-1}
\|\bm{X}\|_2 \bm{\Psi}_p(\bm{\alpha}^{(0)})$. The last step is due to
the uniform convergence of $\hat{\bm{\Gamma}}_j$ in
\eqref{eq:max-gamma-err} and \Cref{lem:asymptotics-rr-bound} to
uniformly control $r_j$. Now we can show, by using the uniform
convergence rate of $\hat{\sigma}^2$, that
\[
\begin{split}
\max_{j \in \mathcal{N}_p} |v_j|  =& \max_{j \in \mathcal{N}_p} \left|\frac{\|\bm{X}\|_2\hat{\beta}_j}{\hat{\sigma}_j \sqrt{1 +
    \|\hat{\bm{\alpha}}\|^2}} - z_j\right| \\
 =& \max_{j \in \mathcal{N}_p} \frac{\left|\|\bm{X}\|_2\hat{\beta}_j - \hat{\sigma}_j \sqrt{1 +
    \|\hat{\bm{\alpha}}\|^2 } z_j \right|}{\hat{\sigma}_j \sqrt{1 +
    \|\hat{\bm{\alpha}}\|^2}} \\
=& O_p\left(\max_{j \in \mathcal{N}_p}\left|\|\bm{X}\|_2\hat{\beta}_j - \hat{\sigma}_j \sqrt{1 +
    \|\hat{\bm{\alpha}}\|^2 } z_j \right|\right) \\
\le& O_p\left(\max_{j \in \mathcal{N}_p}\left|\|\bm{X}\|_2 \hat{\beta}_j - \sigma_j \sqrt{1 +
    \|\bm{\alpha}\|^2} z_j \right| \right) \\ &+ O_p\left(\max_{j \in
  \mathcal{N}_p}\left| \hat{\sigma}_j \sqrt{1 +
    \|\hat{\bm{\alpha}}\|^2} z_j - \sigma_j \sqrt{1 +
    \|\bm{\alpha}\|^2} z_j   \right| \right) \\
=&o_p(1).
\end{split}
\]
For the negative control estimator, the same argument holds
by noticing the $\bm{o}_p(1)$ term in
\eqref{eq:beta-nc-asymptotics-intermediate} is also uniform over $j$ (similar to $r_j$).
\end{proof}

To prove the first conclusion in \Cref{cor:type-I}, we
show the left hand side of \eqref{eq:type-I} has expectation
converging to $\alpha$ and variance converging to zero. For the expectation, for
  any $\epsilon > 0$,
\[
\begin{split}
\frac{1}{|\mathcal{N}_p|} \sum_{j \in \mathcal{N}_p} \mathrm{P}(|t_j| >
z_{\alpha/2})
&\le \frac{1}{|\mathcal{N}_p|} \sum_{j \in \mathcal{N}_p}
\mathrm{P}(|z_j| > z_{\alpha/2} - \epsilon) + \mathrm{P}(|v_j| >
\epsilon) \\
& = 2 (1 - \mathrm{\Phi}(z_{\alpha/2} - \epsilon)) +
\frac{1}{|\mathcal{N}_p|} \sum_{j \in \mathcal{N}_p} \mathrm{P}(|v_j| >
\epsilon) \\
& \le 2 (1 - \mathrm{\Phi}(z_{\alpha/2} - \epsilon)) +
\mathrm{P}(\max_{1 \le j \le p} |v_j| > \epsilon) \\
& \to 2 (1 - \mathrm{\Phi}(z_{\alpha/2} - \epsilon)).
\end{split}
\]
Similarly, one can prove $\lim_{n,p \to \infty}\frac{1}{|\mathcal{N}_p|} \sum_{j \in
  \mathcal{N}_p} \mathrm{P}(|t_j| > z_{\alpha/2}) \ge 2
(1 - \mathrm{\Phi}(z_{\alpha/2} + \epsilon))$ for any $\epsilon > 0$. Thus
the expectation converges to $\alpha$ when $n,p \to \infty$.

For the variance, we compute the second moment of the left hand side
of \eqref{eq:type-I}: for any $\epsilon > 0$,
\[
\begin{split}
&\frac{1}{|\mathcal{N}_p|^2} \sum_{j,k \in \mathcal{N}_p} \mathrm{P}(|t_j| >
z_{\alpha/2}, |t_k| > z_{\alpha / 2}) \\
=& \frac{1}{|\mathcal{N}_p|^2} \sum_{j \in \mathcal{N}_p} \mathrm{P}(|t_j| >
z_{\alpha/2}) + \frac{1}{|\mathcal{N}_p|^2} \sum_{j,k \in
  \mathcal{N}_p, j \ne k} \mathrm{P}(|t_j| > z_{\alpha/2}, |t_k| > z_{\alpha / 2}) \\
\le& \frac{1}{|\mathcal{N}_p|} \left[2 (1 - \mathrm{\Phi}(z_{\alpha/2} - \epsilon)) +
\mathrm{P}(\max_{1 \le j \le p} |v_j| > \epsilon) \right] \\
&+ \frac{1}{|\mathcal{N}_p|^2} \sum_{j,k \in
  \mathcal{N}_p, j \ne k} \mathrm{P}(|z_j| > z_{\alpha/2} -
\epsilon,|z_k| > z_{\alpha/2} - \epsilon) \\
&+ \mathrm{P}(|v_j| >
\epsilon) + \mathrm{P}(|v_k| > \epsilon) \\
= & \frac{1}{|\mathcal{N}_p|^2} \sum_{j,k \in
  \mathcal{N}_p, j \ne k} \mathrm{P}(|z_j| > z_{\alpha/2} -
\epsilon,|z_k| > z_{\alpha/2} - \epsilon) + o(1) \\
= & \frac{|\mathcal{N}_p|-1}{|\mathcal{N}_p|} [2 (1 - \mathrm{\Phi}(z_{\alpha/2}
- \epsilon))]^2 + o(1) \\
\to & 4[1 - \mathrm{\Phi}(z_{\alpha/2}
- \epsilon)]^2
\end{split}
\]
Similarly we can prove the lower bound of the second moment. In
conclusion, the second moment converges to $\alpha^2$, hence the
variance of \eqref{eq:type-I} converges to $0$.

To prove the second conclusion in \Cref{cor:type-I}, we begin with
\[
\begin{split}
&\mathrm{P}\Big(\sum_{j \in \mathcal{N}_p} I(|t_j| > \mathrm{\Phi}^{-1}(1 - \alpha/(2p))) \ge
  1\Big) \\
=& \mathrm{P}\Big(\max_{j \in \mathcal{N}_p} |t_j|
  >\mathrm{\Phi}^{-1}(1 - \alpha/(2p)) \Big) \\
=& \mathrm{P}\Big(\max_{j \in \mathcal{N}_p} |z_j + v_j|
  >\mathrm{\Phi}^{-1}(1 - \alpha/(2p)) \Big) \\
\le & \mathrm{P}\Big(\max_{j \in \mathcal{N}_p} |z_j|
  >\mathrm{\Phi}^{-1}(1 - \alpha/(2p)) - \max_{j \in \mathcal{N}_p}
  |v_j| \Big) \\
\le & \mathrm{P}\Big(\max_{1 \le j \le p} |z_j|
  >\mathrm{\Phi}^{-1}(1 - \alpha/(2p)) - \max_{j \in \mathcal{N}_p}
  |v_j| \Big) \\
\end{split}
\]
The conclusion \eqref{eq:fwer} then follows from $\mathrm{P}\Big(\max_{1 \le j \le p} |z_j|
  >\mathrm{\Phi}^{-1}(1 - \alpha/(2p)) \Big) \le \alpha$ (the validity
  of Bonferroni for i.i.d.\ normals), the fact that
  $\mathrm{\Phi}^{-1}(1 - \alpha/(2p)) \to \infty$ as $p \to \infty$,
  and the result in \Cref{lem:test-statistic-uniform} that $\max_{j
    \in \mathcal{N}_p} |v_j| = o_p(1)$.

\subsection{Proof of \Cref{thm:alpha-test}}
\label{app:thm-alpha-0}

First, we point out that when $\bm\alpha = \bm 0$, as $n, p \to \infty$ 
\begin{equation}\label{eq:asy-alpha-0}
\sqrt n \cdot \bm \alpha^{(0)} = \sqrt n \cdot \bm R^{-1} \tilde{\bm W}_1  \overset{d}{\to} \mathrm{N}(0, \bm I_r)
\end{equation}
where $\bm \alpha^{(0)}$, $\bm R$ and $\tilde{\bm W}_1$ are defined in \Cref{sec:inference-beta-alpha}. 
This is due to the fact that $\bm R \to \bm I_r$ (\Cref{rmk:rotation}) and 
$\sqrt n\tilde{\bm W}_1 \sim \mathrm{N}(\bm 0, \bm I_r)$, thus Slutsky's Theorem implies \eqref{eq:asy-alpha-0}.
Next, we show that 
\begin{equation}\label{eq:asy-hat-alpha}
\sqrt n \cdot({\hat {\bm \alpha} - \bm\alpha^{(0)}}) = \bm o_p(1)
\end{equation}

For the negative control scenario, using the expression of $\hat{\bm \alpha}^{\NC}$ in \eqref{eq:alpha-nc} and 
$\tilde{\bm Y}_{1, \calC}/\|\bm X\|_2$ in \eqref{eq:y-tilde-1-separate}, we get
\[
\begin{split}
\sqrt n (\hat{\bm \alpha} - \bm \alpha^{(0)}) 
=& \sqrt n(\hat{\bm \Gamma}_\calC \hat{\bm \Sigma}_\calC^{-1}\hat{\bm \Gamma}_\calC )^{-1}
\hat{\bm \Gamma}_\calC \hat{\bm \Sigma}_\calC^{-1}(\bm \Gamma_\calC^{(0)} - \hat{\bm \Gamma}_\calC) \bm \alpha^{(0)}\\
+& \frac{\sqrt n}{\|\bm X\|_2}(\hat{\bm \Gamma}_\calC \hat{\bm \Sigma}_\calC^{-1}\hat{\bm \Gamma}_\calC )^{-1}
\hat{\bm \Gamma}_\calC \hat{\bm \Sigma}_\calC^{-1}\tilde{\bm E}_{1\calC}^T.
\end{split}
\]
Using the facts we got in \eqref{eq:nc-facts} and $\bm \alpha^{(0)} = \bm o_p(1)$, we further get
\[
\sqrt n (\hat{\bm \alpha} - \bm \alpha^{(0)}) 
= (\bm \Gamma_\calC \bm \Sigma_\calC^{-1} \bm \Gamma_\calC)^{-1} \bm \Gamma_\calC \bm \Sigma_\calC^{-1}
\tilde{\bm E}_{1\calC}^T + \bm o_p(1).
\]
Under \Cref{assumption:nc-factor}, if $|\calC| \to \infty$, the maximum eigenvalue of 
$(\bm \Gamma_\calC \bm \Sigma_\calC^{-1} \bm \Gamma_\calC)^{-1}$ goes to $0$, thus \eqref{eq:asy-hat-alpha} holds 
for the negative control scenario.

For the sparsity scenario, in the proof of \Cref{thm:asymptotics-rr}, we have shown that 
\[
 \sqrt n(\hat{\bm \alpha}^\RR - \bm \alpha^{(0)}) = -\sqrt n \left[ \nabla 
 \bm{\Psi}_p(\bm{\alpha}^{(0)}) + \bm o_p(1) \right]^{-1} 
 \bm{\Psi}_p(\bm \alpha^{(0)})
\]
Thus, because of \Cref{lem:asymptotics-rr-bound}, \eqref{eq:asy-hat-alpha} also holds for the sparsity scenario.

Finally, combining \eqref{eq:asy-alpha-0} and \eqref{eq:asy-hat-alpha}, \Cref{thm:alpha-test} holds.

\subsection{Proof of \Cref{lem:part-asym}}
\label{app:lem-multiple}

First, note that by the strong law of large numbers
$\frac{1}{n}
\begin{pmatrix}
  \bm{X}_0 & \bm{X}_1 \\
\end{pmatrix}^T
\begin{pmatrix}
  \bm{X}_0 & \bm{X}_1 \\
\end{pmatrix}
\overset{a.s.}{\to} \bm{\Sigma}_{\bm{X}}$. Using the
QR decomposition of
$
\begin{pmatrix}
  \bm{X}_0 & \bm{X}_1 \\
\end{pmatrix} = \bm{Q} \bm{U}
$ and writing $\bm{U} =
\begin{pmatrix}
  \bm{V} \\
  \bm{0} \\
\end{pmatrix}
$ and $\bm{V} =
\begin{pmatrix}
  \bm{U}_{00} & \bm{U}_{01} \\
  \bm{0} & \bm{U}_{11} \\
\end{pmatrix}
$,
it's clear that $\frac{1}{n} \bm{V}^T \bm{V}
\overset{a.s.}{\rightarrow} \bm{\Sigma}_{\bm{X}}$.
Since $\bm{\Sigma}_{\bm{X}}$ is nonsingular, both
$\bm{U}_{00}$ and $\bm{U}_{11}$ are full rank
square matrices with probability $1$. Thus using the block matrix inversion formula, we have
$\bm{V}^{-1} =
\begin{pmatrix}
  \star & \star \\
    \bm{0} & \bm{U}_{11}^{- 1}
\end{pmatrix}$ where $\star$ represents some $d_0 \times d_0$ or $d_0
\times d_1$ matrix.
Therefore the right bottom block of $n \bm{V}^{-1} \bm{V}^{-T}$
is $n \bm{U}_{11}^{-1} \bm{U}_{11}^{-T}$ and converges to
$\bm{\Omega}_{11}$ almost surely.

\subsection{Proof of \Cref{thm:multiple-rr}}
\label{app:thm-multiple}
First, for the known zero indices scenario,
$\hat{\bm{\Alpha}}_1^{\mathrm{NC}}$ has the following formula, which is similar to
\eqref{eq:alpha-nc}:
\begin{equation}\label{eq:alpha-nc-ext}
\hat{\bm{\Alpha}}_1^{\mathrm{NC}} = (\hat{\bm{\Gamma}}_{\mathcal{C}}^T \hat{\bm{\Sigma}}_{\mathcal{C}}^{-1}
  \hat{\bm{\Gamma}}_{\mathcal{C}})^{-1} \hat{\bm{\Gamma}}_{\mathcal{C}}^T \hat{\bm{\Sigma}}_{\mathcal{C}}^{-1}
  \tilde{\bm{Y}}_{1,\mathcal{C}}^T \bm{U}_{11}^{-T}
\end{equation}
which implies a similar formula as \eqref{eq:beta-nc-asymptotics-intermediate}:
\begin{equation}
\label{eq:beta-nc-asymptotics-intermediate-ext}
\begin{split}
\sqrt{n} (\hat{\bm{\Beta}}^{\mathrm{NC}}_{1,S} -
\bm{\Beta}_{1,S}) =& \sqrt n\tilde{\bm{E}}_{1,S}^T\bm{U}_{11}^{-T}  -
\sqrt n\cdot\bm{\Gamma}_{S} (\bm{\Gamma}_{\mathcal{C}}^T  \bm{\Sigma}_\mathcal{C}^{-1}
\bm{\Gamma}_{\mathcal{C}})^{-1} \bm{\Gamma}_{\mathcal{C}}^T \bm{\Sigma}_\mathcal{C}^{-1} \tilde{\bm{E}}_{1,C}^T\bm{U}_{11}^{-T} \\
  +& \sqrt{n} \cdot (\bm{\Gamma}_{S}^{(0)}  -
\hat{\bm{\Gamma}}_{S} ) {\bm{\Alpha}}_1^{(0)} \\
  +& \sqrt{n} \cdot \bm{\Gamma}_{S} (\bm{\Gamma}_{\mathcal{C}}^T
  \bm{\Sigma}_\mathcal{C}^{-1}  \bm{\Gamma}_{\mathcal{C}})^{-1} \bm{\Gamma}_{\mathcal{C}}^T \bm{\Sigma}_\mathcal{C}^{-1}
  (\hat{\bm{\Gamma}}_{\mathcal{C}}  -  \bm{\Gamma}_{\mathcal{C}}^{(0)})
  {\bm{\Alpha}}_1^{(0)} + \bm{o}_p(1),
\end{split}
\end{equation}
where ${\bm{\Alpha}}_1^{(0)} = \bm{R}^{-1}(\bm{A}_1 + \bm{\mathrm{U}}\bm{U}_{11}^{-T})$.
Following
the proof of \Cref{thm:asymptotics-nc} by using \Cref{lem:part-asym}, we get
\eqref{eq:nc-asym-ext}.

For the unknown zero indices scenario, \Cref{lem:part-asym} guarantees the consistency of
each column of $\hat{\bm{\Alpha}}_1^{\mathrm{RR}}$ by using \Cref{thm:consistency-rr}.
Then the Taylor expansion used in the proof of \Cref{thm:asymptotics-rr}
still works at each column of $\bm{\Alpha}_1^{(0)}$. Similar to \eqref{eq:beta-expand},
we get
\begin{equation} \label{eq:beta-expand-ext}
\begin{split}
\sqrt{n}(\hat{\bm{\Beta}}_1^{\mathrm{RR}} - \bm{\Beta}_1) =&
  \sqrt{n} \tilde{\bm{E}}_1^T\bm{U}_{11}^{-T}
+ \sqrt{n} (\bm{\Gamma}^{(0)} -\hat{\bm{\Gamma}})
\hat{\bm{\Alpha}}_1^{\mathrm{RR}} \\
&+ \hat{\bm{\Gamma}}
\begin{pmatrix}\bm{g}_1 & \bm{g}_2 & \cdots & \bm{g}_{d_1}\end{pmatrix}
\end{split}
\end{equation}
where $ \bm{g}_i = \left[ \nabla
\bm{\Psi}_p(\bm{\Alpha}_{1, i}^{(0)}) \right]^{-1}
(\sqrt{n} \bm{\Psi}_p\big(\bm{\Alpha}_{1, i}^{(0)}) + \bm{o}_p(1)\big)$.
Following the proof of \Cref{thm:asymptotics-rr}, we get each $\bm{g}_i = \bm{o}_p(1)$.
Thus
$$
\sqrt{n}(\hat{\bm{\Beta}}_1^{\mathrm{RR}} - \bm{\Beta}_1) =
  \sqrt{n} \tilde{\bm{E}}_1^T\bm{U}_{11}^{-T}
+ \sqrt{n} \cdot(\bm{\Gamma}^{(0)} -\hat{\bm{\Gamma}})
\hat{\bm{\Alpha}}_1^{\mathrm{RR}} + \bm{o}_p(1)
$$
and \eqref{eq:rr-asym-ext} holds.


\section{Supplementary Figures and Tables}

\begin{figure}[hp]
  \centering
   \includegraphics[width = \textwidth]{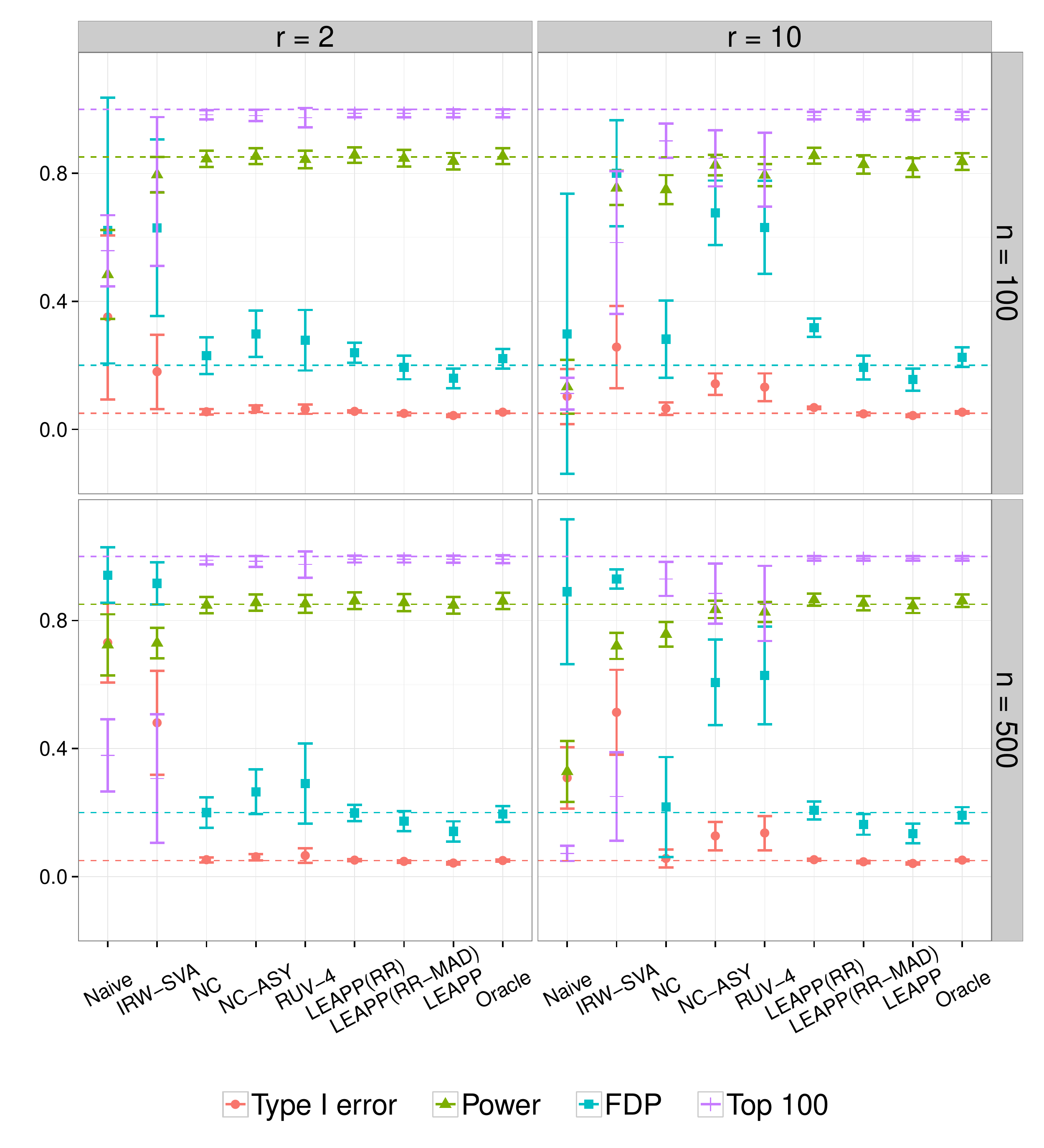}
  \caption{Compare the performance of nine different approaches when the variance of
    $\bm{X}$ explained by the confounding factors is $5\%$.   The error bars are one
  standard deviation over $100$ repeated simulations. The
  three dashed horizontal lines from bottom to top are the nominal significance level, FDR level,
 oracle power and the precision of the smallest $100$ p-values, respectively.}
  \label{fig:p-5000-methods05}
\end{figure}

\begin{figure}[hp]
  \centering
   \includegraphics[width = \textwidth]{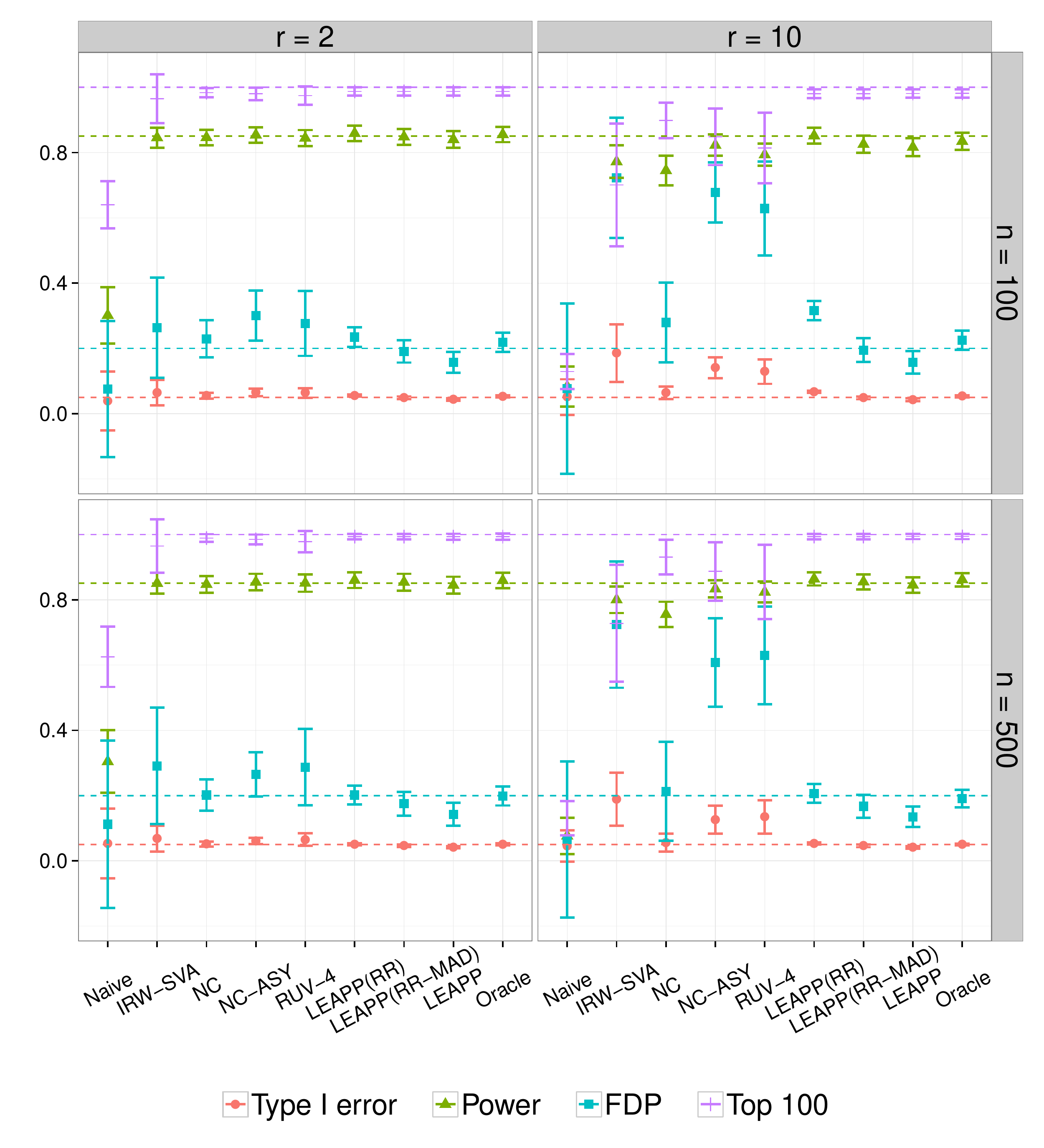}
  \caption{Compare the performance of nine different approaches when latent factors are
   unconfounding.   The error bars are one
  standard deviation over $100$ repeated simulations. The
  three dashed horizontal lines from bottom to top are the nominal significance level, FDR level,
 oracle power and the precision of the smallest $100$ p-values, respectively.}
  \label{fig:p-5000-methods05}
\end{figure}

\setlength{\tabcolsep}{3pt}
\begin{table}[h]
\centering
\footnotesize
\begin{tabular}{|$r|^r^r^r^r^r^r|^r^r^r^r|}
  \hline
r & mean & median & sd & mad & skewness & medc. & \#sig. & X/Y & top 100 & p-value \\
  \hline
0 & 0.077 & 0.14 & 1.25 & 1.04 & -0.949 & -0.064 & 605 & 97 & 58 &  NA\\
  1 & 0.19 & 0.21 & 1.37 &  1.2 & -0.556 & -0.02 & 458 & 90 & 72 & 0.013 \\
  2 & 0.15 & 0.19 & 1.41 & 1.23 & -0.464 & -0.027 & 457 & 91 & 74 & 0.039 \\
  3 & 0.015 & 0.055 & 1.38 & 1.18 & -0.442 & -0.035 & 509 & 97 & 75 & 0.00096 \\
  4 & 0.045 & 0.065 & 1.27 & 1.03 & -0.661 & -0.018 & 608 & 101 & 78 & 5.2e-07 \\
  5 & 0.044 & 0.067 &  1.3 & 1.06 & -0.573 & -0.019 & 612 & 100 & 76 & 1.8e-06 \\
  7 & 0.071 & 0.088 & 1.34 & 1.11 & -0.527 & -0.0097 & 572 & 99 & 76 & 2.7e-06 \\
  10 & 0.024 & 0.057 & 1.39 & 1.15 & -0.539 & -0.025 & 658 & 100 & 75 & 5.3e-07 \\
  15 & 0.097 & 0.12 & 1.48 & 1.23 & -0.619 & -0.018 & 659 & 102 & 76 & 2.4e-08 \\
  20 & 0.051 & 0.072 & 1.48 & 1.24 & -0.628 & -0.015 & 625 & 102 & 76 & 6.5e-11 \\
  30 & 0.021 & 0.038 & 1.58 & 1.28 & -0.709 & -0.01 & 748 & 109 & 81 & 5.6e-13 \\
  \rowstyle{\bfseries}
  33 & 0.032 & 0.052 & 1.63 & 1.33 & -0.669 & -0.013 & 751 & 109 & 79 & 2.7e-12 \\
  40 & 0.027 & 0.052 & 1.75 & 1.44 & -0.544 & -0.017 & 846 & 111 & 78 & 6.3e-11 \\
  50 & 0.034 & 0.054 & 1.93 & 1.59 & -0.389 & -0.0088 & 954 & 117 & 76 & 1.3e-09 \\
   \hline
\end{tabular}
\caption{\small Supplementary Dataset ($n = 143$, $p = 54675$). This
  is  same as dataset 1 except the primary variable is gender instead
  of COPD severity.}
  \label{table:supp_copd}
\end{table}

\bibliographystyle{imsart-nameyear}
\bibliography{reference/ref}

\begin{thebibliography}{65}

\bibitem[\protect\citeauthoryear{Alter, Brown and Botstein}{2000}]{alter2000}
\begin{barticle}[author]
\bauthor{\bsnm{Alter},~\bfnm{Orly}\binits{O.}},
  \bauthor{\bsnm{Brown},~\bfnm{Patrick~O}\binits{P.~O.}} \AND
  \bauthor{\bsnm{Botstein},~\bfnm{David}\binits{D.}}
(\byear{2000}).
\btitle{Singular value decomposition for genome-wide expression data processing
  and modeling}.
\bjournal{Proceedings of the National Academy of Sciences}
\bvolume{97}
\bpages{10101--10106}.
\end{barticle}
\endbibitem

\bibitem[\protect\citeauthoryear{Anderson and Rubin}{1956}]{anderson1956}
\begin{binproceedings}[author]
\bauthor{\bsnm{Anderson},~\bfnm{Theodore~W}\binits{T.~W.}} \AND
  \bauthor{\bsnm{Rubin},~\bfnm{Herman}\binits{H.}}
(\byear{1956}).
\btitle{Statistical inference in factor analysis}.
In \bbooktitle{Proceedings of the third Berkeley symposium on mathematical
  statistics and probability}
\bvolume{5}.
\end{binproceedings}
\endbibitem

\bibitem[\protect\citeauthoryear{Bai and Li}{2012a}]{bai2012}
\begin{barticle}[author]
\bauthor{\bsnm{Bai},~\bfnm{Jushan}\binits{J.}} \AND
  \bauthor{\bsnm{Li},~\bfnm{Kunpeng}\binits{K.}}
(\byear{2012}a).
\btitle{Statistical analysis of factor models of high dimension}.
\bjournal{The Annals of Statistics}
\bvolume{40}
\bpages{436--465}.
\end{barticle}
\endbibitem

\bibitem[\protect\citeauthoryear{Bai and Li}{2012b}]{bai2012supplement}
\begin{barticle}[author]
\bauthor{\bsnm{Bai},~\bfnm{Jushan}\binits{J.}} \AND
  \bauthor{\bsnm{Li},~\bfnm{Kunpeng}\binits{K.}}
(\byear{2012}b).
\btitle{Supplement to "Statistical analysis of factor models of high
  dimension."}.
\bdoi{10.1214/11-AOS966SUPP}
\end{barticle}
\endbibitem

\bibitem[\protect\citeauthoryear{Bai and Li}{2014}]{bai2014theory}
\begin{barticle}[author]
\bauthor{\bsnm{Bai},~\bfnm{Jushan}\binits{J.}} \AND
  \bauthor{\bsnm{Li},~\bfnm{Kunpeng}\binits{K.}}
(\byear{2014}).
\btitle{Theory and methods of panel data models with interactive effects}.
\bjournal{The Annals of Statistics}
\bvolume{42}
\bpages{142--170}.
\end{barticle}
\endbibitem

\bibitem[\protect\citeauthoryear{Bai and Li}{2015}]{bai2012approximate}
\begin{barticle}[author]
\bauthor{\bsnm{Bai},~\bfnm{Jushan}\binits{J.}} \AND
  \bauthor{\bsnm{Li},~\bfnm{Kunpeng}\binits{K.}}
(\byear{2015}).
\btitle{Maximum likelihood estimation and inference for approximate factor
  models of high dimension}.
\bjournal{The Review of Economics and Statistics}
\bpages{(to appear)}.
\end{barticle}
\endbibitem

\bibitem[\protect\citeauthoryear{Bai and Ng}{2002}]{bai2002}
\begin{barticle}[author]
\bauthor{\bsnm{Bai},~\bfnm{Jushan}\binits{J.}} \AND
  \bauthor{\bsnm{Ng},~\bfnm{Serena}\binits{S.}}
(\byear{2002}).
\btitle{Determining the Number of Factors in Approximate Factor Models}.
\bjournal{Econometrica}
\bvolume{70}
\bpages{191--221}.
\bdoi{10.1111/1468-0262.00273}
\end{barticle}
\endbibitem

\bibitem[\protect\citeauthoryear{Bai and Ng}{2006}]{bai2006}
\begin{barticle}[author]
\bauthor{\bsnm{Bai},~\bfnm{Jushan}\binits{J.}} \AND
  \bauthor{\bsnm{Ng},~\bfnm{Serena}\binits{S.}}
(\byear{2006}).
\btitle{Confidence Intervals for Diffusion Index Forecasts and Inference for
  Factor-Augmented Regressions}.
\bjournal{Econometrica}
\bvolume{74}
\bpages{1133--1150}.
\end{barticle}
\endbibitem

\bibitem[\protect\citeauthoryear{Benjamini and Hochberg}{1995}]{benjamini1995}
\begin{barticle}[author]
\bauthor{\bsnm{Benjamini},~\bfnm{Yoav}\binits{Y.}} \AND
  \bauthor{\bsnm{Hochberg},~\bfnm{Yosef}\binits{Y.}}
(\byear{1995}).
\btitle{Controlling the false discovery rate: a practical and powerful approach
  to multiple testing}.
\bjournal{Journal of the Royal Statistical Society. Series B (Methodological)}
\bvolume{51}
\bpages{289--300}.
\end{barticle}
\endbibitem

\bibitem[\protect\citeauthoryear{Benjamini and Yekutieli}{2001}]{benjamini2001}
\begin{barticle}[author]
\bauthor{\bsnm{Benjamini},~\bfnm{Yoav}\binits{Y.}} \AND
  \bauthor{\bsnm{Yekutieli},~\bfnm{Daniel}\binits{D.}}
(\byear{2001}).
\btitle{The control of the false discovery rate in multiple testing under
  dependency}.
\bjournal{Annals of statistics}
\bvolume{29}
\bpages{1165--1188}.
\end{barticle}
\endbibitem

\bibitem[\protect\citeauthoryear{Blalock et~al.}{2004}]{blalock2004}
\begin{barticle}[author]
\bauthor{\bsnm{Blalock},~\bfnm{Eric~M}\binits{E.~M.}},
  \bauthor{\bsnm{Geddes},~\bfnm{James~W}\binits{J.~W.}},
  \bauthor{\bsnm{Chen},~\bfnm{Kuey~Chu}\binits{K.~C.}},
  \bauthor{\bsnm{Porter},~\bfnm{Nada~M}\binits{N.~M.}},
  \bauthor{\bsnm{Markesbery},~\bfnm{William~R}\binits{W.~R.}} \AND
  \bauthor{\bsnm{Landfield},~\bfnm{Philip~W}\binits{P.~W.}}
(\byear{2004}).
\btitle{Incipient Alzheimer's disease: microarray correlation analyses reveal
  major transcriptional and tumor suppressor responses}.
\bjournal{Proceedings of the National Academy of Sciences of the United States
  of America}
\bvolume{101}
\bpages{2173--2178}.
\end{barticle}
\endbibitem

\bibitem[\protect\citeauthoryear{Bollen}{1989}]{bollen1989}
\begin{bbook}[author]
\bauthor{\bsnm{Bollen},~\bfnm{Kenneth~A}\binits{K.~A.}}
(\byear{1989}).
\btitle{Structural equations with latent variables.}
\bpublisher{John Wiley \& Sons}.
\end{bbook}
\endbibitem

\bibitem[\protect\citeauthoryear{Brys, Hubert and Struyf}{2004}]{brys2004}
\begin{barticle}[author]
\bauthor{\bsnm{Brys},~\bfnm{G}\binits{G.}},
  \bauthor{\bsnm{Hubert},~\bfnm{Mia}\binits{M.}} \AND
  \bauthor{\bsnm{Struyf},~\bfnm{A}\binits{A.}}
(\byear{2004}).
\btitle{A robust measure of skewness}.
\bjournal{Journal of Computational and Graphical Statistics}
\bvolume{13}
\bpages{996--1017}.
\end{barticle}
\endbibitem

\bibitem[\protect\citeauthoryear{Chandrasekaran, Parrilo and
  Willsky}{2012}]{chandrasekaran2012}
\begin{barticle}[author]
\bauthor{\bsnm{Chandrasekaran},~\bfnm{Venkat}\binits{V.}},
  \bauthor{\bsnm{Parrilo},~\bfnm{Pablo~A.}\binits{P.~A.}} \AND
  \bauthor{\bsnm{Willsky},~\bfnm{Alan~S.}\binits{A.~S.}}
(\byear{2012}).
\btitle{Latent variable graphical model selection via convex optimization}.
\bjournal{Ann. Statist.}
\bvolume{40}
\bpages{1935--1967}.
\bdoi{10.1214/11-AOS949}
\end{barticle}
\endbibitem

\bibitem[\protect\citeauthoryear{Clarke and Hall}{2009}]{clarke2009}
\begin{barticle}[author]
\bauthor{\bsnm{Clarke},~\bfnm{Sandy}\binits{S.}} \AND
  \bauthor{\bsnm{Hall},~\bfnm{Peter}\binits{P.}}
(\byear{2009}).
\btitle{Robustness of multiple testing procedures against dependence}.
\bjournal{The Annals of Statistics}
\bvolume{37}
\bpages{332--358}.
\end{barticle}
\endbibitem

\bibitem[\protect\citeauthoryear{Craig et~al.}{2006}]{craig2006}
\begin{barticle}[author]
\bauthor{\bsnm{Craig},~\bfnm{Andrew}\binits{A.}},
  \bauthor{\bsnm{Cloarec},~\bfnm{Olivier}\binits{O.}},
  \bauthor{\bsnm{Holmes},~\bfnm{Elaine}\binits{E.}},
  \bauthor{\bsnm{Nicholson},~\bfnm{Jeremy~K}\binits{J.~K.}} \AND
  \bauthor{\bsnm{Lindon},~\bfnm{John~C}\binits{J.~C.}}
(\byear{2006}).
\btitle{Scaling and normalization effects in NMR spectroscopic metabonomic data
  sets}.
\bjournal{Analytical Chemistry}
\bvolume{78}
\bpages{2262--2267}.
\end{barticle}
\endbibitem

\bibitem[\protect\citeauthoryear{De~La~Fuente et~al.}{2004}]{de2004}
\begin{barticle}[author]
\bauthor{\bsnm{De~La~Fuente},~\bfnm{Alberto}\binits{A.}},
  \bauthor{\bsnm{Bing},~\bfnm{Nan}\binits{N.}},
  \bauthor{\bsnm{Hoeschele},~\bfnm{Ina}\binits{I.}} \AND
  \bauthor{\bsnm{Mendes},~\bfnm{Pedro}\binits{P.}}
(\byear{2004}).
\btitle{Discovery of meaningful associations in genomic data using partial
  correlation coefficients}.
\bjournal{Bioinformatics}
\bvolume{20}
\bpages{3565--3574}.
\end{barticle}
\endbibitem

\bibitem[\protect\citeauthoryear{Desai and Storey}{2012}]{desai2012cross}
\begin{barticle}[author]
\bauthor{\bsnm{Desai},~\bfnm{Keyur~H}\binits{K.~H.}} \AND
  \bauthor{\bsnm{Storey},~\bfnm{John~D}\binits{J.~D.}}
(\byear{2012}).
\btitle{Cross-dimensional inference of dependent high-dimensional data}.
\bjournal{Journal of the American Statistical Association}
\bvolume{107}
\bpages{135--151}.
\end{barticle}
\endbibitem

\bibitem[\protect\citeauthoryear{Efron}{2007}]{efron2007}
\begin{barticle}[author]
\bauthor{\bsnm{Efron},~\bfnm{Bradley}\binits{B.}}
(\byear{2007}).
\btitle{Correlation and large-scale simultaneous significance testing}.
\bjournal{Journal of the American Statistical Association}
\bvolume{102}
\bpages{93-103}.
\end{barticle}
\endbibitem

\bibitem[\protect\citeauthoryear{Efron}{2010}]{efron2010}
\begin{barticle}[author]
\bauthor{\bsnm{Efron},~\bfnm{Bradley}\binits{B.}}
(\byear{2010}).
\btitle{Correlated z-values and the accuracy of large-scale statistical
  estimates}.
\bjournal{Journal of the American Statistical Association}
\bvolume{105}
\bpages{1042--1055}.
\end{barticle}
\endbibitem

\bibitem[\protect\citeauthoryear{Fan, Han and Gu}{2012}]{fan2012}
\begin{barticle}[author]
\bauthor{\bsnm{Fan},~\bfnm{Jianqing}\binits{J.}},
  \bauthor{\bsnm{Han},~\bfnm{Xu}\binits{X.}} \AND
  \bauthor{\bsnm{Gu},~\bfnm{Weijie}\binits{W.}}
(\byear{2012}).
\btitle{Estimating false discovery proportion under arbitrary covariance
  dependence}.
\bjournal{Journal of the American Statistical Association}
\bvolume{107}
\bpages{1019--1035}.
\end{barticle}
\endbibitem

\bibitem[\protect\citeauthoryear{Fan and Han}{2013}]{fan2013}
\begin{barticle}[author]
\bauthor{\bsnm{Fan},~\bfnm{Jianqing}\binits{J.}} \AND
  \bauthor{\bsnm{Han},~\bfnm{Xu}\binits{X.}}
(\byear{2013}).
\btitle{Estimation of false discovery proportion with unknown dependence}.
\bjournal{arXiv:1305.7007}.
\end{barticle}
\endbibitem

\bibitem[\protect\citeauthoryear{Fare et~al.}{2003}]{fare2003}
\begin{barticle}[author]
\bauthor{\bsnm{Fare},~\bfnm{Thomas~L}\binits{T.~L.}},
  \bauthor{\bsnm{Coffey},~\bfnm{Ernest~M}\binits{E.~M.}},
  \bauthor{\bsnm{Dai},~\bfnm{Hongyue}\binits{H.}},
  \bauthor{\bsnm{He},~\bfnm{Yudong~D}\binits{Y.~D.}},
  \bauthor{\bsnm{Kessler},~\bfnm{Deborah~A}\binits{D.~A.}},
  \bauthor{\bsnm{Kilian},~\bfnm{Kristopher~A}\binits{K.~A.}},
  \bauthor{\bsnm{Koch},~\bfnm{John~E}\binits{J.~E.}},
  \bauthor{\bsnm{LeProust},~\bfnm{Eric}\binits{E.}},
  \bauthor{\bsnm{Marton},~\bfnm{Matthew~J}\binits{M.~J.}},
  \bauthor{\bsnm{Meyer},~\bfnm{Michael~R}\binits{M.~R.}} \betal{et~al.}
(\byear{2003}).
\btitle{Effects of atmospheric ozone on microarray data quality}.
\bjournal{Analytical chemistry}
\bvolume{75}
\bpages{4672--4675}.
\end{barticle}
\endbibitem

\bibitem[\protect\citeauthoryear{Fisher}{1935}]{fisher1935}
\begin{bbook}[author]
\bauthor{\bsnm{Fisher},~\bfnm{Ronald~Aylmer}\binits{R.~A.}}
(\byear{1935}).
\btitle{The design of experiments.}
\bpublisher{Oliver \& Boyd}.
\end{bbook}
\endbibitem

\bibitem[\protect\citeauthoryear{Friguet, Kloareg and
  Causeur}{2009}]{friguet2009}
\begin{barticle}[author]
\bauthor{\bsnm{Friguet},~\bfnm{Chlo{\'e}}\binits{C.}},
  \bauthor{\bsnm{Kloareg},~\bfnm{Maela}\binits{M.}} \AND
  \bauthor{\bsnm{Causeur},~\bfnm{David}\binits{D.}}
(\byear{2009}).
\btitle{A factor model approach to multiple testing under dependence}.
\bjournal{Journal of the American Statistical Association}
\bvolume{104}
\bpages{1406--1415}.
\end{barticle}
\endbibitem

\bibitem[\protect\citeauthoryear{Gagnon-Bartsch, Jacob and
  Speed}{2013}]{gagnon2013}
\begin{btechreport}[author]
\bauthor{\bsnm{Gagnon-Bartsch},~\bfnm{J}\binits{J.}},
  \bauthor{\bsnm{Jacob},~\bfnm{L}\binits{L.}} \AND
  \bauthor{\bsnm{Speed},~\bfnm{TP}\binits{T.}}
(\byear{2013}).
\btitle{Removing unwanted variation from high dimensional data with negative
  controls}
\btype{Technical Report},
\bpublisher{Technical Report 820, Department of Statistics, University of
  California, Berkeley}.
\end{btechreport}
\endbibitem

\bibitem[\protect\citeauthoryear{Gagnon-Bartsch and Speed}{2012}]{gagnon2012}
\begin{barticle}[author]
\bauthor{\bsnm{Gagnon-Bartsch},~\bfnm{Johann~A}\binits{J.~A.}} \AND
  \bauthor{\bsnm{Speed},~\bfnm{Terence~P}\binits{T.~P.}}
(\byear{2012}).
\btitle{Using control genes to correct for unwanted variation in microarray
  data}.
\bjournal{Biostatistics}
\bvolume{13}
\bpages{539--552}.
\end{barticle}
\endbibitem

\bibitem[\protect\citeauthoryear{Gasch et~al.}{2000}]{gasch2000}
\begin{barticle}[author]
\bauthor{\bsnm{Gasch},~\bfnm{Audrey~P}\binits{A.~P.}},
  \bauthor{\bsnm{Spellman},~\bfnm{Paul~T}\binits{P.~T.}},
  \bauthor{\bsnm{Kao},~\bfnm{Camilla~M}\binits{C.~M.}},
  \bauthor{\bsnm{Carmel-Harel},~\bfnm{Orna}\binits{O.}},
  \bauthor{\bsnm{Eisen},~\bfnm{Michael~B}\binits{M.~B.}},
  \bauthor{\bsnm{Storz},~\bfnm{Gisela}\binits{G.}},
  \bauthor{\bsnm{Botstein},~\bfnm{David}\binits{D.}} \AND
  \bauthor{\bsnm{Brown},~\bfnm{Patrick~O}\binits{P.~O.}}
(\byear{2000}).
\btitle{Genomic expression programs in the response of yeast cells to
  environmental changes}.
\bjournal{Molecular biology of the cell}
\bvolume{11}
\bpages{4241--4257}.
\end{barticle}
\endbibitem

\bibitem[\protect\citeauthoryear{Greenland, Robins and
  Pearl}{1999}]{greenland1999}
\begin{barticle}[author]
\bauthor{\bsnm{Greenland},~\bfnm{Sander}\binits{S.}},
  \bauthor{\bsnm{Robins},~\bfnm{James~M}\binits{J.~M.}} \AND
  \bauthor{\bsnm{Pearl},~\bfnm{Judea}\binits{J.}}
(\byear{1999}).
\btitle{Confounding and collapsibility in causal inference}.
\bjournal{Statistical Science}
\bvolume{14}
\bpages{29--46}.
\end{barticle}
\endbibitem

\bibitem[\protect\citeauthoryear{Grzebyk, Wild and
  Chouani{\`e}re}{2004}]{grzebyk2004}
\begin{barticle}[author]
\bauthor{\bsnm{Grzebyk},~\bfnm{Michel}\binits{M.}},
  \bauthor{\bsnm{Wild},~\bfnm{Pascal}\binits{P.}} \AND
  \bauthor{\bsnm{Chouani{\`e}re},~\bfnm{Dominique}\binits{D.}}
(\byear{2004}).
\btitle{On identification of multi-factor models with correlated residuals}.
\bjournal{Biometrika}
\bvolume{91}
\bpages{141--151}.
\end{barticle}
\endbibitem

\bibitem[\protect\citeauthoryear{Irizarry et~al.}{2003}]{irizarry2003}
\begin{barticle}[author]
\bauthor{\bsnm{Irizarry},~\bfnm{Rafael~A}\binits{R.~A.}},
  \bauthor{\bsnm{Hobbs},~\bfnm{Bridget}\binits{B.}},
  \bauthor{\bsnm{Collin},~\bfnm{Francois}\binits{F.}},
  \bauthor{\bsnm{Beazer-Barclay},~\bfnm{Yasmin~D}\binits{Y.~D.}},
  \bauthor{\bsnm{Antonellis},~\bfnm{Kristen~J}\binits{K.~J.}},
  \bauthor{\bsnm{Scherf},~\bfnm{Uwe}\binits{U.}},
  \bauthor{\bsnm{Speed},~\bfnm{Terence~P}\binits{T.~P.}} \betal{et~al.}
(\byear{2003}).
\btitle{Exploration, normalization, and summaries of high density
  oligonucleotide array probe level data}.
\bjournal{Biostatistics}
\bvolume{4}
\bpages{249--264}.
\end{barticle}
\endbibitem

\bibitem[\protect\citeauthoryear{Jin}{2012}]{jin2012comment}
\begin{barticle}[author]
\bauthor{\bsnm{Jin},~\bfnm{Jiashun}\binits{J.}}
(\byear{2012}).
\btitle{Comment}.
\bjournal{Journal of the American Statistical Association}
\bvolume{107}
\bpages{1042--1045}.
\end{barticle}
\endbibitem

\bibitem[\protect\citeauthoryear{Kish}{1959}]{kish1959}
\begin{barticle}[author]
\bauthor{\bsnm{Kish},~\bfnm{Leslie}\binits{L.}}
(\byear{1959}).
\btitle{Some statistical problems in research design}.
\bjournal{American Sociological Review}
\bvolume{24}
\bpages{328--338}.
\end{barticle}
\endbibitem

\bibitem[\protect\citeauthoryear{Korn et~al.}{2004}]{korn2004}
\begin{barticle}[author]
\bauthor{\bsnm{Korn},~\bfnm{Edward~L}\binits{E.~L.}},
  \bauthor{\bsnm{Troendle},~\bfnm{James~F}\binits{J.~F.}},
  \bauthor{\bsnm{McShane},~\bfnm{Lisa~M}\binits{L.~M.}} \AND
  \bauthor{\bsnm{Simon},~\bfnm{Richard}\binits{R.}}
(\byear{2004}).
\btitle{Controlling the number of false discoveries: application to
  high-dimensional genomic data}.
\bjournal{Journal of Statistical Planning and Inference}
\bvolume{124}
\bpages{379--398}.
\end{barticle}
\endbibitem

\bibitem[\protect\citeauthoryear{Kuroki and Pearl}{2014}]{kuroki2014}
\begin{barticle}[author]
\bauthor{\bsnm{Kuroki},~\bfnm{Manabu}\binits{M.}} \AND
  \bauthor{\bsnm{Pearl},~\bfnm{Judea}\binits{J.}}
(\byear{2014}).
\btitle{Measurement Bias and Effect Restoration in Causal Inference}.
\bjournal{Biometrika}
\bvolume{101}
\bpages{423--437}.
\end{barticle}
\endbibitem

\bibitem[\protect\citeauthoryear{Lan and Du}{2014}]{lan2014}
\begin{barticle}[author]
\bauthor{\bsnm{Lan},~\bfnm{Wei}\binits{W.}} \AND
  \bauthor{\bsnm{Du},~\bfnm{Lilun}\binits{L.}}
(\byear{2014}).
\btitle{A Factor-Adjusted Multiple Testing Procedure with Application to Mutual
  Fund Selection}.
\bjournal{arXiv:1407.5515}.
\end{barticle}
\endbibitem

\bibitem[\protect\citeauthoryear{Lazar et~al.}{2013}]{lazar2013}
\begin{barticle}[author]
\bauthor{\bsnm{Lazar},~\bfnm{Cosmin}\binits{C.}},
  \bauthor{\bsnm{Meganck},~\bfnm{Stijn}\binits{S.}},
  \bauthor{\bsnm{Taminau},~\bfnm{Jonatan}\binits{J.}},
  \bauthor{\bsnm{Steenhoff},~\bfnm{David}\binits{D.}},
  \bauthor{\bsnm{Coletta},~\bfnm{Alain}\binits{A.}},
  \bauthor{\bsnm{Molter},~\bfnm{Colin}\binits{C.}},
  \bauthor{\bsnm{Weiss-Sol{\'\i}s},~\bfnm{David~Y}\binits{D.~Y.}},
  \bauthor{\bsnm{Duque},~\bfnm{Robin}\binits{R.}},
  \bauthor{\bsnm{Bersini},~\bfnm{Hugues}\binits{H.}} \AND
  \bauthor{\bsnm{Now{\'e}},~\bfnm{Ann}\binits{A.}}
(\byear{2013}).
\btitle{Batch effect removal methods for microarray gene expression data
  integration: a survey}.
\bjournal{Briefings in bioinformatics}
\bvolume{14}
\bpages{469--490}.
\end{barticle}
\endbibitem

\bibitem[\protect\citeauthoryear{Leek and Storey}{2007}]{leek2007}
\begin{barticle}[author]
\bauthor{\bsnm{Leek},~\bfnm{Jeffrey~T}\binits{J.~T.}} \AND
  \bauthor{\bsnm{Storey},~\bfnm{John~D}\binits{J.~D.}}
(\byear{2007}).
\btitle{Capturing heterogeneity in gene expression studies by surrogate
  variable analysis}.
\bjournal{PLoS genetics}
\bvolume{3}
\bpages{1724-1735}.
\end{barticle}
\endbibitem

\bibitem[\protect\citeauthoryear{Leek and Storey}{2008}]{leek2008}
\begin{barticle}[author]
\bauthor{\bsnm{Leek},~\bfnm{Jeffrey~T}\binits{J.~T.}} \AND
  \bauthor{\bsnm{Storey},~\bfnm{John~D}\binits{J.~D.}}
(\byear{2008}).
\btitle{A general framework for multiple testing dependence}.
\bjournal{Proceedings of the National Academy of Sciences}
\bvolume{105}
\bpages{18718--18723}.
\end{barticle}
\endbibitem

\bibitem[\protect\citeauthoryear{Leek et~al.}{2010}]{leek2010}
\begin{barticle}[author]
\bauthor{\bsnm{Leek},~\bfnm{Jeffrey~T}\binits{J.~T.}},
  \bauthor{\bsnm{Scharpf},~\bfnm{Robert~B}\binits{R.~B.}},
  \bauthor{\bsnm{Bravo},~\bfnm{H{\'e}ctor~Corrada}\binits{H.~C.}},
  \bauthor{\bsnm{Simcha},~\bfnm{David}\binits{D.}},
  \bauthor{\bsnm{Langmead},~\bfnm{Benjamin}\binits{B.}},
  \bauthor{\bsnm{Johnson},~\bfnm{W~Evan}\binits{W.~E.}},
  \bauthor{\bsnm{Geman},~\bfnm{Donald}\binits{D.}},
  \bauthor{\bsnm{Baggerly},~\bfnm{Keith}\binits{K.}} \AND
  \bauthor{\bsnm{Irizarry},~\bfnm{Rafael~A}\binits{R.~A.}}
(\byear{2010}).
\btitle{Tackling the widespread and critical impact of batch effects in
  high-throughput data}.
\bjournal{Nature Reviews Genetics}
\bvolume{11}
\bpages{733--739}.
\end{barticle}
\endbibitem

\bibitem[\protect\citeauthoryear{Li and Zhong}{2014}]{li2014rate}
\begin{barticle}[author]
\bauthor{\bsnm{Li},~\bfnm{Jun}\binits{J.}} \AND
  \bauthor{\bsnm{Zhong},~\bfnm{Ping-Shou}\binits{P.-S.}}
(\byear{2014}).
\btitle{A rate optimal procedure for sparse signal recovery under dependence}.
\bjournal{arXiv preprint arXiv:1410.2839}.
\end{barticle}
\endbibitem

\bibitem[\protect\citeauthoryear{Lin et~al.}{2006}]{lin2006}
\begin{barticle}[author]
\bauthor{\bsnm{Lin},~\bfnm{Daniel~W}\binits{D.~W.}},
  \bauthor{\bsnm{Coleman},~\bfnm{Ilsa~M}\binits{I.~M.}},
  \bauthor{\bsnm{Hawley},~\bfnm{Sarah}\binits{S.}},
  \bauthor{\bsnm{Huang},~\bfnm{Chung~Y}\binits{C.~Y.}},
  \bauthor{\bsnm{Dumpit},~\bfnm{Ruth}\binits{R.}},
  \bauthor{\bsnm{Gifford},~\bfnm{David}\binits{D.}},
  \bauthor{\bsnm{Kezele},~\bfnm{Philip}\binits{P.}},
  \bauthor{\bsnm{Hung},~\bfnm{Hau}\binits{H.}},
  \bauthor{\bsnm{Knudsen},~\bfnm{Beatrice~S}\binits{B.~S.}},
  \bauthor{\bsnm{Kristal},~\bfnm{Alan~R}\binits{A.~R.}} \betal{et~al.}
(\byear{2006}).
\btitle{Influence of surgical manipulation on prostate gene expression:
  implications for molecular correlates of treatment effects and disease
  prognosis}.
\bjournal{Journal of clinical oncology}
\bvolume{24}
\bpages{3763--3770}.
\end{barticle}
\endbibitem

\bibitem[\protect\citeauthoryear{Maronna, Martin and Yohai}{2006}]{maronna2006}
\begin{bbook}[author]
\bauthor{\bsnm{Maronna},~\bfnm{Ricardo~A.}\binits{R.~A.}},
  \bauthor{\bsnm{Martin},~\bfnm{Douglas~R.}\binits{D.~R.}} \AND
  \bauthor{\bsnm{Yohai},~\bfnm{Victor~J.}\binits{V.~J.}}
(\byear{2006}).
\btitle{Robust statistics: Theory and Methods}.
\bpublisher{John Wiley \& Sons, Chichester}.
\end{bbook}
\endbibitem

\bibitem[\protect\citeauthoryear{Onatski}{2010}]{onatski2010determining}
\begin{barticle}[author]
\bauthor{\bsnm{Onatski},~\bfnm{Alexei}\binits{A.}}
(\byear{2010}).
\btitle{Determining the number of factors from empirical distribution of
  eigenvalues}.
\bjournal{The Review of Economics and Statistics}
\bvolume{92}
\bpages{1004--1016}.
\end{barticle}
\endbibitem

\bibitem[\protect\citeauthoryear{Owen}{2005}]{owen2005}
\begin{barticle}[author]
\bauthor{\bsnm{Owen},~\bfnm{Art~B}\binits{A.~B.}}
(\byear{2005}).
\btitle{Variance of the number of false discoveries}.
\bjournal{Journal of the Royal Statistical Society: Series B (Statistical
  Methodology)}
\bvolume{67}
\bpages{411--426}.
\end{barticle}
\endbibitem

\bibitem[\protect\citeauthoryear{Owen and Wang}{2016}]{owen2015}
\begin{barticle}[author]
\bauthor{\bsnm{Owen},~\bfnm{Art~B}\binits{A.~B.}} \AND
  \bauthor{\bsnm{Wang},~\bfnm{Jingshu}\binits{J.}}
(\byear{2016}).
\btitle{Bi-cross-validation for factor analysis}.
\bjournal{Statistical Science}
\bpages{(to appear)}.
\end{barticle}
\endbibitem

\bibitem[\protect\citeauthoryear{Pearl}{2009}]{pearl2009}
\begin{bbook}[author]
\bauthor{\bsnm{Pearl},~\bfnm{Judea}\binits{J.}}
(\byear{2009}).
\btitle{Causality: models, reasoning and inference}.
\bpublisher{Cambridge Univ Press}.
\end{bbook}
\endbibitem

\bibitem[\protect\citeauthoryear{Perry and Pillai}{2013}]{perry2013degrees}
\begin{barticle}[author]
\bauthor{\bsnm{Perry},~\bfnm{P.~O.}\binits{P.~O.}} \AND
  \bauthor{\bsnm{Pillai},~\bfnm{N.~S.}\binits{N.~S.}}
(\byear{2013}).
\btitle{Degrees of freedom for combining regression with factor analysis}.
\bjournal{arXiv preprint arXiv:1310.7269}.
\end{barticle}
\endbibitem

\bibitem[\protect\citeauthoryear{Pesaran}{2004}]{pesaran2004}
\begin{btechreport}[author]
\bauthor{\bsnm{Pesaran},~\bfnm{M.~H.}\binits{M.~H.}}
(\byear{2004}).
\btitle{{‘General Diagnostic Tests for Cross Section Dependence in
  Panels’}}
\btype{Cambridge Working Papers in Economics} No. \bnumber{0435},
\bpublisher{Faculty of Economics, University of Cambridge}.
\end{btechreport}
\endbibitem

\bibitem[\protect\citeauthoryear{Price et~al.}{2006}]{price2006}
\begin{barticle}[author]
\bauthor{\bsnm{Price},~\bfnm{Alkes~L}\binits{A.~L.}},
  \bauthor{\bsnm{Patterson},~\bfnm{Nick~J}\binits{N.~J.}},
  \bauthor{\bsnm{Plenge},~\bfnm{Robert~M}\binits{R.~M.}},
  \bauthor{\bsnm{Weinblatt},~\bfnm{Michael~E}\binits{M.~E.}},
  \bauthor{\bsnm{Shadick},~\bfnm{Nancy~A}\binits{N.~A.}} \AND
  \bauthor{\bsnm{Reich},~\bfnm{David}\binits{D.}}
(\byear{2006}).
\btitle{Principal components analysis corrects for stratification in
  genome-wide association studies}.
\bjournal{Nature genetics}
\bvolume{38}
\bpages{904--909}.
\end{barticle}
\endbibitem

\bibitem[\protect\citeauthoryear{Ransohoff}{2005}]{ransohoff2005}
\begin{barticle}[author]
\bauthor{\bsnm{Ransohoff},~\bfnm{David~F}\binits{D.~F.}}
(\byear{2005}).
\btitle{Bias as a threat to the validity of cancer molecular-marker research}.
\bjournal{Nature Reviews Cancer}
\bvolume{5}
\bpages{142--149}.
\end{barticle}
\endbibitem

\bibitem[\protect\citeauthoryear{Rhodes and Chinnaiyan}{2005}]{rhodes2005}
\begin{barticle}[author]
\bauthor{\bsnm{Rhodes},~\bfnm{Daniel~R}\binits{D.~R.}} \AND
  \bauthor{\bsnm{Chinnaiyan},~\bfnm{Arul~M}\binits{A.~M.}}
(\byear{2005}).
\btitle{Integrative analysis of the cancer transcriptome}.
\bjournal{Nature genetics}
\bvolume{37}
\bpages{S31--S37}.
\end{barticle}
\endbibitem

\bibitem[\protect\citeauthoryear{Schwartzman}{2010}]{schwartzman2010comment}
\begin{barticle}[author]
\bauthor{\bsnm{Schwartzman},~\bfnm{Armin}\binits{A.}}
(\byear{2010}).
\btitle{Comment}.
\bjournal{Journal of the American Statistical Association}
\bvolume{105}
\bpages{1059--1063}.
\end{barticle}
\endbibitem

\bibitem[\protect\citeauthoryear{Schwartzman, Dougherty and
  Taylor}{2008}]{schwartzman2008}
\begin{barticle}[author]
\bauthor{\bsnm{Schwartzman},~\bfnm{Armin}\binits{A.}},
  \bauthor{\bsnm{Dougherty},~\bfnm{Robert~F}\binits{R.~F.}} \AND
  \bauthor{\bsnm{Taylor},~\bfnm{Jonathan~E}\binits{J.~E.}}
(\byear{2008}).
\btitle{False discovery rate analysis of brain diffusion direction maps}.
\bjournal{The Annals of Applied Statistics}
\bvolume{2}
\bpages{153--175}.
\end{barticle}
\endbibitem

\bibitem[\protect\citeauthoryear{She and Owen}{2011}]{she2011}
\begin{barticle}[author]
\bauthor{\bsnm{She},~\bfnm{Yiyuan}\binits{Y.}} \AND
  \bauthor{\bsnm{Owen},~\bfnm{Art~B}\binits{A.~B.}}
(\byear{2011}).
\btitle{Outlier detection using nonconvex penalized regression}.
\bjournal{Journal of the American Statistical Association}
\bvolume{106}
\bpages{626-639}.
\end{barticle}
\endbibitem

\bibitem[\protect\citeauthoryear{Singh et~al.}{2011}]{singh2011}
\begin{barticle}[author]
\bauthor{\bsnm{Singh},~\bfnm{Dave}\binits{D.}},
  \bauthor{\bsnm{Fox},~\bfnm{Steven~M}\binits{S.~M.}},
  \bauthor{\bsnm{Tal-Singer},~\bfnm{Ruth}\binits{R.}},
  \bauthor{\bsnm{Plumb},~\bfnm{Jonathan}\binits{J.}},
  \bauthor{\bsnm{Bates},~\bfnm{Stewart}\binits{S.}},
  \bauthor{\bsnm{Broad},~\bfnm{Peter}\binits{P.}},
  \bauthor{\bsnm{Riley},~\bfnm{John~H}\binits{J.~H.}} \AND
  \bauthor{\bsnm{Celli},~\bfnm{Bartolome}\binits{B.}}
(\byear{2011}).
\btitle{Induced sputum genes associated with spirometric and radiological
  disease severity in {COPD} ex-smokers}.
\bjournal{Thorax}
\bvolume{66}
\bpages{489--495}.
\end{barticle}
\endbibitem

\bibitem[\protect\citeauthoryear{Storey, Taylor and
  Siegmund}{2004}]{storey2004}
\begin{barticle}[author]
\bauthor{\bsnm{Storey},~\bfnm{John~D}\binits{J.~D.}},
  \bauthor{\bsnm{Taylor},~\bfnm{Jonathan~E}\binits{J.~E.}} \AND
  \bauthor{\bsnm{Siegmund},~\bfnm{David}\binits{D.}}
(\byear{2004}).
\btitle{Strong control, conservative point estimation and simultaneous
  conservative consistency of false discovery rates: a unified approach}.
\bjournal{Journal of the Royal Statistical Society: Series B (Statistical
  Methodology)}
\bvolume{66}
\bpages{187--205}.
\end{barticle}
\endbibitem

\bibitem[\protect\citeauthoryear{Sun}{2011}]{sun2011}
\begin{bphdthesis}[author]
\bauthor{\bsnm{Sun},~\bfnm{Yunting}\binits{Y.}}
(\byear{2011}).
\btitle{On latent systemic effects in multiple hypotheses}
\btype{PhD thesis},
\bpublisher{Stanford University}.
\end{bphdthesis}
\endbibitem

\bibitem[\protect\citeauthoryear{Sun and Cai}{2009}]{sun2009}
\begin{barticle}[author]
\bauthor{\bsnm{Sun},~\bfnm{Wenguang}\binits{W.}} \AND
  \bauthor{\bsnm{Cai},~\bfnm{Tony}\binits{T.}}
(\byear{2009}).
\btitle{Large-scale multiple testing under dependence}.
\bjournal{Journal of the Royal Statistical Society: Series B (Statistical
  Methodology)}
\bvolume{71}
\bpages{393--424}.
\end{barticle}
\endbibitem

\bibitem[\protect\citeauthoryear{Sun, Zhang and Owen}{2012}]{sun2012}
\begin{barticle}[author]
\bauthor{\bsnm{Sun},~\bfnm{Yunting}\binits{Y.}},
  \bauthor{\bsnm{Zhang},~\bfnm{Nancy~R}\binits{N.~R.}} \AND
  \bauthor{\bsnm{Owen},~\bfnm{Art~B}\binits{A.~B.}}
(\byear{2012}).
\btitle{Multiple hypothesis testing adjusted for latent variables, with an
  application to the agemap gene expression data}.
\bjournal{The Annals of Applied Statistics}
\bvolume{6}
\bpages{1664--1688}.
\end{barticle}
\endbibitem

\bibitem[\protect\citeauthoryear{Tusher, Tibshirani and Chu}{2001}]{tusher2001}
\begin{barticle}[author]
\bauthor{\bsnm{Tusher},~\bfnm{Virginia~Goss}\binits{V.~G.}},
  \bauthor{\bsnm{Tibshirani},~\bfnm{Robert}\binits{R.}} \AND
  \bauthor{\bsnm{Chu},~\bfnm{Gilbert}\binits{G.}}
(\byear{2001}).
\btitle{Significance analysis of microarrays applied to the ionizing radiation
  response}.
\bjournal{Proceedings of the National Academy of Sciences}
\bvolume{98}
\bpages{5116--5121}.
\end{barticle}
\endbibitem

\bibitem[\protect\citeauthoryear{Vawter et~al.}{2004}]{vawter2004}
\begin{barticle}[author]
\bauthor{\bsnm{Vawter},~\bfnm{Marquis~P}\binits{M.~P.}},
  \bauthor{\bsnm{Evans},~\bfnm{Simon}\binits{S.}},
  \bauthor{\bsnm{Choudary},~\bfnm{Prabhakara}\binits{P.}},
  \bauthor{\bsnm{Tomita},~\bfnm{Hiroaki}\binits{H.}},
  \bauthor{\bsnm{Meador-Woodruff},~\bfnm{Jim}\binits{J.}},
  \bauthor{\bsnm{Molnar},~\bfnm{Margherita}\binits{M.}},
  \bauthor{\bsnm{Li},~\bfnm{Jun}\binits{J.}},
  \bauthor{\bsnm{Lopez},~\bfnm{Juan~F}\binits{J.~F.}},
  \bauthor{\bsnm{Myers},~\bfnm{Rick}\binits{R.}},
  \bauthor{\bsnm{Cox},~\bfnm{David}\binits{D.}} \betal{et~al.}
(\byear{2004}).
\btitle{Gender-specific gene expression in post-mortem human brain:
  localization to sex chromosomes}.
\bjournal{Neuropsychopharmacology}
\bvolume{29}
\bpages{373-384}.
\end{barticle}
\endbibitem

\bibitem[\protect\citeauthoryear{Wang, Cui and Li}{2015}]{wang2015}
\begin{barticle}[author]
\bauthor{\bsnm{Wang},~\bfnm{Shaoping}\binits{S.}},
  \bauthor{\bsnm{Cui},~\bfnm{Guowei}\binits{G.}} \AND
  \bauthor{\bsnm{Li},~\bfnm{Kunpeng}\binits{K.}}
(\byear{2015}).
\btitle{Factor-augmented regression models with structural change}.
\bjournal{Economics Letters}
\bvolume{130}
\bpages{124--127}.
\end{barticle}
\endbibitem

\bibitem[\protect\citeauthoryear{Wang et~al.}{2015}]{wang2015supplement}
\begin{barticle}[author]
\bauthor{\bsnm{Wang},~\bfnm{Jingshu}\binits{J.}},
  \bauthor{\bsnm{Zhao},~\bfnm{Qingyuan}\binits{Q.}},
  \bauthor{\bsnm{Hastie},~\bfnm{Trevor}\binits{T.}} \AND
  \bauthor{\bsnm{Owen},~\bfnm{Art~B.}\binits{A.~B.}}
(\byear{2015}).
\btitle{Supplement to "Confounder Adjustment in Multiple Hypothesis Testing"}.
\end{barticle}
\endbibitem

\bibitem[\protect\citeauthoryear{Yohai}{1987}]{yohai1987high}
\begin{barticle}[author]
\bauthor{\bsnm{Yohai},~\bfnm{Victor~J}\binits{V.~J.}}
(\byear{1987}).
\btitle{High breakdown-point and high efficiency robust estimates for
  regression}.
\bjournal{The Annals of Statistics}
\bpages{642--656}.
\end{barticle}
\endbibitem

\end{thebibliography}


\begin{thebibliography}{2}

\bibitem[\protect\citeauthoryear{Bai and Li}{2012a}]{bai2012}
\begin{barticle}[author]
\bauthor{\bsnm{Bai},~\bfnm{Jushan}\binits{J.}} \AND
  \bauthor{\bsnm{Li},~\bfnm{Kunpeng}\binits{K.}}
(\byear{2012}a).
\btitle{Statistical analysis of factor models of high dimension}.
\bjournal{The Annals of Statistics}
\bvolume{40}
\bpages{436--465}.
\end{barticle}
\endbibitem

\bibitem[\protect\citeauthoryear{Bai and Li}{2012b}]{bai2012supplement}
\begin{barticle}[author]
\bauthor{\bsnm{Bai},~\bfnm{Jushan}\binits{J.}} \AND
  \bauthor{\bsnm{Li},~\bfnm{Kunpeng}\binits{K.}}
(\byear{2012}b).
\btitle{Supplement to "Statistical analysis of factor models of high
  dimension."}.
\bdoi{10.1214/11-AOS966SUPP}
\end{barticle}
\endbibitem

\end{thebibliography}


\begin{thebibliography}{0}

\end{thebibliography}

\end{document}